\documentclass[12pt]{article}
\usepackage{geometry}
\usepackage[normalem]{ulem}
\usepackage[onehalfspacing]{setspace}
\usepackage{amsfonts}
\usepackage{amsmath, amssymb}
\usepackage{amsthm}
\usepackage{bbm}
\usepackage{dsfont}
\usepackage{graphicx}
\usepackage{latexsym}
\usepackage[font={small}]{caption}
\usepackage{subcaption}
\usepackage{geometry}
\usepackage[onehalfspacing]{setspace}
\usepackage{color}
\usepackage{array}
\usepackage{rotating}
\usepackage{float}
\usepackage{multirow}
\usepackage{amsthm}
\usepackage{amsmath}
\usepackage{amssymb}
\usepackage{amsfonts}
\usepackage{mathrsfs}
\usepackage{graphicx}
\usepackage{soul}
\usepackage[multiple]{footmisc}
\usepackage{hyperref}
\usepackage{fnpct}
\usepackage{xcolor}

\usepackage{esint}
\usepackage{sgame}
\usepackage[multiple]{footmisc}
\usepackage{natbib}
\usepackage[utf8]{inputenc}
\usepackage{apptools}

\usepackage[title]{appendix}
\usepackage{etoolbox}
\patchcmd{\appendices}{\quad}{: }{}{}

\usepackage{color,hyperref}
\definecolor{darkblue}{rgb}{0.0,0.0,0.7}
\hypersetup{colorlinks,breaklinks,linkcolor=darkblue,urlcolor=darkblue,anchorcolor=darkblue,citecolor=darkblue}

\newtheorem{theorem}{Theorem}

\newtheorem{assu}{Assumption}

\newtheorem{definition}{Definition}

\newtheorem{lemma}{Lemma}

\newtheorem{proposition}{Proposition}
\newtheorem{remark}{Remark}

\newcommand{\diag}{\mathop{\mathrm{diag}}}

\def\QQ{{\rlap {\raise 0.4ex \hbox{$\scriptscriptstyle |$}}\hskip -0.2em Q}}

\def\11{{ I\!\!1}}

\definecolor{darkgreen}{rgb}{0,0.5,0}

\allowdisplaybreaks

\begin{document}

\title{Signaling with Private Monitoring\footnote{Cisternas: MIT Sloan School of Management, 100 Main St., Cambridge, MA 02142, \texttt{gcistern@mit.edu}. Kolb: Indiana University Kelley School of Business, 1309 E. Tenth St., Bloomington, IN 47405 \texttt{kolba@indiana.edu}. We thank Alessandro Bonatti, Isa Chavez, Wouter Dessein, Robert Gibbons, Marina Halac, Stephen Morris, Alessandro Pavan, Andy Skrzypacz, Bruno Strulovici, and Vish Viswanathan for useful conversations.}} 
\author{Gonzalo Cisternas and Aaron Kolb}
\date{May 8, 2020}
\maketitle

\vspace{-.4in}
\begin{abstract}

We study dynamic signaling when the informed party does not observe the signals generated by her actions. A long-run player signals her type continuously over time to a myopic second player who privately monitors her behavior; in turn, the myopic player transmits his private inferences back through an imperfect public signal of his actions. Preferences are linear-quadratic and the information structure is Gaussian. We construct linear Markov equilibria using belief states up to the long-run player's \emph{second-order belief}. Because of the private monitoring, this state is an explicit function of the long-run player's past play.  A novel separation effect then emerges through this second-order belief channel, altering the traditional signaling that arises when beliefs are public. Applications to models of leadership, reputation, and trading are examined.\\



\end{abstract}

\vspace{-0.4in}
\section{Introduction\label{sec:intro}}

The general interest in \emph{signaling}---i.e., information transmission through costly actions---is reflected in its influence in virtually all subfields across economics. Despite this breadth, the great majority of signaling games share a key commonality:  the ``sender'' knows the belief of the ``receiver'' about the sender's type at the moment of action. While this public nature of a receiver's belief can be a sensible approximation in some settings, it is far less appropriate in others, such as when imperfect private signals of behavior are at play: employers subjectively assessing their workers' performances \citep{levin2003relational}; traders handling others' orders  \citep{yang2019back}; or data brokers collecting data about consumers \citep{bonatti2019consumer}. There, the beliefs of employers, financial intermediaries, or data brokers over variables such as a worker's ability, an asset's value, or a consumer's preferences, are \emph{private}.

Allowing for private monitoring of an informed player's actions is an important agenda, as it can open the way for a new set of applied-theory questions to be analyzed. How do leaders gradually influence their followers when they do not know how their actions have been interpreted? Can career-concerned agents benefit by not being able to observe the signals generated by their actions when attempting to manage their reputations? How is trading behavior affected by the possibility of hidden leakages to other traders? While clearly realistic and relevant, these questions nonetheless present substantial challenges. First, higher-order beliefs can arise: in most settings, the senders involved will have to form a nontrivial belief about their receivers' beliefs. Second, such settings can be inherently asymmetric: when facing a sender of a \emph{fixed} type, the receiver develops \emph{evolving} private information in the form of a belief. Third, most analyses will be nonstationary due to ongoing learning effects.

In this paper, we introduce a class of linear-quadratic-Gaussian games of incomplete information and private monitoring in which these questions and challenges can be addressed. A long-run player (she) and a myopic counterpart (he), both with linear-quadratic preferences, interact over a finite horizon. The long-run player has a normally distributed type. Our key innovation is to allow the myopic player to privately observe a noisy signal of the long-run player's action; in turn we let the long-run player receive feedback about the myopic player's inferences via an imperfect public signal of the latter's behavior.  The shocks in both signals are additive and Brownian. Using continuous-time methods, we construct linear Markov equilibria (LMEs) in which the players' beliefs are the relevant states.\footnote{The myopic ``receiver'' assumption is convenient for focusing exclusively on how the long-run player's signaling motives respond to the introduction of higher-order uncertainty, but our construction, methods and main findings remain valid beyond this case. We discuss this and other assumptions in the conclusion.}

\vspace{-0.1in}
\paragraph{Equilibrium construction and signaling.} It is well known that the construction of nontrivial equilibria in games of private monitoring can be a daunting task. In fact, to estimate rivals' continuation behavior under any strategy, players  usually have to make an inference about their opponents' private histories. Not knowing what their rivals have seen, the players will then rely on their past play, but this implies that the players' inferences will vary with their \emph{own} private histories. Thus, (i) probability distributions over histories must be computed, and (ii) the continuation games at off- versus on-path histories may differ.

With incomplete information, one expects this statistical inference problem to become one of the estimation of belief states that summarize the payoff-relevant aspects of the players' private histories---our approach offers a parsimonious treatment of this issue. The quadratic preferences permit our players to employ strategies that are linear in their posterior beliefs' means (henceforth, beliefs). Conjecturing such linear strategies, learning is (conditionally) Gaussian: the myopic player's belief is linear in the history of his private signals, and the long-run player's \emph{second-order belief}---her belief about the myopic player's private belief---is linear in the histories of the public signal and her \emph{past play}. The estimation of histories described in (i) is thus simplified by the fact that these are aggregated linearly.

Critically, the long-run player's second-order belief is also private, as her actions depend on her type; the myopic player must therefore forecast this state. The problem of the state space expanding is then circumvented by a key \textit{representation} of the (candidate, on path) second-order belief in terms of the long-run player's type and the belief about it based exclusively on the public signal (Lemma \ref{lem:BeliefDecomp}). Thus, performing equilibrium analysis requires a nontrivial second-order belief that is spanned by the rest of the states along the path of play, in a reflection of how the game's structure changes after deviations, as noted in (ii). With Markov states as sufficient statistics, we can write the long-run player's best-response problem as one of stochastic control and use dynamic programming to find LMEs.\footnote{The presence of this public belief creates signal-jamming motives for the long-run player.}

The long-run player controls her own second-order belief, in a generalization of the traditional control of a public belief under imperfect public monitoring. But since this state is now an explicit function of past play, private monitoring has novel implications for signaling---our representation result is again key. Specifically, because different types behave differently in equilibrium, their different past behavior leads them to expect their ``receivers'' to hold different beliefs. In other words, the perception of different continuation games---as measured by the value of the second-order belief in the representation---opens an additional channel for separation. We refer to this as the \emph{history-inference effect} on signaling. The potential amplitude of this effect is largest when the public signal is pure noise (``no feedback''), and thus the reliance on past play is strongest; conversely, it disappears when beliefs are public.

From a positive standpoint, the relevance of this effect depends on the plausibility of individuals relying on their \emph{past} behavior to forecast what others \emph{currently} know. Crucially, this notion strongly resonates with reality, such as when leaders reflect on their past behavior when assessing organizations' understanding of the leadership's long-term goals, when politicians gauge their reputations, or when traders estimate how much of their private information has been learned by others. We are not aware of an existing framework where the signaling implications of this natural use of past behavior can be studied. Our approach, which ultimately exploits the use of ordinary differential equations (ODEs), offers a venue.


\vspace{-0.1in}
\paragraph{Applications.} 

To leverage the flexibility of the model, we examine one instance of our baseline specification and two based on extensions of it. Our aim is to show how the precisions of the signals involved shape outcomes via the extent of higher-order uncertainty created.

In our leading application (Section \ref{sec:leading_by_example}), we examine the history-inference effect in a coordination game inspired by the linear-quadratic team theory of \cite{marschak1972team}---this framework, along with its extensions allowing for misaligned preferences, has become the main laboratory for studying the impact of information structures on organizations.

In the setting examined, the performance of a team composed of a leader and a follower increases with the proximity of its members' actions and of the leader's action to the state of the world. The leader shares the team's payoffs, while the follower attempts to match the leader's action at all times. Recognizing the prevalence of information frictions within organizations, we assume that, while the leader knows the state of the world, the follower's learning about it is only gradual and private, albeit influenced by the leader's behavior. In this context, the absence of feedback, via the history-inference effect it creates, can lead to more information being transmitted relative to the case in which the follower's belief is public. However, the team's performance is lower. Thus, organizations with a better understanding of the economic environment can  \emph{underperform} their less informed counterparts.


Uncertainty about others' beliefs is also natural in reputational settings. In Section \ref{subsec:political}, we examine a model of horizontal reputation based on an extension that allows for \emph{terminal} payoffs: the long-run player suffers a terminal quadratic loss that increases in the distance between the myopic player's belief and the type's prior (e.g., a politician facing reelection who desires a reputation for neutrality). In such a context, we show that not directly observing her reputation can benefit the long-run player, despite the negative direct effect of the increased uncertainty over her concave objective. Indeed, since higher types take higher actions due to their higher biases, those types must offset higher beliefs to appear unbiased; the history-inference effect then reduces the informativeness of the long-run player's action, making beliefs less sensitive to new information, a strategic effect that can dominate.

Finally, in Section \ref{subsec:insidertrading}, we exploit the presence of the public belief state in a \emph{linear} trading model in which an informed trader faces both a myopic trader who privately monitors her and a competitive market maker who only observes the public total order flow. In this context, we show that there is no linear Markov equilibrium for any degree of noise of the private signal. Intuitively, the myopic player introduces momentum into the price, as the information he obtains is now distributed to the market maker through all future order flows. This causes prices to move against the insider and creates urgency---with an infinite number of opportunities to trade, the insider trades away all information in the first instant. 

\vspace{-0.1in}
\paragraph{Existence of LME and technical contribution.} 

The bulk of our analysis unfolds in Sections \ref{sec:model} and \ref{sec:eqmanalysis}, where we introduce the general model and lay out the methodological framework. A distinctive feature there is that the environment is  \emph{asymmetric}, both in terms of the players' preferences and their private information (a fixed state versus a changing one). In particular, the players can signal at substantially different rates, which is in stark contrast to the existing literature on symmetric multisided learning. With different rates of learning, however, the equilibrium analysis can become severely complicated.

Specifically, our belief states depend on both the myopic player's posterior variance, which determines the sensitivity of the myopic player's belief, and the weight attached to the long-run player's type in the representation result, which shapes the history-inference effect and is linked to the long-run player's learning. Both functions are deterministic due to the Gaussian structure. Using dynamic programming, one can then show that the problem of the existence of an LME reduces to a \emph{boundary value problem} (BVP) including ODEs for the two aforementioned functions of time and for the weights in the long-run player's linear Markov strategy. The two \emph{learning} ODEs endow the BVP with exogenous initial conditions, while the rest carry terminal conditions arising from myopic play at the end of the game.

With multiple ODEs in both directions, establishing the existence of a solution to such a BVP is a challenging ``shooting'' problem: not only must solutions to all individual ODEs exist, but they must land at specific (potentially endogenous) values. To address this complexity, we distinguish between two types of environments. In a \emph{private-value} setting, the myopic player's best response is only a function of his expectation of the long-run player's action, i.e., it does not depend directly on his expectation of the type. In this context, there is enough (strategic) symmetry that a one-to-one mapping emerges between the solutions to the learning ODEs, which renders the shooting problem unidimensional (Lemma \ref{lem:GammaChiRelationship}). Via traditional continuity arguments, we can guarantee the existence of an LME in the leadership model of Section \ref{sec:leading_by_example} when the public signal is of intermediate quality, for a horizon length that is decreasing in the prior variance of the state of the world (Theorem \ref{thm:LeadingInteriorExistence}).

In \emph{common-value} settings, the multidimensionality issue must be confronted. Building on the literature on BVPs with \emph{intertemporal linear} constraints \citep{keller1968numerical}, we can establish the existence of LME for our BVP that carries \emph{intratemporal nonlinear} (terminal) constraints. Specifically, the multidimensional shooting problem can be formulated as one of finding a  fixed point for a suitable function derived from the BVP, a problem that we tackle for a variation of the leadership model in which the follower directly cares about the state of the world (Theorem \ref{thm:Leading_LME_BVP}). Critically, this approach is general: we show how to apply it to the whole class of games under study, and more generally, it offers a promising venue for examining behavior in other settings exhibiting incomplete information and asymmetries.

\vspace{-0.1in}
\paragraph{Related Literature.} A long literature on multisided private monitoring has developed in repeated games with \emph{complete} information, where the issue of inferences of private histories has been handled very differently relative to us. Closest in spirit is \cite{phelan2012beliefs}, where such inferences are coarsened into beliefs over a finite set of states; instead, our players' states take infinitely many values and completely determine their beliefs about the other player's state. Other approaches include \cite{mailath2002repeated}, examining equilibria that condition on finite histories when monitoring is nearly public, and \cite{ely2002robust}, where mixed-strategy equilibria render such inferences irrelevant. Relative to this literature, we focus on one-sided private monitoring but add private information at the outset to construct and quantify natural, yet nontrivial, belief-dependent Markov equilibria.

Regarding signaling models, in traditional \emph{static} (i.e., sequential-move, one-shot) noisy signaling games (e.g., \citealp{matthews1983equilibrium, carlsson1997noise}), the signal realization is trivially hidden from the sender at the moment of action, but the common prior makes the receiver's belief known at the same time. In dynamic environments, the receiver's belief is public in settings with observable actions and an exogenous, public stochastic process (e.g., \citealp{daley2012waiting, gryglewicz2019strategic, kolb2019strategic}) or when there is imperfect public monitoring, such as in \cite{heinsalu2018dynamic} and \cite{dilme2019dynamic}. By contrast, our assumptions on payoffs and signal structure make all players' beliefs private.

Private beliefs arise in \cite{foster1996strategic} and \cite{bonatti2017dynamic}, where all the players have fixed private information and there is an imperfect public signal; a representation result for first-order beliefs eliminates the need for higher-order beliefs. \cite{bonatti2019consumer} in turn examine two-sided signaling when firms privately observe a summary of a consumer's past behavior to price discriminate; however, via the prices they set, firms perfectly reveal their information to the consumer. Finally, private beliefs can also result from an \emph{exogenous} private signal of the sender's type, as in \cite{feltovich2002too}.

Turning to our applications, adaptation-coordination tradeoffs are a key element in recent analyses of organizations: in static, linear-quadratic, settings, see \cite{dessein2006adaptive} and \cite{rantakari2008governing} for questions of specialization and governance, respectively; our focus is instead on the dynamics of information transmission with private signals of behavior. \cite{bouvard2019horizontal} examine a model of horizontal reputation with quadratic payoffs and symmetric uncertainty; beliefs are public in the linear Markov equilibrium constructed. Lastly, \cite{yang2019back} study a two-period model in which a trader faces, in the second period, a ``backrunner'' who has observed a private signal of the former's first-period trade; there, the feedback element is absent, and so is the need for a belief representation like ours.

To conclude, this paper contributes to a growing literature using continuous-time methods to analyze dynamic incentives. \cite{sannikov2007games} examines two-player games of imperfect public monitoring; \cite{faingold2011reputation} reputation effects with behavioral types; \cite{cisternas2018two} games of ex ante symmetric incomplete information; and \cite{bergemann2015dynamic} revenue maximization with privately informed buyers. Our representation result and derivation of belief states, the distinction between private and common-value settings, and the question of existence of equilibria make these methods virtually necessary.

\section{Application: Coordinated Adaptation\label{sec:leading_by_example}}

A team consisting of a leader (she) and a follower (he) operates over a finite horizon $[0,T]$. The environment is parametrized by a state of the world $\theta$ that is normally distributed with mean $\mu\in\mathbb{R}$ and variance $\gamma^o>0$. Letting $a_t\in\mathbb{R}$ and $\hat a_t\in \mathbb{R}$ denote the leader's action and follower's action at time $t\in [0,T]$, respectively, the team's performance is given by
\begin{eqnarray}\label{eq:leader_payoff}
\int_0^T e^{-rt}\{-(a_t-\theta)^2-(a_t-\hat{a}_t)^2\}dt,
\end{eqnarray}
where $r\geq 0$ is a discount rate. The leader knows the state of the world at the outset, while the follower only knows the prior distribution (and this is common knowledge).

We assume that the leader's preferences coincide with the team's payoffs. In turn, the follower is myopic, always attempting to match the leader's action, i.e., attempting to minimize his expectation of $(a_t-\hat{a}_t)^2$ at all times. To solve this prediction problem, this player relies solely on a private signal of the leader's action that is distorted by Brownian noise:
$$
dY_t = a_tdt+\sigma_Y dZ_t^Y.
$$

Intuitively, as the leader \emph{signals}---e.g., as she takes actions intended to drive the organization in her desired direction---observing $Y$ allows the follower to gradually adjust towards taking the ``right'' action (in this case, $\theta$). But since $Y$ is private to the follower, the leader loses track of the follower's belief in the process. Attempting to adapt the organization to new economic conditions then creates two-sided uncertainty: the organization's members do not know the long-term goals behind the leadership's actions, and the leadership does not know the organization's understanding of what should be done at all instants of time. 

We are motivated by two elements that are critical to the performance of organizations. 

\vspace{-0.1in}
\begin{enumerate}
\item  \emph{Efficient adaptation.} Adjusting to the external economic environment is a key problem for organizations, requiring substantial coordination of multiple functions (\citealp{williamson1996mechanisms}; \citealp{milgrom1992economics}), which implies that (i) misaligned incentives and (ii) information frictions are key threats. The adaptation and coordination concerns are captured by $-(a_t-\theta)^2$ and $-(a_t-\hat{a}_t)^2$, respectively; in turn, the follower's preferences partially align the players' objectives to concentrate on information frictions.\footnote{Our general analysis allows for misalignments in the players' flow payoffs (see Section \ref{subsec:existence_interior}).}

\item \emph{Bounded rationality.} The barriers that people face in solving problems and processing information are at the core of every organization (\citealp{simon1957models}). As \cite{williamson1996mechanisms} further remarks, ``failures of coordination can arise because autonomous parties read and react to signals differently, even though their purpose is to achieve a timely and compatible response'' (p. 102). In our setup, $Y$ is noisy, and its noise is idiosyncratic to the follower: examples include $Y$ being linked to a cognitive process of the follower, or to a chain of imperfect transmissions that hides the final realization from the leader.

\end{enumerate}

The leader's knowledge of the state of the world is therefore understood as expertise relevant to the current economic conditions; the transmission of this knowledge is then linked to behavior, but the transfer is slow and imperfect. In this regard, our choice to shut down communication is essentially a dimensionality constraint intended to reflect situations in which the knowledge involved is substantially richer than the code available.\footnote{These features resonate with the notion of \emph{tacit} knowledge---``know-how'' that is difficult to codify and transfer. Recognized as a key input to production, \cite{garicano2000hierarchies} examines its implications on hierarchies, while \cite{grant1996toward} argues that this knowledge is ``only being observed through its application'' and \cite{nonaka1991knowledge} that it is ``rooted in action and in an individual's commitment to a specific context.''}

We now study two information structures for the leader: in the \emph{perfect-feedback} case, the leader observes the follower's action; in the \emph{no-feedback} case, she observes nothing. That is, we keep the difficulty in transferring knowledge as given (the signal $Y$ is fixed), and we vary the quality of the information fed to the leader (which is the more likely choice variable).  These are limit instances of a model we explore Sections \ref{sec:model} and \ref{sec:eqmanalysis} under general preferences.

\vspace{-0.1in}
\paragraph{Perfect feedback (``public'') case.} If the leader perfectly observes the follower's action she can potentially infer the follower's belief, in which case the latter belief becomes \emph{public}. 

In a \emph{linear Markov equilibrium} (LME), the leader chooses actions that are linear both in her type $\theta$ and in the follower's (commonly known) belief $\hat M_t:=\hat{\mathbb{E}}_t[\theta]$, where $\hat{\mathbb{E}}_t[\cdot]$ denotes the follower's expectation operator; in turn, the follower's action is his best prediction of the leader's action, and hence it is  linear in $\hat M_t$ exclusively.\footnote{This notion of LME is \emph{perfect} when $Y$ is public, but only \emph{Nash} when $Y$ is private but the follower's action is observed---our choice of exposition (observed actions as opposed to a public $Y$) stems from the form of our general model. We keep the abbreviation later on despite its subsequent ``perfection" property. \label{footnote:Public_Nash}} For consistency throughout the paper, we write $\beta_{3t}$ for the weight on the type in the leader's strategy at $t\in [0,T]$.

\begin{proposition}[LME---Public Case]\label{prop:Leading_Public_LME}
For all $r\geq 0$ and $T> 0$:
\begin{itemize}
\item [(i)] Existence of LME: There exists a unique LME. In this equilibrium, $a_t=\beta_{3t}\theta+(1-\beta_{3t})\hat M_t$ and $\hat a_t = \hat{\mathbb{E}}_t[a_t]= \hat M_t,$ where $(\beta_{3t})_{t\in [0,T]}$ is deterministic. 
\item [(ii)] Signaling and learning: $\beta_{3t}\in (1/2,1)$ for $t<T$, $\beta_{3T}=1/2$, and $\beta_3$ is strictly decreasing. Also, $\gamma_t:=\hat{\mathbb{E}}_t[(\theta-\hat M_t)^2]$ evolves according to $\dot\gamma_t=-\left(\frac{\gamma_t\beta_{3t}}{\sigma_Y}\right)^2$.
\end{itemize}
\end{proposition}

In the LME, the leader reduces her degree of adaptation below the full-information solution $a\equiv \theta$ to coordinate with the follower---we refer to the weight $\beta_3$ on the type as the \emph{signaling coefficient}. This coefficient shapes the follower's learning captured by the posterior variance $\gamma_t$, and it remains above 1/2---the value in the static equilibrium $(\frac{1}{2}\theta+\frac{1}{2}\hat M,\hat M)$---except at the end of the game. Indeed, by more aggressively signaling her know-how, the leader can steer the follower's behavior toward the first-best action faster, effectively investing in the follower's adaptation. This incentive falls deterministically (i.e., $\beta_{3t}$ is decreasing) because of both horizon and learning effects---there is less time remaining to enjoy those benefits, and steering behavior becomes more difficult as information accumulates.\footnote{In fact, $\frac{d\hat M_t}{da_t}=\frac{\gamma_t\beta_{3t}}{\sigma_Y^2}$, so the sensitivity of the follower's belief falls with lower values of $\gamma$.}

\vspace{-0.1in}
\paragraph{No-feedback case.} Suppose now that the leader ceases to receive any information about the follower: how does she forecast $\hat M$, and how are signaling and learning affected?

To gain intuition, let us first elaborate on the form of the follower's belief when it is public. Upon conjecturing a linear Markov strategy by the leader, the follower's learning has a Gaussian structure, and so $\hat M$ is a linear function of the history $Y^t:=(Y_s:0\leq s<t)$: namely, there are deterministic $A_1(\cdot)$ and $A_2(\cdot,\cdot)$ such that
\begin{eqnarray}\label{eq:leading_Mhat_integral}
\hat M_t = A_1(t)+\int_0^t A_2(t,s)dY_s,\; t\in [0,T].
\end{eqnarray}
The leader's forecast of $\hat M$ is trivially given by the same formula, as observing $\hat a$ (or $Y$) reveals the follower's belief. That is, the leader forecasts by \emph{output}: $Y$, which reflects the consequences of her actions from the follower's perspective, fully determines her inferences.

In the absence of feedback, the signal $Y$ is not available,  making the leader's forecasting problem nontrivial. However, to the extent that $\hat M$ is as above (for potentially different $A_1$ and $A_2$), the leader can take an expectation in \eqref{eq:leading_Mhat_integral} to obtain her \emph{second-order belief}
\begin{eqnarray}\label{eq:leading_M_integral}
M_t = A_1(t) + \int_0^t A_2(t,s)a_sds,\; t\in [0,T].
\end{eqnarray}
Crucially, this belief now is a function of the leader's past actions, so the leader forecasts by \emph{input}: absent any information, the leader must reflect on her past behavior to assess how much knowledge has been transferred. As natural as it seems, however, observe that past play was completely irrelevant with perfect feedback: higher (lower) past actions only indicated that more negative (positive) shocks thwarted the leader's efforts, which is immaterial for future decision-making. The importance of forecasting by input versus output depends on the setting, but many situations will entail both; our general model captures this feature.

It is clear from \eqref{eq:leading_Mhat_integral} and \eqref{eq:leading_M_integral} that a common element in these two extreme information structures studied is that, in both, future beliefs respond to different continuation strategies linearly; through this \emph{forward-looking exercise}, the leader determines her best response to the follower's strategy. However, the follower's behavior will depend on his assessment of the informational content behind the leader's actions, and this is a \emph{backward-looking exercise}: how do different types behave at their respective histories? Whether beliefs are a function of commonly observed versus private information then introduces important differences.

Abusing notation, let us now consider a linear strategy for the leader of the form
\begin{eqnarray}\label{eq:leading_strat_NF}
a_t =\beta_{0t}\mu+\beta_{1t}M_t+\beta_{3t}\theta,
\end{eqnarray}
with the coefficients again being deterministic. It is evident that $M$ is generically private to the leader under \eqref{eq:leading_strat_NF}, as her actions carry her type---how will the follower coordinate then? Inspection of \eqref{eq:leading_M_integral} and \eqref{eq:leading_strat_NF}, however,  suggests a linear relationship between $M$ and $\theta$. Suppose then that the follower conjectures that, on path, $(M_t)_{t\in [0,T]}$ satisfies the \emph{representation}
\begin{eqnarray}\label{eq:leading_M_representation}
M_t =\left(1- \frac{\gamma_t}{\gamma^o}\right)\theta +\frac{\gamma_t}{\gamma^o}\mu,
\end{eqnarray}
where $(\gamma_t)_{t\in [0,T]}$ again denotes the follower's posterior variance but now under \eqref{eq:leading_strat_NF}--\eqref{eq:leading_M_representation}. As a proof of concept, note that setting $\gamma_0=\gamma^o$ in \eqref{eq:leading_M_representation} leads to $M_0=\mu=\hat{M}_0$, consistent with the common prior at time zero; conversely, if enough signaling has occurred, the leader thinks that the follower must have learned the state: $\gamma_t\approx 0$ leads to $M_t\approx \theta$ in the same formula.

The representation \eqref{eq:leading_M_representation} is key. First, it encodes how private monitoring alters the extent of information transmission. In fact, the new signaling coefficient---denoted $\alpha$---is obtained as the total weight on $\theta$ when inserting \eqref{eq:leading_M_representation} into the leader's strategy \eqref{eq:leading_strat_NF}, which yields
$$
\alpha:= \beta_3+\beta_1\chi,\;\; \text{ where }\;\; \chi:= 1-\frac{\gamma}{\gamma^o}.
$$

We refer to the correction term $\beta_1\chi$ stemming from \eqref{eq:leading_M_representation} as the \emph{history-inference effect} on signaling. Indeed, since the leader forecasts by input, the follower needs to infer the leader's private histories to extract the correct informational content from $Y$. From his perspective, how differently would a leader of a marginally higher type behave \emph{given} a history $Y^t$? With perfect feedback, the overall effect is $\beta_3$, as all types agree on the value that $\hat M$ takes (i.e., they pool along the belief dimension); this is not the case when there is no feedback, as their differing past actions also lead them to perceive different continuation games via $M$.

Second, the representation prevents the state space from growing: via \eqref{eq:leading_M_representation}, the follower's belief about $M$ (i.e., a third-order belief) is a function of $\hat M$, and so $\mathbb{E}_t[\hat{\mathbb{E}}_t[M_t]]$ is a function of $M_t$, and so forth. The linear-quadratic-Gaussian structure then ensures that $(\theta,M,\mu, t)$, with $M$ as in \eqref{eq:leading_M_integral}, summarizes all that is payoff-relevant for the leader after all private histories.\footnote{To conjecture the outcome of the game (i.e., \eqref{eq:leading_M_representation}), and hence to be able to interpret signals correctly, our players must generically anticipate how play unfolds when \eqref{eq:leading_M_representation} fails. Lemma \ref{lem:Belief_Rep_NF} in the Appendix proves that  \eqref{eq:leading_M_representation} holds. The method for showing existence is part of a broader approach discussed in Section \ref{subsec:existence_interior}.}

\begin{proposition}[LME---No-Feedback Case]\label{prop:Leading_NF_LME}
For all $r\geq 0$ and $T> 0$:
\begin{itemize}
\item [(i)] Existence: There exists an LME. In any such equilibrium, $\beta_0+\beta_1+\beta_3=1$; $\beta_{3t}>1/2$, $t\in [0,T)$; $\beta_{3T}=1/2$; and $\beta_1>0$ over $[0,T]$.
\item [(ii)] Signaling and learning: $\alpha:=\beta_3+\beta_1\chi$, where $\chi=1-\gamma_t/\gamma^o$, satisfies $\alpha>1/2$; $\alpha_T\to 1$ as $T\to\infty$; and $\alpha_t'\geq 0$, $t\in [0,T)$, with strict inequality if and only if $r>0$. Also, $\gamma_t:=\hat{\mathbb{E}}_t[(\theta-\hat M)^2]$ evolves as $\dot\gamma_t=-\left(\frac{\alpha_t\gamma_t}{\sigma_Y}\right)^2$.
\end{itemize}
\end{proposition}

From part (ii) we see that private monitoring overturns strictly decreasing signaling effects expected to arise under the traditional logic of public beliefs: $\alpha$ is non-decreasing.

\begin{figure}[htbp]
\begin{subfigure}{.5\linewidth}
\centering
\includegraphics[height=1.5in]{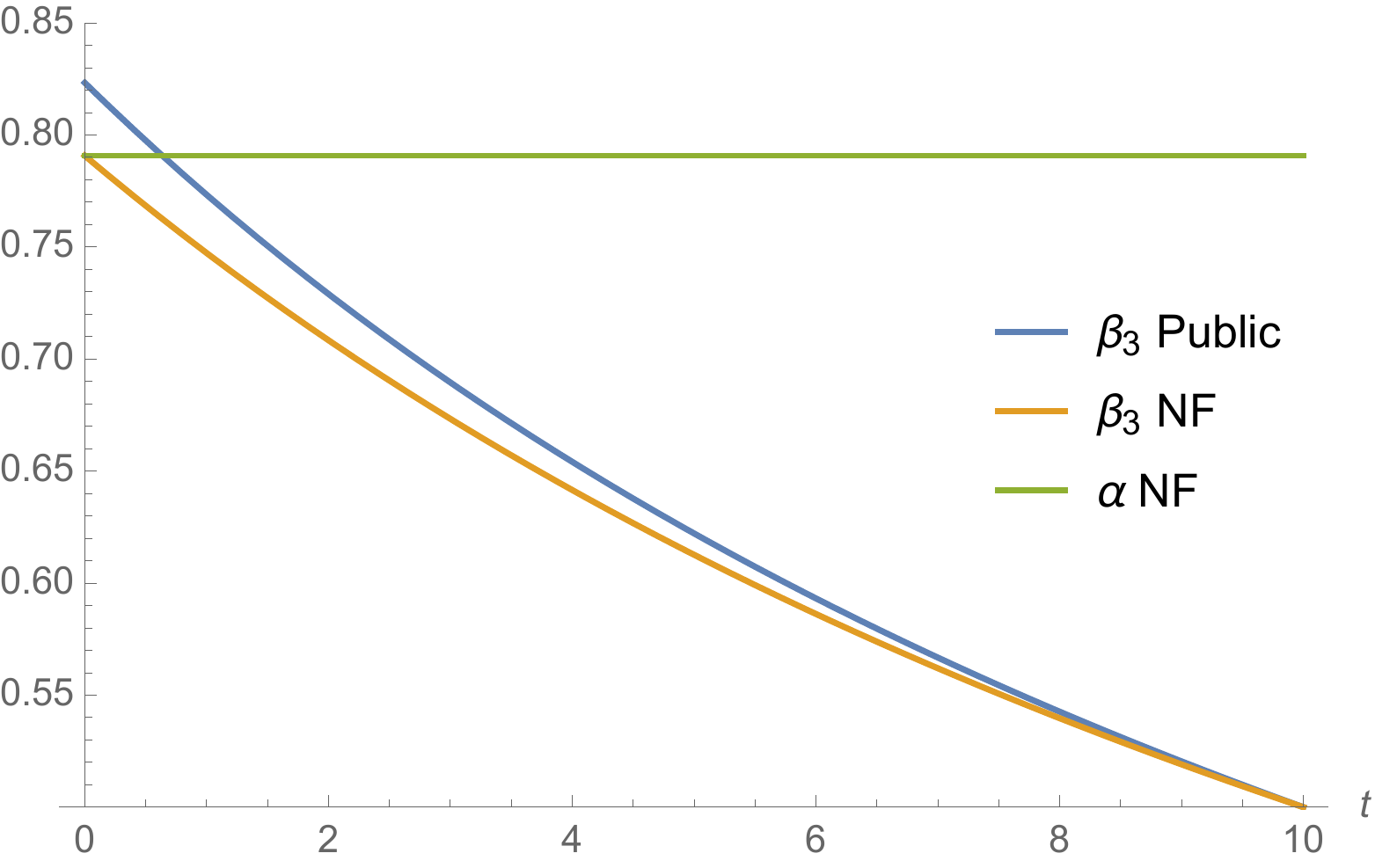}
\end{subfigure}%
\begin{subfigure}{.5\linewidth}
\centering
\includegraphics[height=1.5in]{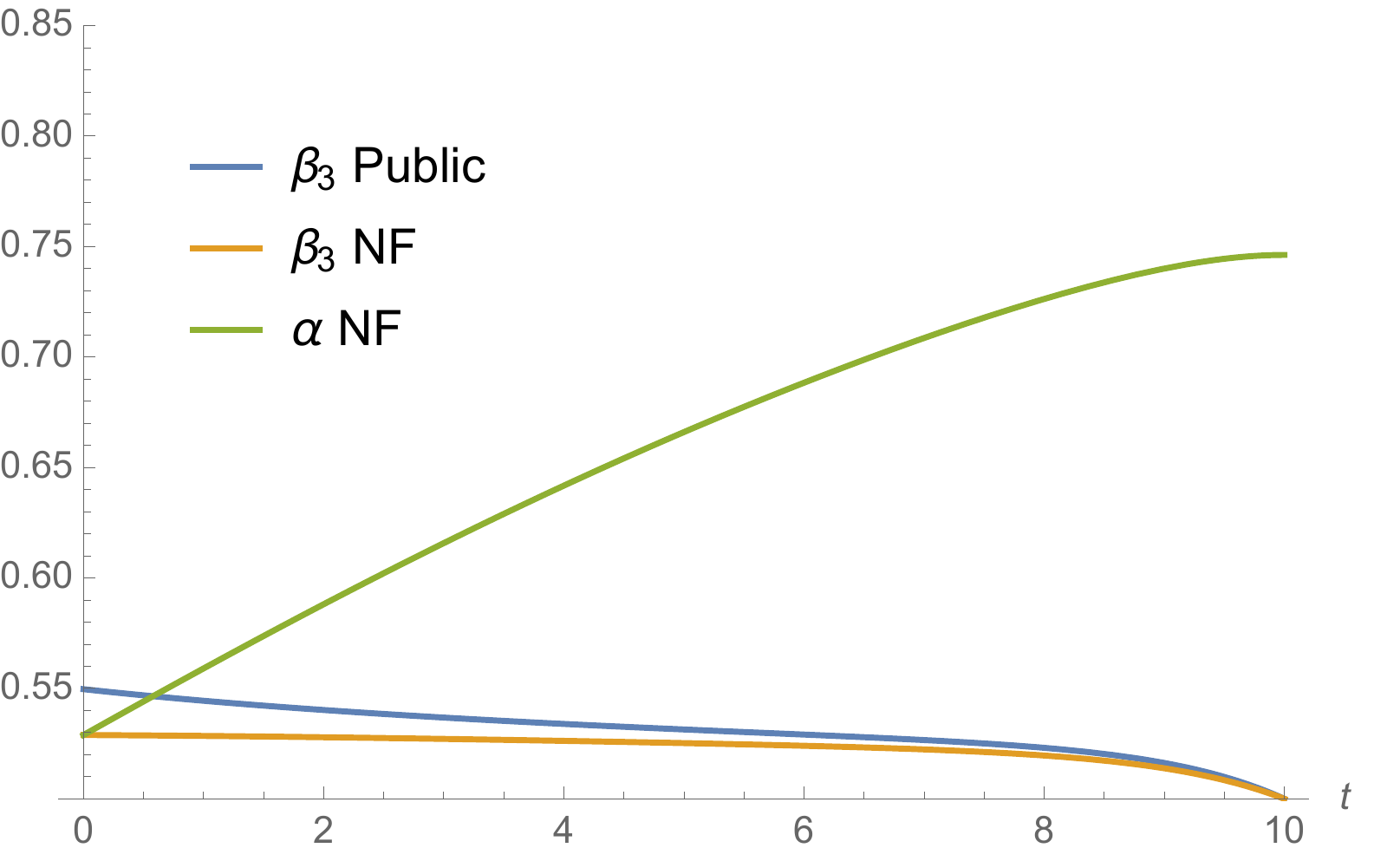}
\end{subfigure}
  \caption{Left $r=0$; Right: $r=1$. Other parameter values: $\gamma^o=1$, $\sigma_Y=1.5$, $T=10$.}\label{f1}
\end{figure}

\vspace{-0.3in}

 \paragraph{Comparison.} Figure 1 plots the signaling \emph{coefficients} in each LME. In the no-feedback case, $\beta_3$ is decreasing, so a nondecreasing $\alpha$ implies that the history-inference effect \emph{increases} over time. Indeed, because higher types take higher actions holding everything else fixed, they will expect their followers to have higher beliefs: this effect then grows---reflected in $M$ attaching an increasing weight $\chi$ to $\theta$ in \eqref{eq:leading_M_representation}---as past play acquires more relevance for predicting the continuation game as time progresses. With a positive coordination motive ($\beta_1>0$) that also strengthens with time, higher types gradually take even higher actions via this second-order belief channel, enhancing the informational content of the leader's action.

Being forced to rely on her past actions to forecast the follower's understanding essentially imposes discipline on the leader: she does not cater to the follower's belief as she would in the public case. In turn, this suggests that more knowledge is transferred to the follower. To assess the validity of this conjecture, we take advantage of the model's analytic solutions in the patient ($r=0$) and myopic ($r=\infty$) cases. Let $\gamma^{\textit{Pub}}$ and $\gamma^{\textit{NF}}$ denote the follower's posterior variance in the public and no-feedback cases, respectively. 

\begin{proposition}[Total learning]\label{prop:Leading_Learning_Comp}
(i) If $r=0$, $\beta_{30}^{\textit{Pub}}> \alpha_0$ and $\gamma_T^{\textit{Pub}}>\gamma_T^{\textit{NF}}$, all $T>0$; (ii) Given $T>0$ and $\delta\in (0,T)$, $\gamma_t^{\textit{Pub}}>\gamma_t^{\textit{NF}}$ for $t\in [T-\delta, T]$ if $r$ is large enough.
\end{proposition}

Consequently, when the leader is either patient or very impatient, in the no-feedback case the follower always has learned more by the end of the interaction. To show that this result is nontrivial, part (i) states that in the beginning, a patient leader always signals less aggressively in the no-feedback case---this is due to the anticipation of the history-inference effect at play later. Conversely, when the leader is myopic, the previous inter-temporal substitution effect disappears, so the only difference between the signaling coefficients in the no-feedback and public cases is the history-inference effect that arises in the former. As argued previously, this effect is always positive; part (ii) then follows by uniform convergence.\footnote{The ranking of terminal learning appears to hold for all $r>0$. See Figure 1 in the online appendix.}

We conclude with a discussion on how payoffs and information transmission connect:

\begin{proposition}[Team's payoff]\label{prop:Leading_Payoff_Comp} (i) If $r=0$, the team's ex ante payoff is larger in the public case; (ii) $\forall r \geq  0$, ex ante \emph{undiscounted} coordination costs equal $\sigma_Y^2\log\left(\frac{\gamma^o}{\gamma_T}\right)$ in each case.
\end{proposition}

The direct effect of shutting down the feedback channel is increased coordination costs: holding everything else fixed, the leader is now uncertain about a concave payoff. To understand the \emph{strategic} effect, part (ii) is key: in both cases, the extent of the follower's learning is a measure of \emph{total coordination costs}. Consider the perfect-feedback case: if the leader chooses an action that the follower can match, no coordination costs are created, but this necessarily implies that the leader is neglecting her type. Consequently, the leadership transmits its knowledge only when it introduces changes to which the organization does not perfectly know how to respond, generating transient miscoordination in the process.

From this perspective, private monitoring exacerbates such costs by making the follower's actions more \emph{volatile} in response to more informative, yet stable, behavior by the leader; in particular, the team is worse off in the no-feedback case when the leader is patient (part (i)). This setting of action-based information transmission then reveals that an organization's better understanding of its leadership's goals need not be indicative of better past, or even future, performance: it can be reflective of the organization's painful struggle to coordinate.\footnote{\cite{marschak1955elements} stresses the importance of incorporating action-based information transmission in the analysis of organizations: ``A realistic theory of teams would be dynamic. It takes time to process and pass messages along a chain of members; and messages must include not only information on external variables, but also information on what has been \emph{done} [emphasis added] by other members in the team'' (p. 137). See \cite{hermalin1998toward} for a theory of leadership featuring noiseless signaling, albeit in a static setting.}

The next sections develop a framework for general quadratic preferences and partially informative public feedback channels. This generality is important not only for the breadth of applications that can be explored but also because information channels often have intermediate quality. For organizations in particular, the value of the analysis is clear: improvements in quality can sometimes be prohibitively costly, and even when partial improvements are feasible, cost-benefit analyses require assessing payoffs under varying levels of uncertainty.

\vspace{-0.1in}
\section{General Model\label{sec:model}}

We consider two-player, linear-quadratic-Gaussian games with private information and private monitoring. Extensions of the baseline setup are presented in Section \ref{sec:applications} via applications.

\vspace{-0.1in}\paragraph{Players, Actions and Payoffs.}

A forward-looking \textit{long-run} player (she) and a myopic counterpart (he) interact in a repeated game that is played continuously over a time interval $[0,T]$, $T<\infty$. At each instant $t\in [0,T]$, the long-run player chooses an action $a_t$, while the myopic player chooses $\hat a_t$, both taking values over the real line. If at any instant the profile of actions chosen is $(a,\hat a)$, the long-run player's and myopic player's flow payoffs are 
\begin{eqnarray}\label{eq:p1utility}
U(a,\hat{a},\theta)\;\;\; \text{and}\;\;\; \hat U(a,\hat{a},\theta),
\end{eqnarray}
respectively, where $U$ and $\hat U$ are quadratic functions. In \eqref{eq:p1utility}, $\theta$ denotes a normally distributed random variable with mean $\mu\in\mathbb{R}$ and variance $\gamma^o>0$. The long-run player discounts the future at a rate $r>0$, while the myopic player cares only about his flow payoff at all times.

To state our main assumptions on the functions $U$ and $\hat U$, we introduce the scalars 
$$
u_{xy}:= \frac{\partial^2 U/\partial x\partial y}{|\partial^2 U/\partial a^2|}\;\; \text{ and }\;\; \hat{u}_{xy}:= \frac{\partial^2 \hat U/\partial x\partial y}{|\partial^2 \hat U/\partial \hat a^2|}, \; \text{ for }\; x,y\in \{a,\hat a,\theta\}.
$$
Indeed, if the players' flow payoffs are concave in their respective actions, best responses will exhibit denominators like the above, allowing us to state our conditions in normalized form:
\begin{assu}\label{assu:static}\strut Flow payoffs satisfy (i) $u_{aa}=\hat u_{\hat a\hat a}=-1$ (strict concavity); (ii) $u_{a\theta}(u_{a\theta}+u_{a\hat a}\hat u_{\hat a\theta})>0$ (nontrivial signaling); (iii) $|\hat u_{\hat a\theta}|+|\hat u_{a\hat a}|\neq 0$ and $|u_{a\hat a}|+|u_{\hat a\hat a}|\neq 0$ (second-order inferences); and (iv) $u_{a\hat a}\hat u_{a\hat a}<1$ (myopic best-replies intersect).
\end{assu}

Our first requirement is that $\theta$ be strategically relevant for the long-run player (i.e., $u_{a\theta}\neq 0$), which is implied by (ii). Part (iii) then invokes the use of higher-order inferences. Specifically, the first condition states that the myopic player's first-order belief influences his behavior, either because he cares about $\theta$ directly ($\hat u_{\hat a\theta}$ term) or indirectly via the long-run player's action ($\hat u_{a\hat a}$ term); in turn, the second condition forces the long-run player to forecast the myopic player's belief, due to either an interaction term ($u_{a\hat a}$) or a nonlinear effect ($u_{\hat a\hat a}$). (Clearly, (iii) is a choice to focus on the most interesting cases rather than a limitation.)

The remaining conditions are used to find equilibria in linear strategies. The  concavity of the players' objectives with respect to their own actions (part (i)) gives rise to linear best responses. Coupled with (iv), the static game of two-sided private information that arises at the end of the interaction will admit a Nash equilibrium; part (ii) then ensures that this equilibrium entails type dependence. We will revisit these assumptions in Section \ref{sec:eqmanalysis}.

\vspace{-0.1in}
\paragraph{Information.}  The long-run player observes the value of $\theta$ before play begins, while the myopic player only knows its distribution $\theta\sim\mathcal{N}(\mu,\gamma^o)$ (and this is common knowledge). There are also two signals with full support due to the presence of Brownian noise: 
\begin{eqnarray}\label{eq:XYsignals}
dX_{t}=\hat a_{t}dt+\sigma_X dZ^X_{t}\;\;\; \text{ and }\;\;\; dY_{t}= {a}_{t}dt+\sigma_Y dZ^Y_{t},
\end{eqnarray}
where $Z^X$ and $Z^Y$ are orthogonal and the volatility parameters $\sigma_Y$ and $\sigma_X$ strictly positive.

Our key departure from existing analyses is to make $Y$---which carries information about the long-run player's actions---privately observed by the myopic player; instead, the signal $X$ carrying the latter player's action remains public. This mixed private-public information structure is important for our construction, but it is also natural for analyzing sender-receiver games: it makes the departure minimal while still economically relevant for applications.

In what follows, we let $\mathbb{E}_t[\cdot]$ denote the long-run player's conditional expectation operator, which can condition on the histories $(\theta,a_{s},X_s:0\leq s\leq t)$ and on her conjecture of the myopic player's play. Similarly, $\hat{\mathbb{E}}_t[\cdot]$ denotes the myopic player's analog, which conditions on $(\hat a_{s},X_s,Y_s:0\leq s\leq t)$ and on his belief about the long-run player's strategy, $t\geq 0$.\footnote{In particular, flow payoffs do not convey any additional information to the players (i.e.,  payoffs accrue after time $T$, or they can be written in terms of the actions and signals observed by each player).}

\vspace{-0.1in}
\paragraph{Strategies and Equilibrium Concept.} With full-support monitoring, the only off-path histories for each player are those in which that player herself/himself has deviated. Thus, we use the Nash equilibrium concept for defining the equilibrium of the game, as imposing full sequential rationality places no additional restrictions on the set of equilibrium outcomes.

From this perspective, an admissible strategy for the long-run player is any square-integrable real-valued process $(a_t)_{t\in [0,T]}$ that is progressively measurable with respect to the filtration generated by $(\theta,X)$. The analogous notion for the myopic player involves the identical integrability condition, but the measurability restriction is with respect to $(X,Y)$.\footnote{Square integrability is in the sense of the time-zero expectations of $\int_0^T a_t^2dt$ and $\int_0^T \hat a_t^2dt$ being finite. This ensures that a strong solution to \eqref{eq:XYsignals} exists and thus that the outcome of the game is well defined.}

\begin{definition}[Nash equilibrium.]\label{def:eqbm} An admissible pair $(a_t,\hat{a}_t)_{t\geq 0}$ is a Nash equilibrium if: (i) the process $(a_t)_{t\in [0,T]}$ maximizes $\mathbb{E}_0\left[\int_0^T e^{-rt} U(a_{t},\hat{a}_{t},\theta) dt\right] $; and (ii) for each $t\in [0,T]$, $\hat{a}_t$ maximizes $\hat{\mathbb{E}}_t[\hat U(a_t,\hat{a}_t,\theta)]$ when $(\hat{a}_s)_{s<t}$ has been followed.
\end{definition}

In the next section, we characterize Nash equilibria supported by strategies that are fully sequentially rational, thereby specifying behavior after deviations. The equilibria studied generalize that of Section \ref{sec:leading_by_example} for the no-feedback case $\sigma_X=\infty$ to the whole range $0<\sigma_X\leq \infty$.

\begin{remark}[Extensions] Our methods can accommodate various extensions, including: (i) a terminal payoff for the long-run player $\Psi(\hat a_T)$, where $\Psi$ is quadratic; and (ii)  the drift of $X$ in \eqref{eq:XYsignals} taking the form $\hat{a}_{t}+\nu a_{t}$, where $\nu \in [0,1]$ is a scalar. See Section \ref{subsec:political} for a reputation model featuring (i) and Section \ref{subsec:insidertrading} for an insider trading model featuring (ii) (that also accommodates $\partial^2 U/\partial a^2=0$, as is traditional in that literature).\end{remark}

\vspace{-0.2in}
\section{Equilibrium Analysis: Linear Markov Equilibria\label{sec:eqmanalysis}}

We construct linear Markov (perfect) equilibria using the players' beliefs as the relevant states. Indeed, the quadratic payoffs and linear signals open the possibility for quadratic value functions that are supported by strategies that are linear in some state variables. For this to work, however, the states themselves must be linear in the signals available. With a Gaussian information structure, it is then natural to appeal to belief-based states.

The appeal of such equilibria is twofold. First, the Markov restriction captures that behavior depends only on the aspects of the histories that the players perceive to be payoff-relevant. Second, in equilibrium, the players' actions are linear in the signals observed, which generalizes traditional linear equilibria widely employed in static applied-theory work.

\subsection{Belief States and Representation Lemma\label{subsec:beliefdecomposition}}

Specifically, we characterize equilibria in which, after their corresponding (on- or off-path) private histories, the long-run player and the myopic counterpart play according to
\begin{eqnarray}
a_t&=&\beta_{0t}+\beta_{1t}M_t+\beta_{2t}L_t+\beta_{3t}\theta \label{eq:P1StrategyGeneral}\\
\hat{a}_t&=&\delta_{0t}+\delta_{1t}\hat M_t+\delta_{2t}L_t\label{eq:P2StrategyGeneral}.
\end{eqnarray}
Here, $\hat M_t:=\hat{\mathbb{E}}_t[\theta]$ is the myopic player's first-order belief, $M_t:= \mathbb{E}_t[\hat M_t]$ the long-run player's second-order counterpart, and $L_t:=\mathbb{E}[\theta|\mathcal{F}_t^X]$ is the belief about $\theta$ using the \emph{public information exclusively}; the coefficients $\beta_{it}$ and $\delta_{jt}$, $i=0,1,2,3$ and $j=0,1,2$, are deterministic.

Intuitively, because the long-run player conditions her actions on her type ((ii) in Assumption \ref{assu:static}), the myopic player's belief $(\hat M_t)_{[0,T]}$ is a relevant state (first part in (iii), Assumption \ref{assu:static}). However, this implies that the long-run player must forecast the myopic player's belief to determine her best response (second part in (iii), Assumption  \ref{assu:static}), which makes $(M_t)_{t\in [0,T]}$ payoff-relevant. The appearance of $L_t$ is in turn linked to the nature of $M_t$, as follows.

As is traditional, the long-run player will use $X$ to forecast $\hat M$ since this signal carries the myopic player's action. The novelty is that, due to the private monitoring, she will also use the history of her past actions in this forecasting exercise. Intuitively, as long as the public signal is imperfect (i.e., $\sigma_X>0$), the long-run player does not perfectly know the inferences made by the myopic player, so higher action profiles become statistically informative of higher private observations by the myopic player, and vice versa. The explicit dependence of the long-run player's second-order belief $M$ on her past actions---as occurred in \eqref{eq:leading_M_integral}---makes this state a private one even in equilibrium, as those actions depend on her actual type. The myopic player is then forced to make an inference about this belief (and so forth).

Along the path of play of any \emph{pure} strategy, however, the outcome of the game should depend only on $(\theta,X,Y)$. In particular, $M$ must be a function of the tuple $(\theta,X)$, which is the long-run player's only source of information. The Gaussian structure then suggests the existence of process $(L_t)_{t\in [0,T]}$ depending only on the public information, and a deterministic function $\chi$, such that, under the linear profile \eqref{eq:P1StrategyGeneral}--\eqref{eq:P2StrategyGeneral} carrying this public process,
\begin{eqnarray}\label{eq:BeliefDecomposition}
M_t=\chi_t\theta+(1-\chi_t)L_t.
\end{eqnarray}
The representation \eqref{eq:BeliefDecomposition} is at the core of our analysis. First, if true, it shows that the ``beliefs about beliefs'' problem is manageable: as the myopic player uses \eqref{eq:BeliefDecomposition} to forecast $M$, all higher-order expectations can be written in terms of the aforementioned belief states. (Clearly, in this third-order inference step by the myopic player, $L$ becomes payoff-relevant.)

Second, the representation encodes the separation that occurs via the second-order belief channel. Indeed, \eqref{eq:BeliefDecomposition} captures how, under linear strategies, the long-run player balances her past play ($\chi\theta$ term) and the public signal ($(1-\chi)L$ term) when forecasting the myopic player's belief. Inserting \eqref{eq:BeliefDecomposition} into the long-run player's strategy \eqref{eq:P1StrategyGeneral} yields the action process
\begin{align}
a_t= \underbrace{\beta_{0t}}_{=:\alpha_{0t}}+\underbrace{(\beta_{2t}+\beta_{1t}(1-\chi_t))}_{=:\alpha_{2t}}L_t+\underbrace{(\beta_{3t}+\beta_{1t}\chi_t)}_{=:\alpha_{3t}}\theta,\label{eq:P1StrategyOnPath}
\end{align}
from which the signaling coefficient is $\alpha_{3t}:=\beta_{3t}+\beta_{1t}\chi_t$. The term $\beta_1\chi_t$ encodes the aforementioned \emph{history-inference effect}: different types take different actions in equilibrium partly because their differing past actions have lead them to hold different beliefs today.

Our approach for characterizing \eqref{eq:BeliefDecomposition} is constructive. To state the formal result, we omit the hat symbol for convenience and denote the myopic player's posterior variance simply by
$$
\gamma_t:=\hat{\mathbb{E}}_t[(\theta-\hat{M}_t)^2].
$$

\begin{lemma}[Second-order belief representation]\label{lem:BeliefDecomp} Suppose that $(X,Y)$ is driven by \eqref{eq:P1StrategyGeneral}--\eqref{eq:P2StrategyGeneral} and the myopic player believes that  \eqref{eq:BeliefDecomposition}, with $(L_t)_{t\in[0,T]}$ a process that depends only on the public information, holds. Then \eqref{eq:BeliefDecomposition} holds at all times (path-by-path of $X$), if and only if
\begin{eqnarray}\label{eq:gammadot}
\dot\gamma_t &=&-\frac{\gamma_{t}^2(\beta_{3t}+\beta_{1t}\chi_{t})^2}{\sigma_Y^2}, \hspace{3.9cm}\gamma_0=\gamma^o,\\\label{eq:chidot}
\dot\chi_t &=&  \frac{\gamma_{t}(\beta_{3t}+\beta_{1t}\chi_t)^2(1-\chi_t)}{\sigma_Y^2}-\frac{\gamma_{t}\chi_t^2\delta_{1t}^2}{\sigma_X^2},\hspace{1cm} \chi_0=0,\\\label{eq:dLonpath}
dL_t&=&(l_{0t}+l_{1t}L_t)dt+B_t dX_t, \hspace{3cm}L_0=\mu,
\end{eqnarray}
with $(l_{0t},l_{1t},B_t)$ given in \eqref{eq:l0_l1_B} deterministic. Also, $L_t=\mathbb{E}[\theta|\mathcal{F}_t^X]$ and $\gamma_t\chi_t=\mathbb{E}_t[(M_t-\hat M_t)^2]$.
\end{lemma}

In light of the lemma, the representation \eqref{eq:BeliefDecomposition} reads $M_t= \frac{\mathrm{Var}_t}{\widehat{\mathrm{Var}}_t}\theta + \left(1- \frac{\mathrm{Var}_t}{\widehat{\mathrm{Var}}_t}\right)\mathbb{E}[\theta|\mathcal{F}_t^X]$, where we have used the notation $\widehat{\mathrm{Var}}_t:=\hat{\mathbb{E}}_t[(\theta-\hat{M}_t)^2]$ and $\mathrm{Var}_t:=\mathbb{E}_t[(M_t-\hat M_t)^2]$, the latter measuring the long-run player's uncertainty about $\hat M$. Indeed, in forecasting $\hat M$, the only informational advantage that the long-run player has relative to an outsider who observes $X$ exclusively is that she knows what actions she has taken, and such actions carry her type. Under linear strategies, learning is Gaussian, so (i) $M_t$ is a linear combination of $\theta$ and $\mathbb{E}[\hat M_t|\mathcal{F}_t^X]$, and (ii) the weights are deterministic; the representation then follows from $\mathbb{E}[\hat M_t|\mathcal{F}_t^X]=\mathbb{E}[\theta|\mathcal{F}_t^X]$.  Observe also that the linearity of $\mathbb{E}[\theta|\mathcal{F}_t^X]$ in the history $(X_s:0\leq s<t)$ can be deduced from the linearity of \eqref{eq:dLonpath} both in $L$ and in the increments of $X$.\footnote{The resulting expression for $L$ holds irrespective of the past private histories of the players: this is because deviations are hidden and hence each player thinks the counterparty has constructed $L$ using \eqref{eq:dLonpath}.}

The $\chi$-ODE \eqref{eq:chidot} quantifies the dynamics of the relevance of past behavior in the previous forecasting exercise. Indeed, by the common prior, $\mathrm{Var}_0=0$ and $\mathbb{E}[\theta|\mathcal{F}_0^X]=\mu$; thus, $M_0=\mu$ in the display above, and so the $\chi$-ODE must start at zero. As signaling progresses, however, the long-run player loses track of $\hat M$ (i.e., $\mathrm{Var}_t>0$), forcing her to rely on her past behavior: this is captured in $\dot\chi>0$ as soon as $\alpha_3>0$ in \eqref{eq:chidot}. The positivity of $\chi$ then simply reflects that the long-run player expects $\hat M$ to gradually incorporate her type via this channel.

The relative importance of past play will naturally depend on the quality of the public information. Consider the last term $\gamma_{t}\chi_t^2\delta_{1t}^2/\sigma_X^2$ in \eqref{eq:chidot}. If $\sigma_X=\infty$ or $\delta_1\equiv 0$, the public signal is uninformative: indeed, $L_t=L_0=\mu$ and $\chi_t=1-\gamma_t/\gamma^o$ in both cases, as in the no-feedback analysis of Section \ref{sec:leading_by_example}.\footnote{In \eqref{eq:chidot}, $\delta_1/\sigma_X\equiv 0$ leads to the $\chi$-ODE in the no-feedback case (see the proof of Lemma \ref{lem:Belief_Rep_NF}). } Apart from these cases, the public information is always useful. In particular, as $\delta_1^2/\sigma_X^2$ grows, more downward pressure is exerted on the growth of $\chi$, reflecting a weaker dependence on past play as the quality of $X$ improves, all else being equal; in the limit, $\chi\equiv 0$ and $L\to \hat M$, and so $M\to \hat M$ (i.e., the environment becomes public). The no-feedback case then maximizes the potential amplitude of the history-inference effect.

Our subsequent analysis takes the system of ODEs for $(\gamma,\chi)$ as an input, so we require \eqref{eq:gammadot}--\eqref{eq:chidot} to have a unique solution to ensure that the ODE-characterization is valid. Note that the weight on $\hat M$ in the myopic player's best reply is given by $\hat u_{\hat a\theta}+\hat u_{\hat a a} [\beta_{3t}+\beta_{1t}\chi_t]$.

\begin{lemma}[Learning ODEs]\label{lem:UniqueLearningSol}
Suppose $(\beta_1,\beta_3)$ is continuous and $\delta_{1t}=\hat u_{\hat a\theta}+\hat u_{\hat a a} [\beta_{3t}+\beta_{1t}\chi_t]$. Then \eqref{eq:gammadot}--\eqref{eq:chidot} has a unique solution and $0<\gamma_t\leq\gamma^o$ and $0\leq\chi_t<1$ for all $t\in [0,T]$. If, moreover, $\beta_{30}\neq 0$, then these inequalities are strict over $(0,T]$.
\end{lemma}

The filtering equations are valid under weak integrability conditions on the coefficients in \eqref{eq:P1StrategyGeneral}--\eqref{eq:P2StrategyGeneral}, from which $\gamma_t=\hat{\mathbb{E}}_t[(\theta-\hat{M}_t)^2]$ and $\chi_t=\mathrm{Var}_t/\widehat{\mathrm{Var}}_t=\mathbb{E}_t[(M_t-\hat M_t)^2]/\gamma_t$ must solve the system; a mild strengthening of the conditions ensures that no other solutions exist, and if $\beta_{30}\neq 0$ (a property our equilibria satisfy), some information indeed gets transmitted.

Our derivation of the representation \eqref{eq:BeliefDecomposition} exploits the tractability of the Gaussian filtering under linear strategies. Due to the full-support monitoring, the myopic player expects  $a_t=\alpha_{0t}+ \alpha_{2t}L_t+\alpha_{3t}\theta$ defined in  \eqref{eq:P1StrategyOnPath}, and $L$ is public. The myopic player's learning problem of filtering $\theta$ from $Y$ is thus (conditionally) Gaussian, so his belief is characterized by a stochastic mean $(\hat M_t)_{t\in [0,T]}$ and the deterministic variance $(\gamma_t)_{t\in [0,T]}$. But the linearity of the signal structure renders the pair $(\hat M, X)$ (conditionally) Gaussian too. The long-run player's filtering then yields a second mean-variance pair, with $M_t$ now an explicit linear function of her past actions. One can insert the linear strategy \eqref{eq:P1StrategyGeneral} into $M_t$ to pin down $(\chi,L)$.

The representation \eqref{eq:BeliefDecomposition} then relies on the long-run player following the linear strategy  \eqref{eq:P1StrategyGeneral}. In particular, that $M$ is spanned by $\theta$ and $L$ is no longer true after deviations, and such deviations are needed to evaluate the candidacy of \eqref{eq:P1StrategyGeneral} as an LME. This brings us to the third property behind \eqref{eq:BeliefDecomposition}: it captures a divergence in the game's structure at on- versus off-path histories, a well known feature of games with private monitoring. The next result introduces the law of motion of $M$ and $L$ for an arbitrary strategy of the long-run player.

\begin{lemma}[Controlled dynamics]\label{lem:M and L laws of motion}
Suppose that the myopic player follows \eqref{eq:P2StrategyGeneral} and believes that \eqref{eq:P1StrategyGeneral} and \eqref{eq:BeliefDecomposition}  hold. Then, if the long-run player follows $(a_t')_{t\in [0,T]}$, from her perspective
\begin{eqnarray}\label{eq:dM_off_path}
dM_t&=&\frac{\gamma_t\alpha_{3t}}{\sigma_Y^2}(a'_{t}-[\alpha_{0t}+\alpha_{2t}L_t+\alpha_{3t}M_t])dt+\frac{\chi_t\gamma_t\delta_{1t}}{\sigma_X}dZ_t\\\label{eq:dL_off_path}
dL_t&=& \frac{\chi_{t}\gamma_{t}\delta_{1t}}{\sigma_X^2(1-\chi_t)}[\delta_{1t}(M_t-L_t)dt+\sigma_X dZ_t],
\end{eqnarray}
with $Z_t:=\frac{1}{\sigma_X}[X_t-\int_0^t(\delta_{0s}+\delta_{1s}M_s+\delta_{2s}L_s)ds]$ being a Brownian motion.
\end{lemma}

The dynamic \eqref{eq:dM_off_path} illustrates how the long-run player expects her future choices to affect her future beliefs. In particular, she will revise her belief upward when $a_t'>\mathbb{E}_t[\alpha_{0t}+\alpha_{2t}L_t+\alpha_{3t}\hat{M}_t]$, i.e., when she expects to beat the myopic player's expectation of her own behavior. The intensity of the revision is given by $\gamma_t\alpha_{3t}/\sigma_Y^2$: it increases with both the myopic player's  uncertainty ($\gamma$) and his conjecture of the long-run player's strength of signaling ($\alpha_3$). Further, it is clear that $M$ is deterministic only if $\delta_1/\sigma_X\equiv 0$, exactly as in Section \ref{sec:leading_by_example}.\footnote{Our choice of dynamic programming over optimal control in Section \ref{sec:leading_by_example} is not only due to the deterministic property being nongeneric: the latter approach obscures how the history-inference effect shapes signaling.}

The appearance of $M$ in the drift of \eqref{eq:dL_off_path} shows that the long-run player expects to influence $L$, despite her actions not entering the public signal: this happens through her influence on the myopic player's behavior. Consequently, a \emph{signal-jamming effect} arises: the incentive to influence a public belief (albeit only indirectly), with such incentives being perfectly accounted for in equilibrium---this effect is obviously absent in the no-feedback case ($\sigma_X=\infty$). The drift of \eqref{eq:dL_off_path} also shows that $L$ chases $M$ on average, reflecting that someone who only observes $X$ is able to gradually learn the long-run player's type over time.

Finally, observe that the pair $(\gamma,\chi)$ appears explicitly in the evolution of $(M,L)$. Indeed, this is because of the role of $(\gamma,\chi)$ in the myopic player's learning process: since deviations are hidden, this player always assumes that \eqref{eq:BeliefDecomposition} holds when constructing his belief.

\begin{remark}[$M$ as a function of past actions]\label{remark:M:pastplay}
In Lemma \ref{lem:M and L laws of motion}, insert the definition of $Z_t$ into \eqref{eq:dM_off_path} to solve for $M_t$ as a linear function of $(a_s,L_s,X_s)_{s<t}$, where $L_s$ is a linear function of $(X_\tau)_{\tau<s}$ via \eqref{eq:dLonpath}. The resulting expression generalizes \eqref{eq:leading_M_integral} for the no-feedback case.\end{remark}

\subsection{Dynamic Programming and the Boundary-Value Problem\label{subsec:dyn_prog_bvp}}

\paragraph{The long-run player's best-response problem.} Given a conjecture $\vec\beta:=(\beta_0,\beta_1,\beta_2,\beta_3)$ by the myopic player, $\vec\delta:=(\delta_{0},\delta_{1},\delta_{2})$ is found by matching coefficients in
\begin{eqnarray}\label{eq:myopicBR}
\hat a_{t}:=\delta_{0t}+\delta_{1t}\hat M_t+\delta_{2t}L_t=\arg\max\limits_{\hat{a}'}\hat{\mathbb{E}}_t[\hat U(\alpha_{0t}+\alpha_{2t}L_t+\alpha_{3t}\theta,\hat{a}',\theta)],
\end{eqnarray}
with $\vec\alpha:=(\alpha_0,\alpha_2,\alpha_3)$ as in \eqref{eq:P1StrategyOnPath}. Since the flow $U(a_t,\hat a_t,\theta)$ is quadratic and $M_t:=\mathbb{E}_t[\hat M_t]$, we can write the long-run player's total payoff as a function of $(M_t)_{t\in[0,T]}$ as follows:
\begin{eqnarray}\label{eq:LRpprogram} 
\mathbb{E}_0\bigg[\int_0^T e^{-rt} U(a_{t},\delta_{0t}+\delta_{1t} M_t+\delta_{2t}L_t,\theta) dt\bigg]+\frac{1}{2}\frac{\partial^2 U}{\partial \hat a^2}\int_0^T e^{-rt}\delta_{1t}^2\gamma_t\chi_t dt.
\end{eqnarray}
Indeed, in writing $\mathbb{E}_t[\hat M_t^2]=M_t^2+\mathbb{E}_t[(M_t-\hat M_t)^2]$, the Gaussian learning structure guarantees that the variances $\mathbb{E}_t[(M_t-\hat M_t)^2]$ are independent of the long-run player's actual behavior, determined instead by the candidate equilibrium profile; by Lemma \ref{lem:BeliefDecomp}, their value is $\chi_t\gamma_t$. From here, it is clear that $(t,\theta,L,M)$ is a sufficient statistic for the long-run player, with the time variable capturing time-horizon effects and the learning effects encoded in $(\gamma,\chi)$.

The long-run player's problem can then be stated as maximizing \eqref{eq:LRpprogram} subject to the dynamics \eqref{eq:dM_off_path}--\eqref{eq:dL_off_path} of $(M, L)$, which depend on $(\gamma,\chi)$ satisfying \eqref{eq:gammadot}--\eqref{eq:chidot}.\footnote{The long-run player's problem is, in practice, one of optimally controlling an \emph{unobserved state}. We are allowed to filter first and then optimize, because the \emph{separation principle} applies. See the proof of Lemma \ref{lem:M and L laws of motion}.} To tackle this best response problem, we postulate a quadratic value function
\begin{align}
V(\theta,m,\ell,t)=v_{0t}+v_{1t}\theta+v_{2t}m+v_{3t}\ell+v_{4t}\theta^2+v_{5t}m^2+v_{6t}\ell^2+v_{7t}\theta m+v_{8t}\theta\ell+v_{9t}m\ell,\notag
\end{align}
where $v_{i\cdot}$, $i=0,...,9$ depend on time only. The Hamilton-Jacobi-Bellman (HJB) equation is
\begin{align*}
r V=\sup_{a'} \left\lbrace \tilde U(a',\mathbb{E}_t[\hat{a}_t],\theta)+V_t+ \mu_M(a') V_m+\mu_L V_\ell+\frac{\sigma_M^2}{2} V_{mm}+\sigma_M\sigma_L V_{m\ell}+\frac{\sigma_L^2}{2}V_{\ell\ell}\right\rbrace,
\end{align*}
where $\tilde U:=U+\frac{1}{2}\frac{\partial^2 U}{\partial \hat a^2}\delta_{1t}^2\gamma_t\chi_t$, $\mu_M(a')$ and $\mu_L$ (respectively, $\sigma_M$ and $\sigma_L$) denote the drifts (respectively, volatilities) in  \eqref{eq:dM_off_path} and \eqref{eq:dL_off_path}, and $\hat a_t$ is determined via \eqref{eq:myopicBR}.

A Nash equilibrium in linear Markov strategies immediately follows when $\beta_{0t}+\beta_{1t}M+\beta_{2t}L+\beta_{3t}\theta$ is an optimal policy for the long-run player. Indeed, along the path of play of such a policy, the representation \eqref{eq:BeliefDecomposition} holds by construction, and so the long-run player's behavior is given by $a_t=\alpha_{0t}+\alpha_{2t}L_t+\alpha_{3t}\theta$, where $(L_t)_{t\in [0,T]}$ follows \eqref{eq:dLonpath} in Lemma \ref{lem:BeliefDecomp}; i.e., actions are a function of $(t,\theta,X)$ exclusively. However, conditioning differently on $L$ and $M$ is profitable after deviations---the policy $\beta_{0t}+\beta_{1t}M+\beta_{2t}L+\beta_{3t}\theta$ then specifies how to behave at such off-path histories, effectively inducing an LME that is also \emph{perfect}.\footnote{The myopic player's behavior is specified in \eqref{eq:myopicBR}, where $(t,L,\hat M)$ is the relevant state (\eqref{eq:dhatM} shows how $\hat M$ evolves). While deviations by this player do affect $L$, it is clear that no additional states are needed for our players after deviations. Also, all the payoff-relevant histories are reachable on path, so the sequential rationality requirement is trivial for this player in an LME. All this is true if this player is forward looking.}

\vspace{-0.1in}
\paragraph{The boundary-value problem (BVP).} We briefly explain how to obtain a system of ordinary differential equations (ODEs) for $\vec \beta$. Letting $a(\theta,m,\ell,t)$ denote the maximizer of the right-hand side in the HJB equation, the first-order condition (FOC) reads
\begin{eqnarray}\label{eq:FOCHJB}
\frac{\partial U}{\partial a}(a(\theta,m,\ell,t),\delta_{0t}+\delta_{1t} m+\delta_{2t}\ell,\theta)+\underbrace{\frac{\gamma_t\alpha_{3t}}{\sigma_Y^2}}_{dM_t/da_t}\underbrace{[v_{2t}+2v_{5t}m+v_{7t}\theta+v_{9t}\ell]}_{V_m(\theta,m,\ell,t)}=0.
\end{eqnarray}

Solving for $a(\theta,m,\ell,t)$ in \eqref{eq:FOCHJB}, the equilibrium condition becomes $a(\theta,m,\ell,t)=\beta_{0t}+\beta_{1t}m+\beta_{2t}\ell+\beta_{3t}\theta$, which is a linear equation. We can then solve for $(v_2,v_5,v_7,v_9)$ as a function of the coefficients $\vec\beta$ and insert the resulting expressions into the HJB equation along with $a(\theta,m,\ell,t)=\beta_{0t}+\beta_{1t}m+\beta_{2t}\ell+\beta_{3t}\theta$, to obtain a system of ODEs for the $\vec\beta$ coefficients. Critically, the system is coupled with the ODEs that $v_6$ and $v_8$ satisfy (and that are readily obtained from the HJB equation): since $M$  feeds into $L$ (see \ref{eq:dL_off_path}), the envelope condition with respect to the controlled state $M$ cannot deliver a self-contained system for the optimal policy. Finally, because the pair $(\gamma,\chi)$ affects the law of motion of $(M,L)$, it also influences $(\vec\beta,v_6,v_8)$, and so the ODEs \eqref{eq:gammadot}--\eqref{eq:chidot} must be included.


This procedure leads to a system of ODEs for $(\beta_0,\beta_1,\beta_2,\beta_3,v_6,v_8,\gamma,\chi)$; we also need the \emph{boundary conditions}. First, there are the exogenous initial conditions that $\gamma$ and $\chi$ satisfy, i.e., $\gamma_0=\gamma^o>0$ and $\chi_0=0$. Second, there are terminal conditions $v_{6T}=v_{8T}=0$ due to the absence of a lump-sum terminal payoff in the long-run player's problem. Third, and more interesting, there are \emph{endogenous} terminal conditions that are determined by the static (Bayes) Nash equilibrium that arises from myopic play at time $T$. In fact, letting $u_0:=\frac{\partial U/\partial a}{|\partial^2 U/\partial a^2|}(0,0,0)$ and $\hat u_0:=\frac{\partial \hat U/\partial \hat a}{|\partial^2 \hat U/\partial \hat a^2|}(0,0,0)$ denote the intercepts of the players' static best responses, it is easy to verify that this equilibrium entails the coefficients
\begin{eqnarray}
\beta_{0T}=\frac{u_0+u_{a\hat a}\hat u_0}{1-u_{a\hat a}\hat u_{\hat a a}},\; \beta_{1T}=\frac{u_{a\hat a}[u_{a\theta}\hat u_{\hat a a}+\hat u_{\hat a\theta}]}{1-u_{a\hat a}\hat u_{\hat a a}\chi_T},\; \beta_{2T}=\frac{u_{a\hat a}^2\hat u_{\hat a a}[u_{a\theta}\hat u_{\hat a a}+\hat u_{\hat a\theta}](1-\chi_T)}{(1-u_{a\hat a}\hat u_{\hat a a})(1-u_{a\hat a}\hat u_{\hat a a}\chi_T)},\; \beta_{3T}=u_{a\theta}.\notag
\end{eqnarray}
By part (iv) in Assumption \ref{assu:static} and the fact that $\chi_T\in (0,1)$, the previous denominators never vanish, so the equilibrium indeed always exists.\footnote{If both players are myopic, the strategies carry coefficients with $\chi_t$ as above, $t\in[0,T]$. By the arguments used for Lemma \ref{lem:UniqueLearningSol}, the induced system in $(\gamma,\chi)$ has a unique solution, so there is a unique LME for all $T>0$.} Moreover, the terminal signaling coefficient $\alpha_{3T} = \beta_{3T}+\beta_{1T}\chi_T$ is proportional to $u_{a\theta}+u_{a\hat a}\hat u_{\hat a \theta}\chi_T$, which, by part (ii) in Assumption \ref{assu:static}, never vanishes either. This latter property is sufficient for the dynamic equilibria that we construct to always exhibit nontrivial signaling throughout the game.\footnote{This requirement at time $T$ can be relaxed, but it is beyond our scope of interest. Also, it is easy to see that $\delta_{0T}=\hat u_0+\hat u_{\hat a a} \beta_{0T}$, $\delta_{1T}=\hat u_{\hat a\theta}+\hat u_{\hat a a}[\beta_{3T}+\beta_{1T}\chi_T]$ and $\delta_{2T}=\hat u_{\hat a a} [\beta_{2T}+\beta_{1T}(1-\chi_T)]$. }

We conclude that $\mathbf{b}:=(\beta_0,\beta_1,\beta_2,\beta_3,v_6,v_8,\gamma,\chi)^\mathsf{T}$ satisfies a BVP of the form
\begin{eqnarray}\label{eq:BVPgeneral}
\dot{\mathbf{b}_t}=\mathrm{f}(\mathbf{b}_t),\; \text{ s.t. }\; \mathrm{D}_0 \mathbf{b}_0+\mathrm{D}_T \mathbf{b}_T = (\mathbf{B}(\chi_T)^\mathsf{T},\gamma^o,0)^\mathsf{T}
\end{eqnarray}

where $\mathrm f: \mathbb{R}^6\times \mathbb{R}_+\times [0,1)\to\mathbb{R}^8$; $\mathrm{D_0}:=\diag(0,0,0,0,0,0,1,1)$ and $\mathrm{D_T}:=\diag(1,1,1,1,1,1,0,0)$ are diagonal matrices;
and $\mathbf{B}(\chi):[0,1]\to\mathbb{R}^6$ carries the terminal conditions via
\begin{eqnarray}\label{eq:B_terminal}
\mathbf{B}(\chi):=\left(\frac{u_0+u_{a\hat a}\hat u_0}{1-u_{a\hat a}\hat u_{\hat a a}}, \frac{u_{a\hat a}[u_{a\theta}\hat u_{\hat a a}+\hat u_{\hat a\theta}]}{1-u_{a\hat a}\hat u_{\hat a a}\chi},\frac{u_{a\hat a}^2\hat u_{\hat a a}[u_{a\theta}\hat u_{\hat a a}+\hat u_{\hat a\theta}](1-\chi)}{(1-u_{a\hat a}\hat u_{\hat a a})(1-u_{a\hat a}\hat u_{\hat a a}\chi)},u_{a\theta},0,0\right)^\mathsf{T} \in \mathbb{R}^6.
\end{eqnarray}
The general expression that $\mathrm f(\cdot)$ in \eqref{eq:BVPgeneral} takes for a generic pair $(U,\hat U)$ satisfying Assumption \ref{assu:static} is long, and can be found in \texttt{spm.nb} on our websites. (There, to simplify notation, we work with normalized payoffs $U/|\partial^2 U/\partial a^2|$ and $\hat U/|\partial^2 \hat U/\partial \hat a^2|$.) In the next subsection, we provide examples that exhibit all the relevant properties that any such $\mathrm f(\cdot)$ can satisfy.

The task of finding an LME is then reduced to solving the BVP \eqref{eq:BVPgeneral} (and checking that the rest of the coefficients in the value function are well defined, which is a simpler task).

\subsection{Existence of Linear Markov Perfect Equilibria\label{subsec:existence_interior}}

In this section, we present two existence results for LME. Behind these results are two approaches to separately address \emph{common} and \emph{private-value} environments, as the corresponding BVPs have a different structure linked to the extent of asymmetry between the players' signaling rates that arise in each case. For the sake of exposition, we state the theorems for variations of the leadership game of Section \ref{sec:leading_by_example}. The ``common-value" method is, nevertheless, fully general and can be exported to other asymmetric settings.

\vspace{-0.1in}
\paragraph{The shooting problem.} Establishing the existence of a solution to the BVP \eqref{eq:BVPgeneral} is complex because there are multiple ODEs in both directions: $(\vec \beta,v_6,v_8)$ is traced backward from its terminal values, while $(\gamma,\chi)$ is traced forward using its initial values---see Figure \ref{fig:BVP1}. This means that some notion of ``shooting'' must be applied: say, to construct a backward \emph{initial value problem} (IVP) in which $(\gamma,\chi)$ has a parametrized initial condition at $T$, and be able to ensure that the terminal value (now, at 0) exactly matches $(\gamma^o,0)$. Attempting to use traditional continuity arguments widely used for one-dimensional problems---i.e., tracing the initial condition over an interval so that the target is hit by continuity---is hopeless:  accurate knowledge of the relationship between $\gamma$ and $\chi$ at $T$ for all possible coefficients $\vec\beta$ would be required to find the right ``tracing path'' in a multidimensional domain. 
\begin{figure}[htbp]
  \centering
  \includegraphics[keepaspectratio, height=1.5in]{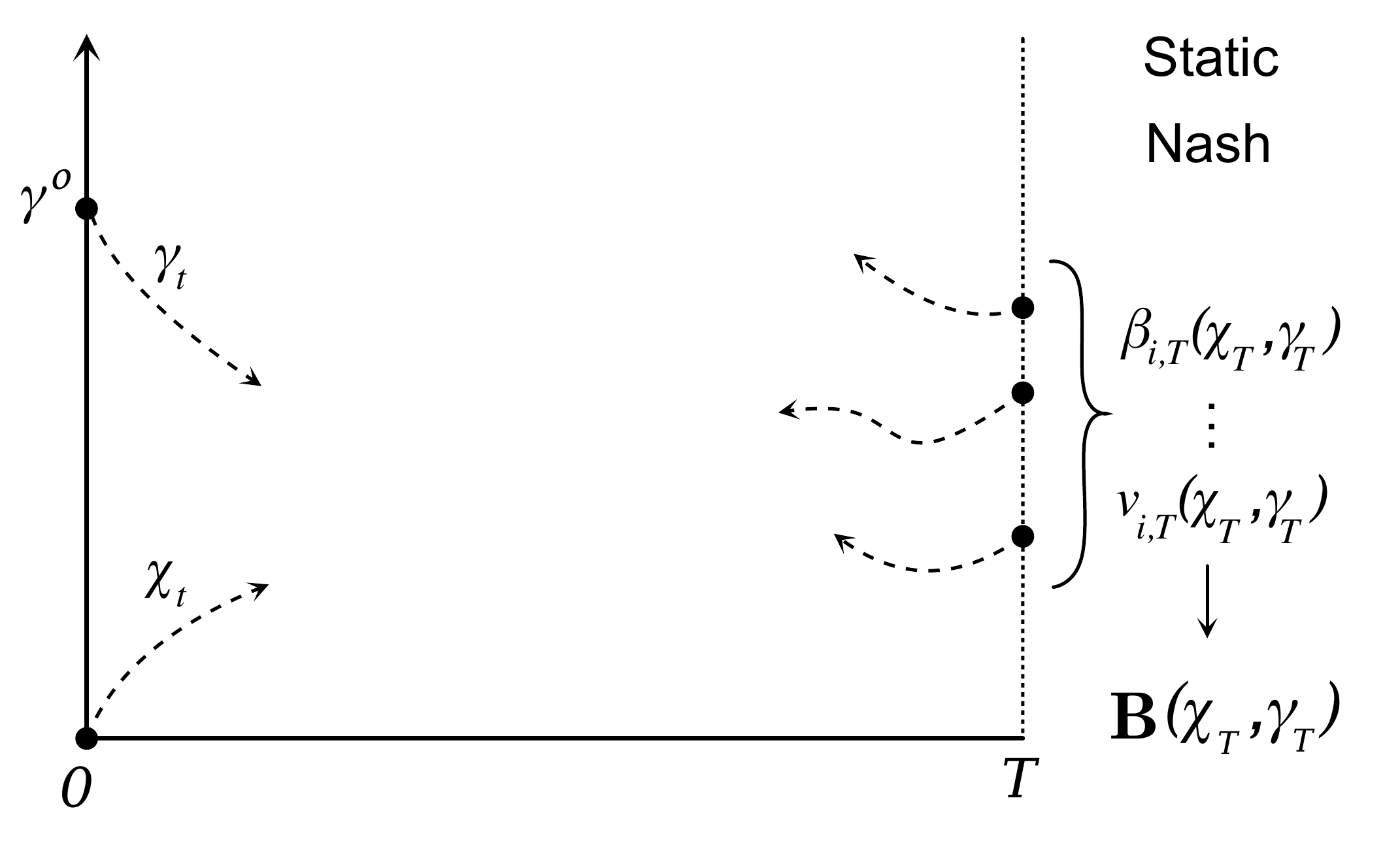}
  \vspace{-0.1in}
  \caption{The tuple $(\gamma,\chi)$ has initial conditions, while $(\vec\beta,v_6,v_8)$ has terminal ones. We have allowed for non-zero $v$'s and for a dependence on $\gamma_T$, as this can occur if terminal payoffs are allowed.}\label{fig:BVP1}
\end{figure}

\vspace{-0.12in}
The reason behind this dimensionality problem is the asymmetry in the environment: the rate at which the long-run player signals, $\alpha_3:= \beta_3+\beta_1\chi$, can be very different from the myopic player's counterpart, $\delta_1$. When this is the case, a nontrivial history dependence between $\gamma$ and $\chi$---reflected in the coupled system of ODEs they satisfy---ensues. Two questions naturally arise: first, under what conditions can such history dependence be simplified; second, how can one tackle the issue of existence of LME when a simplification is not feasible?

\vspace{-0.1in}
\paragraph{Private values: one-dimensional shooting.} We say that the environment is one of \emph{private values} if the myopic player's flow utility satisfies $\hat u_{\hat a\theta}=0$, i.e., his best reply does not directly depend on his belief about $\theta$, but only indirectly via the long-run player's action. Otherwise, the setting is one of \emph{common values} (despite the long-run player knowing $\theta$).

In a private-value setting, the players signal to each other at rates that are proportional. Indeed, the weight attached to $\hat M$ in the myopic player's best response becomes $\delta_1=\hat u_{\hat a a}\alpha_3$.

\begin{lemma}[One-to-one mapping]\label{lem:GammaChiRelationship}
Suppose that $\beta_1$ and $\beta_3$ are continuous and that $\delta_1=\hat u_{\hat a a}\alpha_3$. If $\hat u_{\hat a a}\neq 0$, there are positive constants $c_1, c_2$ and $d$ independent of $\gamma^o$ such that
$$
\chi_t= \frac{c_1c_2(1-[\gamma_t/\gamma^o]^d)}{c_1+c_2[\gamma_t/\gamma^o]^d}.
$$
Moreover, (i) $0\leq \chi_t< c_2<1$ for all $t\in [0,T]$ and (ii) $c_2\to 0$ as $\sigma_X\to0$ and $c_2\to 1$ as $\sigma_X\to\infty$. If instead $\hat u_{\hat a a}=0$ or $\sigma_X=\infty$, $\chi_t=1-\gamma_t/\gamma^o$.

\end{lemma}

Private-value settings, by inducing proportional signaling rates, create useful symmetry: while the players' posterior variances are not proportional, there is a decreasing relationship between $\chi$ and $\gamma$ at all times. By (i), $\chi$ is uniformly below 1 if the public signal is informative, reflecting that the scope for the history-inference effect falls relative to the no-feedback case. By (ii), the public and no-feedback cases are recovered as we take limits; further, the characterization of $\chi$ obtained in the latter case is recovered when, in addition, $\hat u_{\hat a a}=0$.

This result enables the use of standard one-dimensional shooting arguments, making our leadership application of Section \ref{sec:leading_by_example} a valid laboratory:  the game has private values because the follower wants to match the leader's action. Below is the corresponding BVP for $\sigma_X\in (0,\infty)$ and, for simplicity, for $r=0$; we omit the ODE for $\beta_0$ because it is uncoupled and linear:
\begin{eqnarray*}
\dot{v}_{6t}&=&\beta_{2t}^2+2\beta_{1t}\beta_{2t}(1-\chi_t)-\beta_{1t}^2(1-\chi_t)^2+\frac{2  v_{6t}\alpha_{3t}^2\gamma_t\chi_t}{\sigma_X^2(1-\chi_t)}\label{eq:leadexampleBVP1}\\
\dot{v}_{8t}&=&-2\beta_{2t}-2(1-2\alpha_{3t})\beta_{1t}(1-\chi_t)-4\beta_{1t}^2\chi_t(1-\chi_t)+\frac{ v_{8t}\alpha_{3t}^2\gamma_t\chi_t}{\sigma_X^2(1-\chi_t)}\label{eq:leadexampleBVP2}\\
\dot{\beta}_{1t}&=&\frac{\alpha_{3t}\gamma_t}{2\sigma_X^2\sigma_Y^2(1-\chi_t)}\left\lbrace 2\sigma_X^2(\alpha_{3t}-\beta_{1t})\beta_{1t}(1-\chi_t)-\alpha_{3t}^2\beta_{1t}\gamma_t\chi_t v_{8t}\right.\notag\\
&& \left.-2\sigma_Y^2 \alpha_{3t}\chi_t(\beta_{2t}-\beta_{1t}[1-\chi_t-2\beta_{2t}\chi_t])\right\rbrace\label{eq:leadexampleBVP3}\\
\dot{\beta}_{2t}&=&\frac{\alpha_{3t}\gamma_t}{2\sigma_X^2\sigma_Y^2(1-\chi_t)}\left\lbrace 2\sigma_X^2\beta_{1t}^2(1-\chi_t)^2+2\sigma_Y^2 \alpha_{3t}\beta_{2t}\chi_t^2 (1-2\beta_{2t})-\alpha_{3t}^2\gamma_t\chi_t (2 v_{6t}+\beta_{2t} v_{8t})\right\rbrace \label{eq:leadexampleBVP4}\\
\dot{\beta}_{3t}&=&\frac{\alpha_{3t}\gamma_t}{2\sigma_X^2\sigma_Y^2(1-\chi_t)}\left\lbrace -2\sigma_X^2 \beta_{1t}(1-\chi_t)\beta_{3t}+2\sigma_Y^2\alpha_{3t}\beta_{2t}\chi_t^2(1-2\beta_{3t})-\alpha_{3t}^2\beta_{3t}\gamma_t\chi_t v_{8t}\right\rbrace \label{eq:leadexampleBVP5}\\
\dot{\gamma}_t&=&-\frac{\gamma_t^2\alpha_{3t}^2}{\sigma_Y^2},\label{eq:leadexampleBVP6}
\end{eqnarray*}

with boundary conditions $v_{6T}=v_{8T}=0$, $\beta_{1T}=\frac{1}{2(2-\chi_T)}$, $\beta_{2T}=\frac{1-\chi_T}{2(2-\chi_T)}$, $\beta_{3T}=\frac{1}{2}$ and $\gamma_{0}=\gamma^o$, and where $\alpha_3:=\beta_3+\beta_1\chi$ and $\chi_t$ is as in the previous lemma. We have the following:

\begin{theorem}[Existence of LME---private values]\label{thm:LeadingInteriorExistence}
Let $\sigma_X\in(0,\infty)$ and $r=0$. There exists a strictly positive $T(\gamma^o)\in O(1/\gamma^o)$ such that, for all $T<T(\gamma^o)$, there is an LME based on the solution to the previous BVP that satisfies $\beta_{0t} = 0$, $\beta_{1t}+\beta_{2t}+\beta_{3t}=1\; \text{ and }\; \alpha_{3t}>0,\;  t\in [0,T].$ 
\end{theorem}

The key step in the proof is to show that $(\beta_1,\beta_2,\beta_3,v_6,v_8,\gamma)$ can be bounded uniformly over $[0,T(\gamma^o))$, for some $T(\gamma^o)>0$, when $\gamma_t\in [0,\gamma^o]$ at all times. This implies that tracing the (parametrized) initial condition of $\gamma$ in the backward IVP from $0$ upwards will lead to at least one $\gamma$-path landing at $\gamma^o$ due to the continuity of the solutions with respect to the initial conditions, while the rest of the ODEs still admit solutions.

Figure \ref{fig:LeadingSignalingComp} below illustrates the signaling coefficient $\alpha_3$ for various values of $\sigma_X$: as the latter increases, the dashed lines rotate counterclockwise from the public to the no-feedback case, justifying our earlier focus on $\sigma_X\in\{0,\infty\}$. Interestingly, when $r>0$ and $\sigma_X<\infty$, $\alpha_3$ is \emph{nonmonotonic}. Intuitively, a partially informative signal combines the increasing history-inference effect of the no-feedback case with the decreasing signaling motive driving the public case. Discounting weakens the latter, while the former grows over time even with a myopic leader. As $\sigma_X$ increases, moreover, the history-inference effect gains strength and the maximum of $\alpha_3$ shifts to the right. Only a fully dynamic model can uncover such effects.\footnote{Existence in the discounted case can be shown with identical methods. For sharper visual effects, we are potentially plotting beyond the interval of existence ensured by the theorem (which is a crude lower bound).}\\

\vspace{-0.1in}
\begin{figure}[htbp]
\begin{subfigure}{.5\linewidth}
\centering
\includegraphics[height=1.65in]{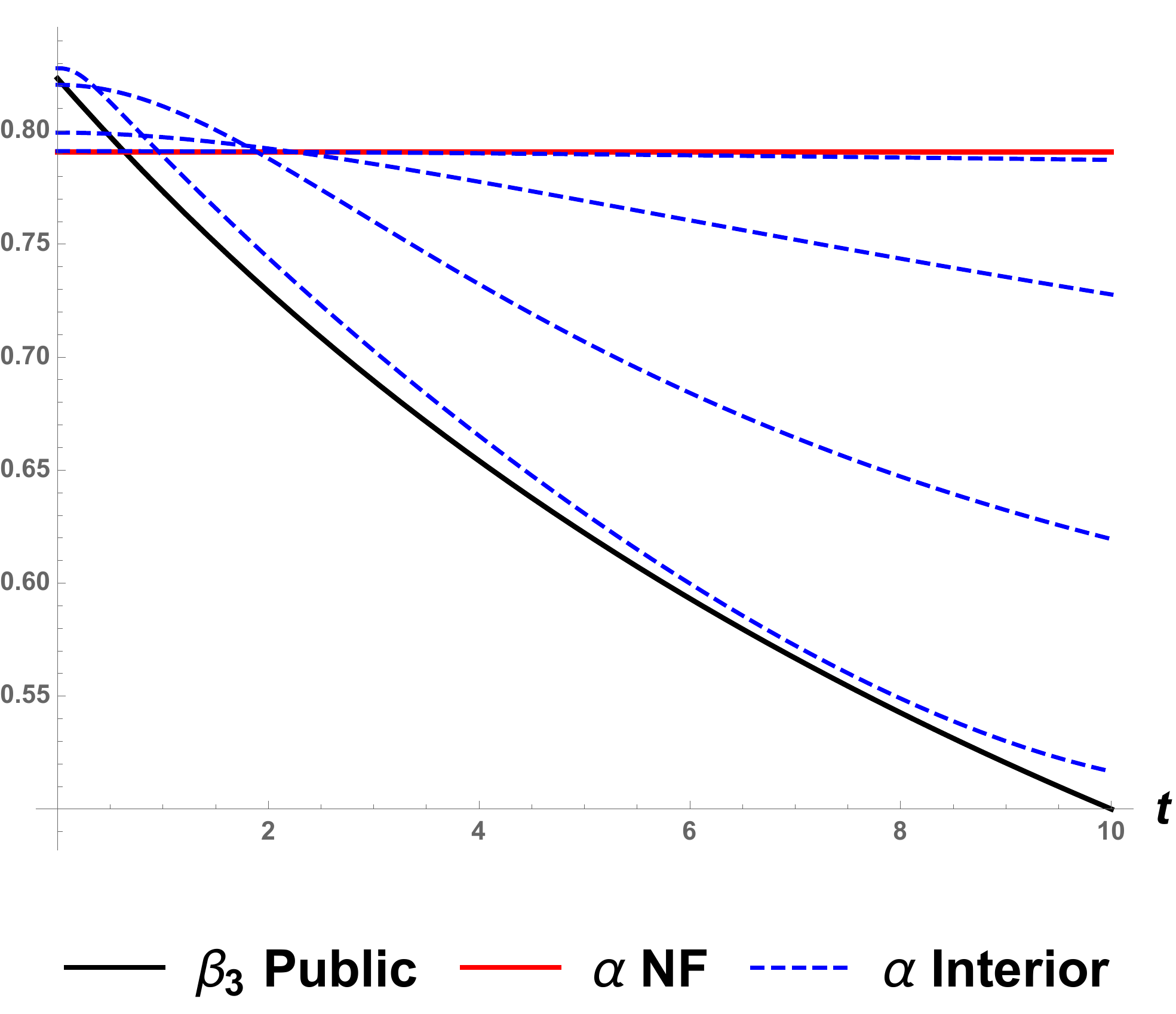}
\caption{$r=0$}
\end{subfigure}%
\begin{subfigure}{.5\linewidth}
\centering
\includegraphics[height=1.65in]{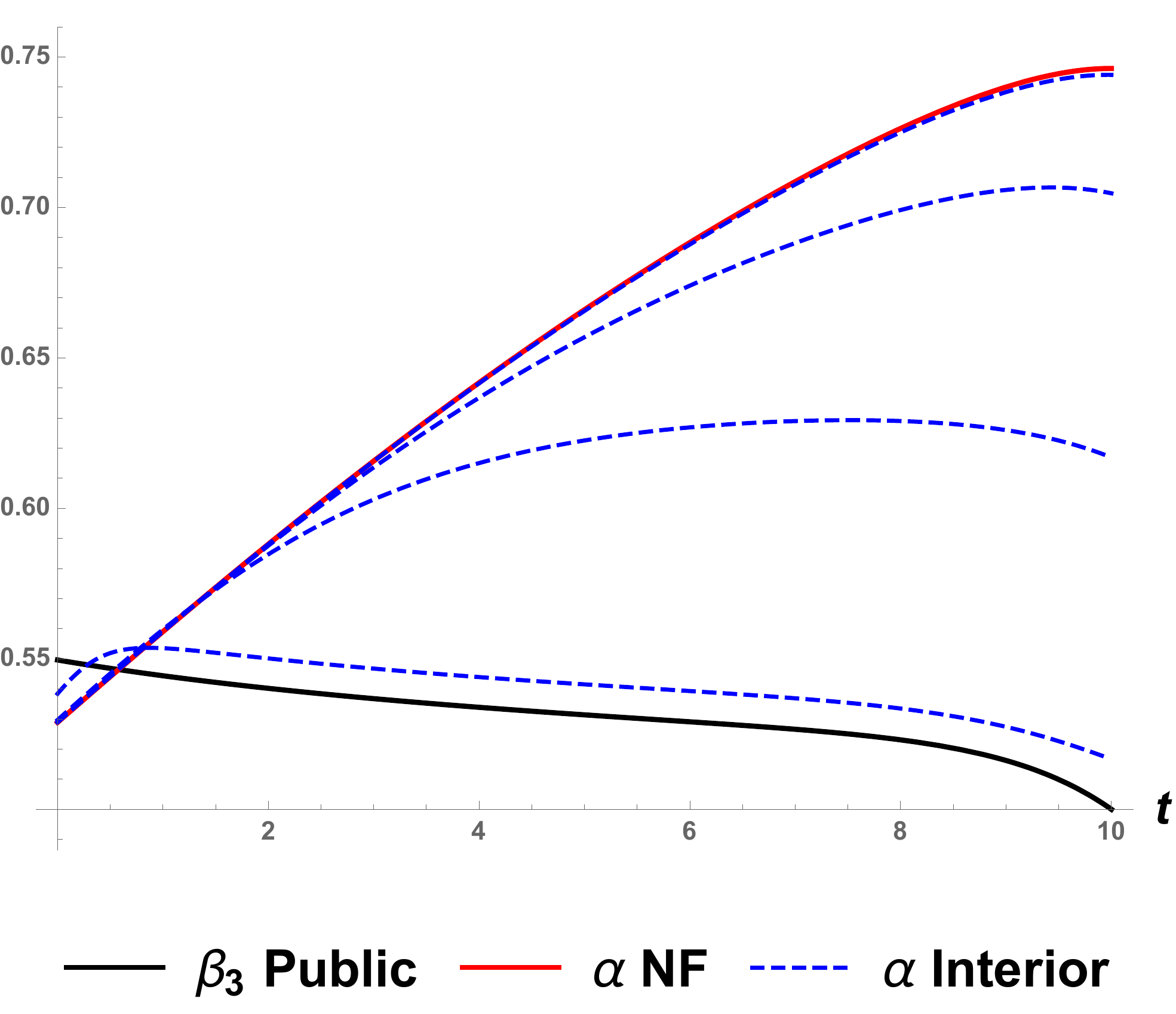}
\caption{$r=1$}
\end{subfigure}
  \caption{Signaling coefficients for $\sigma_X\in \{0,.1,.75,2,10,+\infty\}$; ``$\alpha$ Interior" denotes $\alpha_3$. }\label{fig:LeadingSignalingComp}
\end{figure}

Theorem 1 can be easily generalized to accommodate a conflict of interest between the players. Indeed, the one-to-one mapping between $\gamma$ and $\chi$ still holds for any best response of the myopic player that is a \emph{time-independent affine function} of his expectation of the long-run player's action. We defer a discussion of this topic to the common-value setting, where we address asymmetries at a more general level.

\paragraph{Common-value settings: fixed-point methods.} When $\alpha_3$ and $\delta_1$ are not proportional, $\chi$ can depend on both current and past values of $\gamma$---the dimensionality problem resurfaces. 

Our key observation is that finding a solution to any given instance of the BVP \eqref{eq:BVPgeneral} is, mathematically, a \emph{fixed-point problem}. Specifically, note that the static Nash equilibrium at time $T$ depends on the value that $\chi$ takes at that point. The latter value, however, depends on how much signaling has taken place along the way, i.e., on values of the coefficients $\vec\beta$ at times prior to $T$. Those values, in turn, depend on the value of the equilibrium coefficients at $T$ by backward induction---thus, we are back to the same point where we started.

Our approach therefore applies a fixed-point argument adapted from the literature on BVPs with \emph{intertemporal linear constraints} \citep{keller1968numerical} to our problem with \emph{intratemporal nonlinear constraints}. Because the method is novel and has the generality required to become useful in other settings, we briefly elaborate on how it works.

In essence, we will be ``shooting'' six ODEs \emph{forward}. Specifically, let $t\mapsto \mathbf{b}_t(s,\gamma^o,0)$ denote a solution to the forward IVP version of \eqref{eq:BVPgeneral} when the initial condition is $(s,\gamma^o,0)$, with $s\in \mathbb{R}^6$ parametrizing the initial value of $(\vec\beta, v_6,v_8)$. From Lemma \ref{lem:UniqueLearningSol}, the last two components of $\mathbf{b}$, i.e., $\gamma$ and $\chi$, admit solutions as long as the others do; moreover, there are no constraints on their terminal values. Thus, for our fixed-point argument, we can focus on the first six components in $\mathbf{b}:=(\beta_0,\beta_1,\beta_2,\beta_3,v_6,v_8,\gamma,\chi)^\mathsf{T}$ by defining the \emph{gap} function 
$$
g(s)= \mathbf{B}(\chi_T(s,\gamma^o,0)) - \mathrm{D}_T\int_0^T \mathrm{f}(\mathbf{b}_t(s,\gamma^o,0))dt,
$$
where $\mathrm{D}_T:=\diag(1,1,1,1,1,1,0,0)$. This function measures the distance between the total growth of $(\vec\beta,v_6,v_8)$ (last term), and its target value, $\mathbf{B}(\chi_T(s,\gamma^o,0))$. By \eqref{eq:B_terminal}, $\mathbf{B}(\chi)$ is \emph{nonlinear}: the static equilibrium imposes nonlinear relationships across variables at time $T$.

By definition, $\mathbf{b}_0(s,\gamma^o,0)=s$. Consequently, it follows that
$$
g(s)=s\; \iff\; \mathbf{B}(\chi_T(s,\gamma^o,0)) = s+ \mathrm{D}_T\int_0^T \mathrm{f}(\mathbf{b}_t(s,\gamma^o,0))dt = \mathrm{D}_T\mathbf{b}_T(s,\gamma^o,0),
$$
where the last equality follows from the definition of the ODE-system that $\mathrm{D}_T\mathbf{b}$ satisfies. Thus, the shooting problem of finding $s\in \mathbb{R}^6$ such that $\mathbf{B}(\chi_T(s,\gamma^o,0))=\mathrm{D}_T\mathbf{b}_T(s,\gamma^o,0)$ can be restated as one of finding a fixed point of the function $g$.\footnote{A BVP with intertemporal linear constraints \citep{keller1968numerical} differs from ours in that $\mathrm{D}_0 \mathbf{b}_0+\mathrm{D}_T \mathbf{b}_T = (\mathbf{B}(\chi_T)^\mathsf{T},\gamma^o,0)^{\mathsf{T}}$ becomes $A\mathbf{b_0}+B \mathbf{b_T} = \zeta$, where $\zeta$ is a constant column vector and $A$ and $B$ are general matrices. On the one hand, since $A$ and $B$ are not necessarily diagonal matrices, one may not be able to dispense with a subset of the system. On the other hand, our version of $\zeta$ is a nonlinear function of a subset of (endogenous) components of $\mathbf{b}_T$, which makes the fixed-point argument more involved.}

The goal is then to find a time $T(\gamma^o)$ and a compact set $\mathcal{S}$ such that (i) for all $s\in \mathcal{S}$, a unique solution to the aforementioned IVP over $[0,T(\gamma^o)]$ exists, and (ii) $g$ is continuous from $\mathcal{S}$ to itself. The natural choice for $\mathcal{S}$ is a ball with center $s_0:=\mathbf{B}(0)$, the terminal condition of the trivial game with $T=0$; we then apply Brouwer's fixed-point theorem.

We can now establish our main existence result for a variation of the leadership application in which the follower's best response is of the form $\hat a_t = \hat u_{\hat a \theta} \hat{\mathbb{E}}_t[\theta] +\hat u_{\hat a a}\hat{\mathbb{E}}_t[a_t]$ for $(\hat u_{\hat a \theta},\hat u_{\hat a a})$ as in Assumption \ref{assu:static}; in particular, the myopic player's signaling coefficient is $\delta_{1t} = \hat u_{\hat a\theta}+\hat u_{\hat a a}\alpha_{3t}$.\footnote{Since $u_{a\hat a}=u_{a\theta}=1/2$ for the leader, signaling is nontrivial if $\hat u_{\hat a \theta}>-1$; best responses intersect if $\hat u_{\hat a a}<2$; and second-order inferences arise if $(\hat u_{\hat a \theta},\hat u_{\hat a a})\neq (0,0)$. Intercepts can be easily incorporated.} The BVP is stated in \eqref{eq:LeadingIntAsymChangev6}-\eqref{eq:LeadingIntAsymChi} in the Appendix.

\begin{theorem}[Existence of LME---common values]\label{thm:Leading_LME_BVP}
Set $\sigma_X\in (0,\infty)$ and $r=0$ in the leadership model, and let $(\hat u_{\hat a \theta},\hat u_{\hat a a})\in \mathbb{R}^2$ satisfy Assumption \ref{assu:static}. There is a strictly positive function $T(\gamma^o)\in O(1/\gamma^o)$ such that if $T<T(\gamma^o)$, there exists an LME based on the BVP \eqref{eq:LeadingIntAsymChangev6}-\eqref{eq:LeadingIntAsymChi}. In such an equilibrium, $\alpha_{3}>0$.
\end{theorem}

There are three observations from this theorem. First, the time for which an LME is ensured to exist grows without bound as $\gamma^o\searrow 0$. Indeed, $\mathrm{f}(\cdot)$ naturally scales with this parameter, so the solutions converge to the full-information benchmark $(\beta_0,\beta_1,\beta_2,\beta_3,v_6,v_8,\chi,\gamma)=(0,u_{a\hat a}[u_{a\theta}\hat u_{\hat a a}+\hat u_{\hat a\theta}],\frac{u_{a\hat a}^2\hat u_{\hat a a}[u_{a\theta}\hat u_{\hat a a}+\hat u_{\hat a\theta}]}{1-u_{a\hat a}\hat u_{\hat a a}},u_{a\theta},0,0,0,0)$,
which is defined for all $T>0$. 

Second, while the self-map condition does not affect the order of $T(\gamma^o)$ relative to Theorem \ref{thm:LeadingInteriorExistence}, it is not vacuous either. In fact, since $s_0=\mathbf{B}(0)$ is the center of $\mathcal{S}$, we have that
$$
g(s)-s_0= \mathbf{B}(\chi_T(s,\gamma^o,0)) - \mathbf{B}(0) - \mathrm{D}_T\int_0^T \mathrm{f}(\mathbf{b}_t(s,\gamma^o,0))dt.
$$
Bounding $\mathbf{B}(\chi_T(s,\gamma^o,0)) - \mathbf{B}(0)$ therefore imposes an additional constraint relative to those that ensure that the system is uniformly bounded (which in turn bound the last integral). 

Finally, the bound $T(\gamma^o)$ is obtained under minimal knowledge of the system: it is the outcome of bounds that only exploit the degree of the polynomials in $\mathrm{f}(\mathbf{b})$ and hence that do not exploit any relationship between the equilibrium coefficients. Thus, the proof technique is (i) fully general and (ii) improvable provided more is known about the system at hand.

Appendix \ref{appendix_sketch} sketches how the proof of Theorem \ref{thm:Leading_LME_BVP} applies to the whole class of games satisfying Assumption \ref{assu:static}. Moreover, observe that this method, by being able to handle multiple ODEs in each direction, has the power to be applied to other asymmetric games of learning beyond the class under study (see the concluding remarks for more on this topic).

\vspace{-0.1in}
\paragraph{Asymmetric games.} Let us briefly elaborate on how Theorems \ref{thm:LeadingInteriorExistence} and \ref{thm:Leading_LME_BVP} enable us to explore natural settings in which the players rates of signaling are inevitably different.

The simplest common-value version of the leadership game is one where the follower only wants to match the state of the world. In this case, the equilibrium obtained in Proposition \ref{prop:Leading_Public_LME} for the public case still goes through, but the equilibrium for $\sigma_X>0$ changes. Specifically, since the follower now behaves according to $\hat a= \hat M$ as opposed to $\hat a= \alpha_3 \hat M$, his actions are more sensitive to his private information. This, in turn, magnifies the leader's signaling in two ways: the leader has a stronger incentive to steer the follower's behavior (i.e., $\beta_3$ increases), and due to the imperfect learning, the leader relies more on $M$ (at the expense of $L$) to coordinate (i.e., $\beta_1$ also increases).  This results in more signaling and learning, also compounded by an overall higher $\chi$ in the history-inference effect (despite the negative direct impact that a more informative $X$ has on the reliance on past play); the left and center panels in Figure \ref{fig:asy} illustrate these forces. It is noteworthy that this example must utilize our most  general Theorem \ref{thm:Leading_LME_BVP}, in spite of its a priori simplicity stemming from the follower's myopic best reply being independent of the leader's strategy.

\begin{figure}[htbp]
\centering
\includegraphics[height=1.5in]{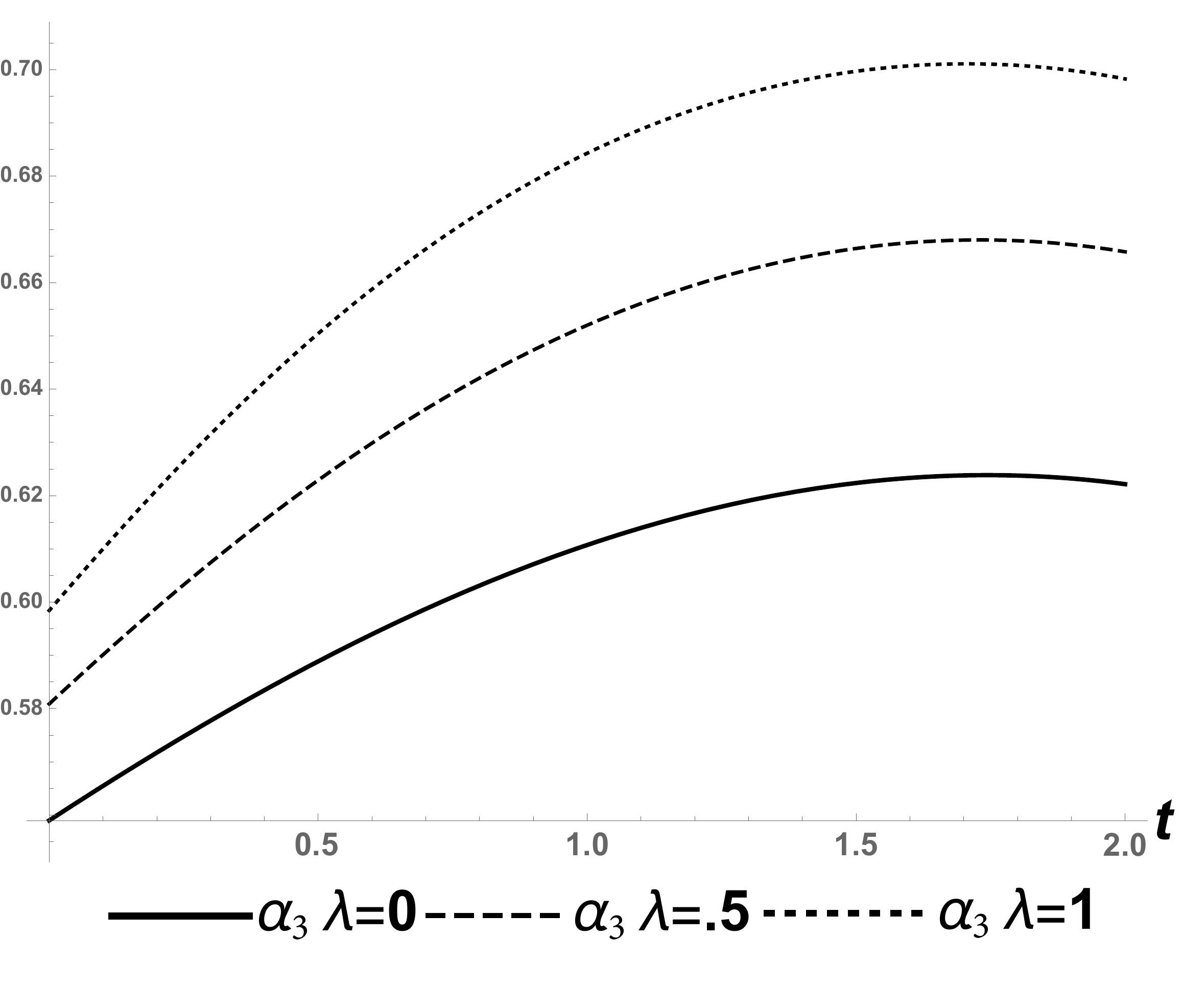}
\quad \includegraphics[height=1.5in]{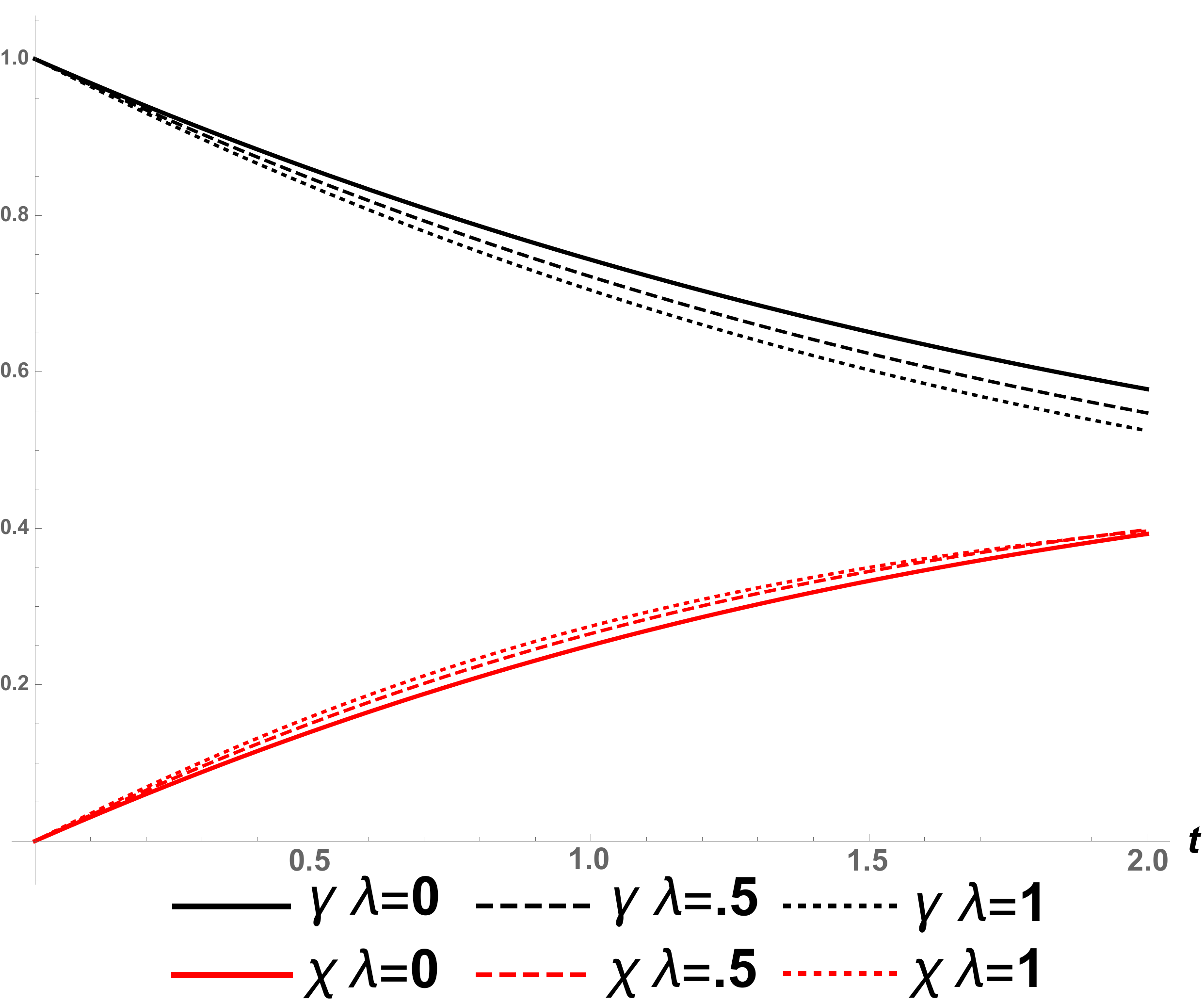}
  \quad \includegraphics[height=1.5in]{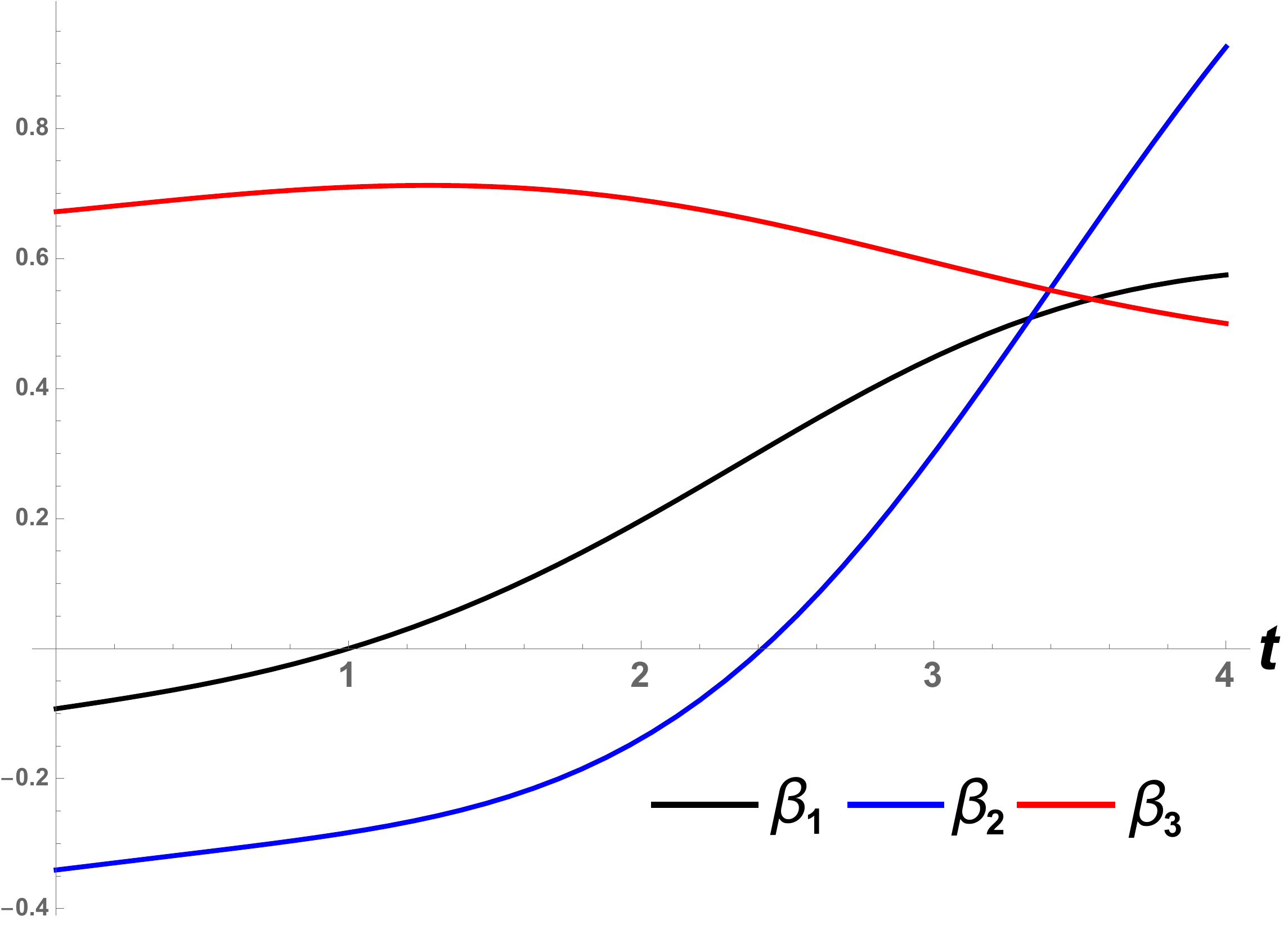}
\caption{Signaling (left) and learning (center) for $\hat U(a,\hat a,\theta)=-\lambda(\hat a-\theta)^2-(1-\lambda)(\hat a-a)^2$, $\lambda\in \{0, 0.5,1\}$. Right: Leader's strategy coefficients for follower payoffs $-\left(\hat a-(3/2)a\right)^2$. }\label{fig:asy}
\end{figure}

Asymmetries also naturally arise from a conflict of interest between the players. Let us now use Theorem \ref{thm:LeadingInteriorExistence}. In the leadership application, fixing the leader's payoffs, incentives are misaligned if $\hat u_{\hat a a}+\hat u_{\hat a \theta}\neq 1$ on the follower's side. Suppose then that $\hat u_0=\hat u_{\hat a \theta}=0$ and $\hat u_{\hat a a}>1$, i.e., the follower overreacts to the leader's action. If the horizon is sufficiently long, the leader has an initial incentive to invest in mitigating the follower's reaction by shrinking the latter's belief, so $\beta_1$ and $\beta_2$ (the weights on $M$ and $L$, respectively) are negative. Furthermore, the right panel in Figure \ref{fig:asy} shows that $\beta_2$ is the main component in this attempt: the manipulation occurs largely via the leader \emph{jamming} the public belief $L$, as her direct incentives to coordinate soon are strong. Finally, as the time to enjoy the benefits of such manipulation shrinks, the leader accommodates the follower, and these effects reverse. 

We have chosen to continue with the leadership application for expositional reasons. In the next section, we explore other applications based on extensions of our baseline model.

\section{Extensions\label{sec:applications}}

We extend our model to allow (i) a quadratic terminal payoff in a career-concerns model, and (ii) the long-run player affecting the public signal in a trading model a la \cite{kyle1985continuous}.

\subsection{Reputation for Neutrality\label{subsec:political}}

Suppose that the long-run player is now an expert or politician with career concerns. This agent has a hidden ideological bias $\theta$ and takes repeated actions---for example, adopting positions on critical issues or making campaign promises. The mean of the prior distribution denotes the unbiased type---without loss, let us use the normalization $\mu=0$.

We interpret the myopic player as a news outlet that always attempts to report on the true bias, i.e., that maximizes $-\hat{\mathbb{E}}_t[(\hat{a}_t-\theta)^2]$ at all times. In turn, the politician's payoff is
\begin{align*}
-\int_0^T (a_t-\theta)^2 dt-\psi \hat{a}_T^2,
\end{align*}
with $\psi>0$. Given the myopic player's preferences, the termination payoff takes the form $-\psi \hat M_T^2$, and so the politician has career concerns: she wants to appear as unbiased at the end of the horizon. But this long-term goal conflicts with her short-term ideological desires: in her flow payoff, she benefits from taking actions that conform to her bias.

The private nature of $dY=a_tdt+\sigma_Y dZ_t^Y$ is understood as the outlet having access to imperfect private sources regarding the politician's actions. In turn, $dX_t=\hat M_tdt+\sigma_X dZ_t^X$ is the outlet's \emph{news} process: the (public) reporting on the bias is fair on average, but imperfect. When does the politician fare better? In settings where the reporting is precise---i.e., low $\sigma_X$---and hence she can tailor her actions to her reputation? Clearly, noisier environments entail a direct cost: they introduce increased uncertainty over a concave objective. The next result shows that increasing an agent's uncertainty over her own reputation, thereby undermining her ability to take appropriate actions, can be beneficial:

\begin{proposition}\label{prop:PCCcomparison}
\begin{itemize}
\item [(i)] Suppose that $\sigma_X\in \{0,+\infty\}$. Then, for all $\psi, T>0$ there exists an LME. Moreover, if $\psi<\sigma_Y^2/\gamma^o$, the LME is unique, and learning is lower and ex ante payoffs higher in the no-feedback case.
\item [(ii)] If $\sigma_X\in (0,\infty)$, there exists $T(\gamma^o)\in O(1/\gamma^o)$ s.t. an LME exists for all $T<T(\gamma^o)$.
\end{itemize}
\end{proposition}

Politicians or experts with larger biases take more extreme actions, and hence the equilibrium strategy attaches a positive weight to the type. Because of career concerns, however, the greater the perceived value of $\hat{M}$, the greater the incentive to manipulate it downward. With private monitoring, higher types therefore must offset higher beliefs from their perspectives, leading to a history-inference effect that dampens the signaling coefficient $\alpha_3$. The belief is then less responsive from an ex ante perspective, which facilitates maintaining a reputation for neutrality.\footnote{It is easy to show that the ex ante expectation of $\hat{M}_T^2$ is $\gamma^o-\gamma_T$, so that greater learning by the myopic player results in larger terminal losses for the long-run player. This reverses for slightly \textit{negative} $\psi$, but so does the history-inference effect: there is \textit{more} learning but again a higher payoff in the no-feedback case.}  Indeed, provided that the objective is not too concave and the environment not too uncertain (which strengthen the direct cost), this strategic effect dominates. 

Regarding (ii), because this environment is one of common values, one can establish the existence of an LME with minimal changes to the method behind Theorem \ref{thm:Leading_LME_BVP}. Indeed, the only difference is that our baseline BVP changes to incorporate terminal conditions that depend not only on $\chi_T$, but also on $\gamma_T$ via $\beta_{1T}=-\frac{\psi \gamma_T}{\sigma_Y^2+\psi \gamma_T\chi_T}$: with terminal lump-sum payoffs, there are last-minute incentives to manipulate the myopic player's belief that decrease in the associated precision. Our approach does not vary with this  dependence.

\subsection{Insider Trading\label{subsec:insidertrading}}
An asset with fixed fundamental value $\theta$ is traded in continuous time until date $T$, the time at which its true value is revealed, ending the game. A patient \emph{insider} (the long-run player) privately observes $\theta$ prior to the start of the game. As in \cite{yang2019back}, a second trader has a technology that allows him to privately observe imperfect signals of the insider's trades; this player is myopic. Both players and a flow of noise traders submit orders to a \emph{market maker} who then executes those trades at a public price $L_t=\mathbb{E}[\theta|\mathcal{F}_t^X]$.

We depart from the baseline model along three dimensions. First, the public signal---the \emph{total order flow}---is  $dX_t=(a_t+\hat{a}_t)dt+\sigma_X dZ_t^X$, which now includes the long-run player's action; hence, the myopic player learns from both the private monitoring channel and the public price. Second, the players' flow payoffs depend directly on $L$, interpreted as the action taken by the market maker: the myopic player's flow payoff is given by $\xi(\theta-L)\hat{a}-\frac{\hat{a}^2}{2}$, where $\xi\geq 0$, while the long-run player's flow payoff is $(\theta-L_t)a_t$; the inverse of the parameter $\xi$ is a measure of transaction costs for the myopic player. Finally, observe that the long-run player's flow payoff is linear in her action $a_t$ at all instants $t\in [0,T]$. 

Following the literature, we seek an equilibrium in which the informed trader reveals her private information gradually over time through a linear strategy of the form \eqref{eq:P1StrategyGeneral}. Hence, we require that the coefficients of the insider's strategy be $C^1$ functions over strict compact subsets of $[0,T)$; we can then apply Lemmas \ref{lem:BeliefDecomp} and \ref{lem:UniqueLearningSol} to such sets.\footnote{The $C^1$ requirement suffices for the total order to be ``inverted" from the price for $t<T$  (hence, it is without loss to make $X$ the source of public information), while the open interval allows for the possibility of full revelation of information by time $T$.  The proof of Lemma \ref{lem:BeliefDecomp} derives learning ODEs for an additive drift in $X$, and it is easy to see that the steps of Lemma \ref{lem:UniqueLearningSol} (with $\hat{u}_{\hat a\theta}=\xi, \hat{u}_{\hat a a}=0$) go through for this case.}

Clearly, when $\xi=0$ (or $\sigma_Y=\infty$), the model reduces to  \cite{kyle1985continuous}, and hence an LME with trading strategy of the form $\beta_3(\theta-L)$ always exists. This is not the case when $\xi>0$.

\begin{proposition}\label{prop:InsiderTradingNonexistence}
Fix $\xi>0$. For all $\sigma_Y>0$, there does not exist a linear Markov equilibrium.
\end{proposition}

With linear Markov strategies, the myopic player acquires private information about $\theta$ over time. Thus, the myopic player's own  repeated trades carry further information to the market maker, beyond that which the market maker learns from the insider alone. This introduces momentum into the price from the insider's perspective, measured by a term $\xi(m-l)$ in the drift of $L$. Future trades then become less attractive to the insider, thereby placing the insider in a race against herself that results in all her information being traded away in the first instant, regardless of the amount of noise in the private signal $Y$.\footnote{The linearity of the setting prevents a Nash equilibrium from being defined at $T$. Thus, our argument does not stem from a problem with a BVP but rather an impossibility of indifference for the long-run player.}

In a closely related result, \cite{yang2019back} show that a linear equilibrium can cease to exist in a two-period setting where a trader who only participates in the last round receives a sufficiently precise signal of an informed player's first-period trade; a mixed-strategy equilibrium emerges instead. More generally, the existence problem relates to how, with common information, an informed player's rush to trade depends on the number of trading opportunities. The analysis of \cite{foster1994strategic} is illuminating in this respect: in a setting with nested information structures, a better informed insider trades a commonly known piece of information first, exploiting her superior information only later. While there are important differences between our setups (in their model, the belief of the less informed player is always known to the more informed player, and the common source of information is exogenous) there is a unifying theme: once common information is created, there is pressure to trade quickly on it. Such pressure increases with the number of rounds ahead.\footnote{For a similar result in a symmetric setting, see  \cite{back2000imperfect}. Both the presence of a myopic counterpart and a quadratic trading cost for this player only strengthen our nonexistence result.}

\vspace{-0.1in}
\section{Conclusion}

We have examined an important departure from the vast literature on signaling games: the case of a sender who does not see the signals emanating from her actions. A complex ``beliefs about beliefs'' problem arises in this case and leads to a novel separation effect via a second-order belief channel. Our contributions---namely, constructing belief-dependent equilibria and quantifying the impact of this natural separation effect on outcomes, along with the necessary new methodologies introduced---are at the frontier of what is known in these settings.

We conclude with a discussion of our modeling and expositional choices. First, the signal structure that we employ is important in that it enables us to ``close'' the set of states at the second order. If instead the long-run player had a stochastic type, more states would be needed at the very least; and if both players had access to imperfect private signals, beliefs of even higher order would be payoff-relevant. While these are interesting exercises, a natural question is whether behavior truly relies on such considerably more complex strategies.

Second, a model with a forward-looking receiver is a tractable extension that requires no major conceptual changes. In fact, most of the results are derived for, or can be generalized to, continuous coefficients in the myopic player's strategy. Those coefficients would then satisfy ODEs capturing optimal dynamic behavior, but crucially (i) no additional states are needed, and (ii) the fixed-point argument is applicable to an enlarged boundary value problem. Such an extension, however, only brings an old known force to the analysis: since $X$ is public, a forward-looking receiver would exhibit a traditional signal-jamming motive.

Our choice of applications stems from their proximity to our informational assumptions, but others are also plausible: a deception game for business or military strategy arises in an asymmetric version of our coordination game in which the leader enjoys miscoordination; a leadership model to study encouragement effects is obtained when complementarities between the state of the world and aggregate effort are allowed; and trading models with quadratic trading costs that restore existence can shed light on an informed trader's behavior.

Finally, while stylized, the linear-quadratic-Gaussian class is of great value. First, it uncovers effects that are likely to be key in other, more nonlinear, settings: the history-inference effect coupled with the time effects arising from learning seem to exhaust the forces present when behavior depends on the payoff-relevant aspects of the histories. Second, it permits the development of methods that are exportable to other settings: the fixed-point method for BVPs, by handling multidimensional shooting, has the power to be taken to asymmetric oligopolies with private information, or to reputation models with multidimensional types.


\vspace{-0.1in}

\begin{appendices}
\numberwithin{equation}{section}
\numberwithin{lemma}{section}
\numberwithin{proposition}{section}
\numberwithin{theorem}{section}

\section{Proofs for Section \ref{sec:leading_by_example}\label{appendix_A}}






\paragraph{Preliminary results.} We state standard results on ODEs \citep{teschl2012ordinary} which we use in the proofs that follow. Let $f(t,x)$ be continuous from $[0,T]\times \mathbb{R}^n$ to $\mathbb{R}^n$, where $T>0$. 

- \underline{\emph{Peano's} Theorem (Theorem 2.19, p. 56)}: There exists $T'\in (0,T)$, such that there is at least one solution to the IVP $\dot{x}=f(t,x),\; x(0)=x_0$ over $t\in [0,T')$.

If, moreover, $f$ is locally Lipschitz continuous in $x$, uniformly in $t$, then: 

- \underline{\emph{Picard-Lindel{\"o}f} Theorem (Theorem 2.2, p. 38)}: For $(t_0,x_0)\in [0,T)\times \mathbb{R}^n$, there is an open interval $I$ over which the IVP $\dot{x}=f(t,x),\; x(t_0)=x_0$ admits a unique solution.

- \underline{\emph{Comparison} theorem (Theorem 1.3, p. 27)}: If $x(\cdot), y(\cdot)$ are differentiable, $x(t_0)\leq y(t_0)$ for some $t_0\in [0,T)$, and $\dot{x}_t-f(t,x(t))\leq \dot{y}_t-f(t,y(t))$ $\forall t\in [t_0,T)$, then $x(t)\leq y(t)$  $\forall t\in [t_0,T)$. If, moreover, $x(t)<y(t)$ for some $t\in [t_0,T)$, then $x(s)<y(s)$ $\forall s\in [t,T)$.

\vspace{-0.1in}
\subsection{Proofs for Public Case}
\begin{proof}[Proof of Proposition \ref{prop:Leading_Public_LME}.]
We aim to characterize an LME in which the leader backs out the follower's belief from his action at all times, with strategies of the form $
a_t=\beta_{0t}+\beta_{1t} \hat M_t +\beta_{3t}\theta$ and $\hat a_{t} =\hat{\mathbb{E}}_t[a_t]= \beta_{0t}+(\beta_{1t}+\beta_{3t})\hat M_t$, 
where $\hat M_t:=\hat{\mathbb{E}}_t[\theta]$, and $\beta_{it}$, $i=0,1,3$, are deterministic, satisfying $\beta_{1t}+\beta_{3t}\neq 0$, $t\in [0,T]$. From standard results in filtering theory, if the follower expects $(a_t)_{t\geq 0}$ as above, then whenever he is on path\footnote{If instead of $\sigma_X=0$, $Y$ is public, this holds also after deviations by the myopic player; see footnote \ref{footnote:Public_Nash}.} his beliefs  are $\theta\sim \mathcal{N}(\hat{M}_t,\gamma_t)$, where
\begin{eqnarray}\label{eq:leading_belief_public}
d\hat M_t=\frac{\beta_{3t}\gamma_t}{\sigma_Y^2}[dY_t-\underbrace{\{\beta_{0t}+(\beta_{1t}+\beta_{3t})\hat M_t\}}_{\hat{\mathbb{E}}_t [a_t]=}dt],\; \hat M_0=\mu \; \text{ and }\; \dot\gamma_t=-\left(\frac{\gamma_t\beta_{3t}}{\sigma_Y^2}\right)^2,\; \gamma_0=\gamma^o.\;\;
\end{eqnarray}

Let $V:\mathbb{R}^2\times [0,T]\to \mathbb{R}$ denote the leader's value function. The HJB equation is $r V =\sup\limits_{a\in\mathbb{R}} \{  -(a-\theta)^2-(a-\hat{a}_t)^2+\frac{\beta_{3t}\gamma_t}{\sigma_Y^2} [a-\beta_{0t}-(\beta_{1t}+\beta_{3t})m] V_m +\frac{\beta_{3t}^2\gamma_t^2}{2\sigma_Y^2}V_{mm}+V_t\}$. We guess a quadratic solution $V(\theta,m,t)=v_{0t}+v_{1t}\theta+v_{2t}m+v_{3t}\theta^2+v_{4t}m^2+v_{5t}\theta m$, from which the FOC in the HJB reads $0=-2(\beta_{0t}+\beta_{1t}m+\beta_{3t}\theta-\theta)-2\beta_{3t}(\theta-m)+\frac{\beta_{3t}\gamma_t[v_{2t}+2m v_{4t}+\theta v_{5t}]}{\sigma_Y^2}$ when the maximizer is $a^*:=\beta_{0t}+\beta_{1t}m+\beta_{3t}\theta$.

From here, $(v_{2t},v_{4t},v_{5t})=\left(\frac{2\sigma_Y^2 \beta_{0t}}{\beta_{3t}\gamma_t},\frac{\sigma_Y^2(\beta_{1t}-\beta_{3t})}{\beta_{3t}\gamma_t},\frac{2\sigma_Y^2(2\beta_{3t}-1)}{\beta_{3t}\gamma_t}\right)$, 
due to the FOC holding for all $(\theta,m,t)\in \mathbb{R}^2\times[0,T]$. And since $v_{iT}=0$ for $i\in \{0,\dots,5\}$, we deduce that $(\beta_{0T},\beta_{1T},\beta_{3T})=(0,1/2,1/2)$, the \emph{myopic equilibrium coefficients}.

Inserting $a^*$ into the HJB equation, and using the previous expressions for $(v_{2t}, v_{4t}, v_{5t})$ to replace $(v_{2t}, v_{4t}, v_{5t}, \dot{v}_{2t}, \dot{v}_{4t}, \dot{v}_{5t})$, yields an equation in $\vec{\beta}:=(\beta_0,\beta_1,\beta_3)$ and $\dot{\vec{\beta}}$. Grouping by coefficients ($\theta, m, \theta^2$,..., etc.) in the latter, we obtain a system of ODEs for $(v_0,v_1,v_3,\beta_0,\beta_1,\beta_3)$: $\dot{v}_{0t}=r v_{0t}+\beta_{3t}\gamma_t(\beta_{3t}-\beta_{1t})$, $\dot{v}_{1t}=r v_{1t}-2\beta_{0t}\beta_{3t}$, and $\dot{v}_{3t}=1+r v_{3t}-2\beta_{3t}^2$ along with
\begin{eqnarray}
\label{eq:leadexamplebetavecpublicFWD}
(\dot{\beta}_{0t},\dot{\beta}_{1t},\dot{\beta}_{3t})=\left(2r \beta_{0t}\beta_{3t},\beta_{3t}\left[r(2\beta_{1t}-1)+\frac{\beta_{1t}\beta_{3t}\gamma_t}{\sigma_Y^2}\right],\beta_{3t}\left[r(2\beta_{3t}-1)-\frac{\beta_{1t}\beta_{3t}\gamma_t}{\sigma_Y^2}\right]\right),
\end{eqnarray}
with conditions $(v_{0T},v_{1T},v_{3T},\beta_{0T},\beta_{1T},\beta_{3T})=(0,0,0,0,1/2,1/2)$. Critically, observe that solving the subsystem $(\beta_0,\beta_1,\beta_3,\gamma)$ delivers the remaining $v_i$, as their ODEs are uncoupled from one another and linear in themselves. The existence of a LME then reduces to the BVP defined by the $\gamma$-ODE \eqref{eq:leading_belief_public} and \eqref{eq:leadexamplebetavecpublicFWD} above, with $\gamma_0=\gamma^o$ and $(\beta_{0T},\beta_{1T},\beta_{3T})=(0,1/2,1/2)$.

To show existence, we transform this BVP into a backward (i.e., reversing the direction of time) IVP problem, using a parametrized initial value for $\gamma$. Abusing notation,
\begin{eqnarray}
(\dot{\beta}_{0t},\dot{\beta}_{1t},\dot{\beta}_{3t},\dot{\gamma}_t)=\beta_{3t}\times\left(-2r \beta_{0t},r(1-2\beta_{1t})-\frac{\beta_{1t}\beta_{3t}\gamma_t}{\sigma_Y^2},r(1-2\beta_{3t})+\frac{\beta_{1t}\beta_{3t}\gamma_t}{\sigma_Y},\frac{\beta_{3t}\gamma_t^2}{\sigma_Y^2}\right)\label{eq:leadexampleODEspublicBWD},\;\;\;
\end{eqnarray}
with initial conditions $\beta_{00}=0$, $\beta_{10}=\beta_{30}=\frac{1}{2}$ and $\gamma_0=\gamma^F\geq 0$. Define $B^{\textit{Pub}}_t:=\beta_{1t}+\beta_{3t}$.

\begin{lemma}\label{lem:LeadingPublicIVP}
Fix any $\gamma^F\geq 0$. If a solution to the backward system exists over $[0,T]$, then any such solution must have the following properties. If $\gamma^F>0$, then (i) $B^{\textit{Pub}}_t=1$ for all $t\in [0,T]$, (ii) $\beta_{3t}\in (1/2,1)$ and $\beta_{1t}\in (0,1/2)$ for all $t\in (0,T]$, (iii) $\beta_3$ is monotonically increasing while $\beta_1$ is monotonically decreasing, and (iv) $\gamma$ is strictly increasing. If $\gamma^F=0$, then $\beta_{1t}=\beta_{3t}=\frac{1}{2}$ and $\gamma_t=0$ for all $t\in [0,T]$. For any $\gamma^F\geq 0$, $\beta_{0}\equiv 0$. 
\end{lemma}

\begin{proof}[Proof of Lemma \ref{lem:LeadingPublicIVP}.]
Because the system \eqref{eq:leadexampleODEspublicBWD} is $C^1$, the solution is unique when it exists. If $\gamma^F=0$, it is clear by inspection that $(\beta_0,\beta_1,\beta_3,\gamma)= (0,1/2,1/2,0)$ (uniquely) solves the IVP, so assume hereafter that $\gamma^F>0$. We first claim that $\beta_{3}>0$. Indeed, let $f^{\beta_{3}}(t,\beta_{3t})$ denote the RHS of the $\beta_3$-ODE in \eqref{eq:leadexampleODEspublicBWD}. Letting $x_t:=0$ for all $t\in[0,T]$, we have $\beta_{30}=1/2>x_0$ and $\dot{\beta}_{3t}-f^{\beta_3}(t,\beta_{3t})=0=\dot{x}_t-f^{\beta_3}(t,x_t)$; by the comparison theorem, the claim follows. Now, add the ODEs that $\beta_1$ and $\beta_3$ satisfy to get $\dot{B}^{\textit{Pub}}_t=2r\beta_{3t}(1-B^{\textit{Pub}}_t)$  with $B^{\textit{Pub}}_0=1$; because the RHS is of class $C^1$, it has a unique solution, which is clearly $B^{\textit{Pub}}=1$. Hence, $\beta_1+\beta_3=1$ and $\dot{\beta}_{3t}=\beta_{3t}\left[r(1-2\beta_{3t})+\frac{\beta_{3t}(1-\beta_{3t})\gamma_t}{\sigma_Y^2}\right]$, and we maintain the label $f^{\beta_3}(t,\beta_{3t})$ for its RHS. Defining $x_t:=1$ for all $t\in[0,T]$, then, $x_0=1>\beta_{30}=\frac{1}{2}$, and $\dot{\beta}_{3t}-f^{\beta_3}(t,\beta_{3t})=0\leq r=\dot{x}_t-f^{\beta_3}(t,x_t)$; thus, $\beta_3<1$ and $\beta_1=1-\beta_3>0$.

Since $\beta_3>0$,  $\gamma$ is clearly strictly increasing, and hence $\gamma_t>0$ for all $t\in[0,T]$. Now,  $\dot{\beta}_{3t}=\frac{1}{2}\left[0+\frac{\gamma_t}{4\sigma_Y^2}\right]>0$ whenever $\beta_{3t}=\frac{1}{2}$, and thus $\beta_{3t}>1/2$ and $\beta_{1t}<1/2$ for all $t\in (0,T]$. 

We now turn to (iii). Since $\dot{\beta}_{1t}+\dot{\beta}_{3t}=0$, we just show that $\dot{\beta}_{3}>0$; in turn, it suffices to show that $H_t:=\dot{\beta}_{3t}/\beta_{3t}=r(1-2\beta_{3t})+\frac{\beta_{3t}(1-\beta_{3t})\gamma_t}{\sigma_Y^2}>0$ for all $t\in [0,T]$. Observe that $H_0=\frac{\gamma_0}{4\sigma_Y^2}>0$, and with algebra it can be shown that if $H_t=0$, $\dot{H}_t=\frac{(1-\beta_{3t})\beta_{3t}^3\gamma_t^2}{\sigma_Y^4}>0$. It follows that $H>0$ as desired. Finally, note that in all cases, we have $\beta_3>0$, so from  \eqref{eq:leadexampleODEspublicBWD} $\beta_0\equiv 0$. Also, as long as $\gamma^F>0$, $\gamma>0$, so $(v_2,v_4,v_5)$ are well defined.\end{proof}

It remains to show that there is a $\gamma^F>0$ such that $\gamma_T=\gamma^o$ in the backward system while all the other ODEs admit solutions. As we argue in the proof of Theorem \ref{thm:LeadingInteriorExistence}, it suffices to show that the solutions are uniformly bounded when $\gamma_t\in [0,\gamma^o]$ for $t\in [0,T]$---refer to that proof for the details of the argument. Applied to this context, the bounds $\beta_0,\beta_1,\beta_3\in [0,1]$ from Lemma \ref{lem:LeadingPublicIVP} are valid more generally as long as $\gamma$ does not explode, so there is indeed a solution to the BVP, and hence a LME exists. To conclude, part (ii) in the proposition is implied by Lemma \ref{lem:LeadingPublicIVP}, while the uniqueness property is shown in the online appendix.\end{proof}

\vspace{-.2in}
\subsection{Proofs for No-Feedback Case}

\begin{lemma}[Belief Representation]\label{lem:Belief_Rep_NF}
Suppose that the follower expects $a_t=[\beta_{0t}+\beta_{1t}(1-\chi_t)]\mu+\alpha_t\theta$, where $\alpha=\beta_3+\beta_1\chi$, $\chi=1-\gamma/\gamma^o$, and $\gamma_t:=\hat{\mathbb{E}}_t[(\theta-\hat M_t)^2]$. Then $\dot\gamma_t = -\left(\frac{\gamma_t\alpha_t}{\sigma_Y^2}\right)^2$. Moreover, if the leader follows \eqref{eq:leading_strat_NF}, $M_t =\chi_t\theta + (1-\chi_t)\mu$ holds at all times.
\end{lemma}

\begin{proof}[Proof of Lemma \ref{lem:Belief_Rep_NF}.] Anticipating $a_t=\alpha_{0t}\mu+\alpha_t \theta$, with $\alpha_0=[\beta_0+\beta_1(1-\chi)]\mu$ and $\alpha=\beta_3+\beta_1\chi$, the myopic player's belief is $\sim\mathcal{N}(M_t,\gamma_t)$ where $d\hat M_{t}=\frac{\alpha_{t}\gamma_{t}}{\sigma_Y^2}[dY_t-(\alpha_{0t}+\alpha_{t}\hat M_t)dt]$ and $\dot{\gamma_{t}}=-\frac{\gamma_{t}^2\alpha_{t}^2}{\sigma_Y^2}$. Thus, $\hat M_{t}= \mu R(t,0)+\int_0^t R(t,s)\frac{\alpha_{s}\gamma_{s}}{\sigma_Y^2}[(a_{s}-\alpha_{0s})ds+\sigma_Y dZ_s^Y]$ and $M_{t}= \mu R(t,0)+\int_0^t R(t,s)\frac{\alpha_{s}\gamma_{s}}{\sigma_Y^2}(a_{s}-\alpha_{0s})ds$ where $R(t,s)=\exp(-\int_s^t\frac{\alpha_{u}^2\gamma_{u}}{\sigma_Y^2}du)$. Solving for $M$ after inserting $a_{t}=\beta_{0t}\mu+\beta_{1t}M_{t}+\beta_{3t}\theta$, and imposing the representation, it is easy to conclude that \eqref{eq:leading_M_representation} will hold if and only if  $\dot\chi_t=\frac{\alpha_{t}^2\gamma_{t}}{\sigma_Y^2}(1-\chi_t)$. By arguments analogous to those used for Lemma \ref{lem:UniqueLearningSol}, the $(\gamma,\chi)$-ODE pair admits a unique solution, and it satisfies $\chi=1-\gamma/\gamma^o$. \end{proof}

\begin{proof}[Proof of Proposition \ref{prop:Leading_NF_LME}.] If the leader uses $a_t=\beta_{0t}\mu+\beta_{1t}M_t+\beta_{3t}\theta$ then, using the representation $M_t=\chi_t\theta+(1-\chi_t)\mu$, $\hat{a}_{t}=\hat{\mathbb{E}}_t[a_t]=\alpha_{0t}\mu+\alpha_t M_t$, where $\alpha_{0t}:=\beta_{0t}+\beta_{1t}[1-\chi_t$ and $\alpha_t=\beta_{1t}\chi_t+\beta_{3t}$. Taking an expectation in the leader's flow payoff $-(a_t-\theta)^2-(a_t-\hat{a}_t)^2$ then yields that $(\theta,M_t,t)$ is the relevant state on and off path. (Indeed, expanding the squares in the previous expression the only nontrivial component is $\mathbb{E}_t[\hat{a}_t^2]$, which makes $\mathbb{E}_t[\hat{M}_t^2]$ appear; however, $\mathbb{E}_t[\hat{M}_t^2]=M_t^2+\mathbb{E}_t[(\hat{M}_t-M_t)^2]=M_t^2+\gamma_{t}\chi_t$ after all private histories.\footnote{From the proof of Lemma \ref{lem:Belief_Rep_NF}, $\mathbb{E}_t[(\hat{M}_t-M_t)^2]=\mathbb{E}_t[(\int_0^t R(t,s)\frac{\alpha_{s}\gamma_{s}}{\sigma_Y} dZ_s^Y)^2]=\int_0^t R(t,s)^2\frac{\alpha_{s}^2\gamma_{s}^2}{\sigma_Y^2} ds=\int_0^t \exp(2\int_s^t \frac{\dot \gamma_u}{\gamma_u}du)(-\dot{\gamma}_s) ds=\int_0^t \left(\gamma_t/\gamma_s\right)^2(-\dot{\gamma}_s)ds=\gamma_t^2(1/\gamma_t-1/\gamma^o)=\gamma_t\chi_t$.}) 

We can then set up the HJB equation. Since $dM_t=\frac{\alpha_t\gamma_t}{\sigma_Y^2}\left(a-\alpha_{0t}-\alpha_t m\right)dt$ from the proof of Lemma \ref{lem:Belief_Rep_NF}$, 
r V =\sup\limits_{a\in\mathbb{R}} \{  -(a-\theta)^2-\left(a^2-2a [\alpha_{0t}+\alpha_{t}m]+\alpha_{0t}^2+2\alpha_{0t}\alpha_{t}m+\alpha_{t}^2[m^2+\gamma_t\chi_t]\right)+V_t+\frac{\alpha_t\gamma_t}{\sigma_Y^2}\left(a-\alpha_{0t}-\alpha_t m\right) V_{m}\}$. We then guess $V(\theta,m,t)=v_{0t}+v_{1t}\theta+v_{2t}m+v_{3t}\theta^2+v_{4t}m^2+v_{5t}\theta m$ and take analogous steps to those in the proof of Proposition \ref{prop:Leading_Public_LME}. Namely, we first show that there is a core BVP consisting of $(\beta_0,\beta_1,\beta_3,\gamma)$. Second, we construct a backward IVP version of our original BVP that has a parametrized initial condition $\gamma^F$ for the $\gamma-$ODE:\footnote{The detailed steps can be found in the online appendix.}
\begin{align}
\dot\beta_{0t}&=\alpha_t(2\sigma_Y^2)^{-1}\times \left\lbrace -r\sigma_Y^2\beta_{0t}(2-\chi_t)+r\sigma_Y^2(1-\chi_t)-2\gamma_t\beta_{1t}^2(1-\chi_t)\right\rbrace \label{eq:leadeampleNFbeta0ODEBWD}\\
\dot \beta_{1t}&=\alpha_t(2\sigma_Y^2)^{-1}\times \left\lbrace r\sigma_Y^2-2\beta_{1t}[\beta_{3t}\gamma_t+r\sigma_Y^2(2-\chi_t)]+2\beta_{1t}^2\gamma_t(1-\chi_t)\right\rbrace \label{eq:leadeampleNFbeta1ODEBWD}\\
\dot\beta_{3t}&=\alpha_t(2\sigma_Y^2)^{-1}\times \left\lbrace r\sigma_Y^2(2-\chi_t)+2\beta_{3t}[\beta_{1t}\gamma_t-r\sigma_Y^2(2-\chi_t)]\right\rbrace \label{eq:leadeampleNFbeta3ODEBWD}\\
\dot{\gamma}_t&=\alpha_t^2\gamma_t^2/\sigma_Y^2\label{eq:LeadingNFgammaODEBWD}
\end{align}
with initial condition $(\beta_{00},\beta_{10},\beta_{30},\gamma_0)= (\frac{1-\chi_0}{2(2-\chi_0)},\frac{1}{2(2-\chi_0)},\frac{1}{2},\gamma^F)$ and where $\chi=1-\gamma/\gamma^o$.

We aim to prove that there exists $\gamma^F\in (0,\gamma^o)$ such that the IVP has a (unique) solution which satisfies $\gamma_T=\gamma^o$. ($\gamma^F=0$ cannot work, as $(\beta_0,\beta_1,\beta_3,\gamma)=(0,1/2,1/2,0)$ is the unique solution.) As argued in the proof of Proposition \ref{prop:Leading_Public_LME}, it suffices to show that the system is uniformly bounded if $\gamma_t\in [0,\gamma^o]$ over $[0,T]$ (see the proof of Theorem \ref{thm:LeadingInteriorExistence} for further details).

The $\alpha$-ODE is $\dot\alpha_{t}=f^{\alpha}(t,\alpha_t):=r\alpha_t[1-\alpha_t(2-\chi_t)]$ and $\alpha_0=\frac{1}{2-\chi_0}>0$. By the comparison theorem, $\alpha>0$; hence, by the same argument as in the proof of Lemma \ref{lem:LeadingPublicIVP}, $\gamma$ is increasing (in the backward system), so $\chi=1-\gamma/\gamma^o<1$ is decreasing.
As $\alpha_0=\frac{1}{2-\chi_0}$ and $\dot \alpha_0> \frac{d}{dt}\left(\frac{1}{2-\chi_t}\right)|_{t=0}$, the comparison theorem can be applied to $\alpha$ and $1/(2-\chi)$ to show $\alpha_t\geq 1/(2-\chi_t)\geq 1/2$, with both inequalities strict for all $t\in (0,T]$, for all $r\geq 0$; in turn, $\dot{\alpha_t}\leq 0$ (and hence $\dot{\alpha}_t\geq 0$ in the forward system) for all $t\in [0,T]$, with strict inequality for $t\in (0,T]$ if and only if $r>0$. It follows that for all $t\in (0,T]$, $\alpha_t\leq \alpha_0=\frac{1}{2-\chi_0}<1$.

Now, $B^{\textit{NF}}:=\beta_0+\beta_1+\beta_3$ satisfies
$
\dot{B}^{\textit{NF}}_t=\frac{\alpha_t}{2\sigma_Y^2}\left\lbrace 2r\sigma_Y^2  (2-\chi_t) [1-B_t^{\textit{NF}}] \right\rbrace$ with $B_0^{\textit{NF}} = 1
$; thus $B^{\textit{NF}}\equiv1$. By routine application of the comparison theorem to the backward system, $\beta_3\in (1/2,1)$ and $\beta_1\in (0,1)$, from which $\beta_0$ is bounded too; thus, a solution to the BVP for $(\beta_0,\beta_1,\beta_3,\gamma)$ exists, as discussed in the proof of Proposition \ref{prop:Leading_Public_LME}. In the online appendix we check that the rest of the coefficients are well defined, ensuring the existence of an LME.

The final claim is $\alpha_T\to 1$ as $T\to \infty$ in the forward system. Indeed, since $\alpha>1/2$, we have $\gamma_T\to 0$ as $T\to \infty$; thus $\chi_T\to 1$ and $\alpha_T=1/(2-\chi_T)\to 1$, all in forward form.
\end{proof}

For the proofs of Propositions \ref{prop:Leading_Learning_Comp} and \ref{prop:Leading_Payoff_Comp}, see the online appendix.

\vspace{-0.2in}
\section{Proofs for Section \ref{sec:eqmanalysis}\label{appendix_B}}

\noindent\textbf{Proof of Lemma \ref{lem:BeliefDecomp}}.  We consider a drift of the form $\hat a_t+\nu a_t$, $\nu\in [0,1]$, in $X$. Also, let $L$ in \eqref{eq:BeliefDecomposition} denote a process that is measurable with respect to $X$. Inserting \eqref{eq:BeliefDecomposition} into \eqref{eq:P1StrategyGeneral} yields $a_t= \alpha_{0t}+\alpha_{2t}L_t+\alpha_{3t}\theta$ which the myopic player thinks drives $Y$, where $\alpha_{0t}=\beta_{0t}$, $\alpha_{2t}=\beta_{2t}+\beta_{1t}(1-\chi_{t})$, and $\alpha_{3t}=\beta_{3t}+\beta_{1t}\chi_{t}$ (the latter often abbreviated $\alpha_t$ in this appendix).

The myopic player's filtering problem is then conditionally Gaussian. Specifically, define 
\begin{eqnarray}
d\hat{X}_{t}&:=&dX_{t}-[\hat{a}_t+\nu(\alpha_{0t}+\alpha_{2t}L_t)]dt=\nu\alpha_{3t} \theta dt+\sigma_X dZ^X_{t}\notag\\
d\hat{Y}_{t}&:=&dY_{t}-[\alpha_{0t}+\alpha_{2t}L_t]dt=\alpha_{3t} \theta dt+\sigma_Y dZ_t^Y,\notag
\end{eqnarray}
which are in the myopic player's information set, and where the last equalities hold from his perspective. By Theorems 12.6 and 12.7 in \cite{liptser1977statistics}, his posterior belief is Gaussian with mean $\hat{M}_t$ and variance $\gamma_{1t}$ (simply $\gamma_t$ in the main body) that evolve as 
\begin{eqnarray}\label{eq:dhatM}
d\hat{M}_t=\frac{\nu\alpha_{3t}\gamma_{1t}}{\sigma_X^2}[d\hat{X}_{t}-\nu\alpha_{3t}\hat{M}_tdt]+\frac{\alpha_{3t}\gamma_{1t}}{\sigma_Y^2}[d\hat{Y}_{t}-\alpha_{3t}\hat{M}_tdt]\; \text{ and }\; \dot{\gamma_{1t}}=-\gamma_{1t}^2\alpha_{3t}^2\Sigma,
\end{eqnarray}
with $\Sigma:=\nu^2/\sigma_X^2+1/\sigma_Y^2$. (These expressions still hold after deviations, which go undetected.)

The long-run player can affect $\hat{M}_t$ via her choice of actions. Indeed, using that $d\hat X=\nu (a_t-\alpha_{0t}-\alpha_{2t}L_t)dt+\sigma_XdZ_t^X$ and $d\hat Y_t=(a_t-\alpha_{0t}-\alpha_{2t}L_t)dt+\sigma_Y dZ_t^Y$ from her standpoint,
\begin{eqnarray}\label{eq:Mhat}
d\hat{M}_t&=& (\kappa_{0t}+\kappa_{1t}a_t+\kappa_{2t}\hat{M}_t)dt+B_{t}^X dZ^X_{t}+ B_{t}^YdZ^Y_{t},\; \text{ where }\\
\kappa_{1t}= \alpha_{3t}\gamma_{1t}\Sigma,\; &\kappa_{0t}&= -\kappa_{1t}[\alpha_{0t}+\alpha_{2t}L_t],\; \kappa_{2t}= -\alpha_{3t}\kappa_{1t},\; B^X_{t}=\frac{\nu\alpha_{3t}\gamma_{1t}}{\sigma_X},\; B^Y_{t}=\frac{\alpha_{3t}\gamma_{1t}}{\sigma_Y}.\;\;\;\;\;\;\;\;\;\label{eq:mus_Mhat}
\end{eqnarray}

On the other hand, since the long-run player always thinks that the myopic player is on path, the public signal evolves, from her perspective, as
$dX_{t}=(\nu a_t + \delta_{0t}+ \delta_{1t}\hat{M}_tdt+\delta_{2t}L_t)dt+\sigma_X dZ^X_{t}$. Because the dynamics of $\hat M$ and $X$ have drifts that are affine in $\hat M$---with intercepts and slopes that are in the long-run player's information set---and deterministic volatilities, the pair $(\hat M, X)$ is conditionally Gaussian. Thus, by the filtering equations in Theorem 12.7 in \cite{liptser1977statistics}, $M_t:=\mathbb{E}_t[\hat M_t]$  and $\gamma_{2t}:=\mathbb{E}_t[(M_t-\hat M_t)^2]$ satisfy 
\begin{eqnarray}\label{eq:M_filtering}
dM_t&=&\underbrace{(\kappa_{0t}+\kappa_{1t}a_t+\kappa_{2t}M_t)dt}_{=\mathbb{E}_t[(\kappa_{0t}+\kappa_{1t}a_t+\kappa_{2t}\hat{M}_t)dt]}+\frac{\sigma_XB_{t}^X+\gamma_{2t}\delta_{1t}}{\sigma_X^2}[dX_{t}-(\nu a_t+\delta_{0t}+\delta_{1t}M_t+\delta_{2t}L_t)dt]\;\;\;\;\;\;\;\;\\
\dot{\gamma}_{2t}&=& 2\kappa_{2t}\gamma_{2t}+(B^X_{t})^2+(B_{t}^Y)^2-\left( B_{t}^X+\gamma_{2t}\delta_{1t}/\sigma_X\right)^2,\label{eq:gamma2ODE}
\end{eqnarray}
with $dZ_t:= [dX_{t}-(\nu a_t+\delta_{0t}+\delta_{1t}M_t+\delta_{2t}L_t)dt]/\sigma_X$ a Brownian motion from the long-run player's standpoint.\footnote{Theorem 12.7 in \cite{liptser1977statistics} is stated for actions that depend on $(\theta, X)$ exclusively, but it also applies to those that condition on past play (i.e., on $M$). Indeed, from \eqref{eq:Mhat},  $\hat M_t= \hat M_t^\dag+ A_t$ where $M_t^\dag=M_t^\dag[Z_s^X,Z_t^Y; s<t]$ and $A_t=\int_0^t e^{\int_0^s\kappa_{2u}du}\kappa_{1s}a_sds$. Applying the theorem to $(\hat M^\dag_t, X_t-\int_0^t \nu a_sds)_{t\in [0,T]}$, yields a posterior mean $M_t^\dag$ and variance $\gamma_{2t}^\dag$ for $\hat M^\dag$ such that $M^\dag+A_t =M_t$ as in \eqref{eq:M_filtering} and $\gamma_{2t}=\gamma_{2t}^\dag$.} Critically, observe that since \eqref{eq:M_filtering} is linear, one can solve for $M_t$ as an \emph{explicit} function of past actions $(a_s)_{s<t}$ and past realizations of the public history $(X_s)_{s<t}$.

Inserting $a_t=\beta_{0t}+\beta_{1t}M_t+\beta_{2t}L_t+\beta_{3t}\theta$ in  \eqref{eq:M_filtering} and collecting terms yields $dM_t=[\hat\kappa_{0t}+\hat\kappa_{1t}M_t+\hat\kappa_{2t}L_t+\hat\kappa_{3t}\theta]dt+\hat B_t dX_{t}$, where, (i) $\hat\kappa_{0t}=-\alpha_{3t}\gamma_{1t}\alpha_{0t}\Sigma+ \alpha_{3t}\gamma_{1t}\beta_{0t}\Sigma+\frac{\nu\alpha_{3t}\gamma_{1t}+\gamma_{2t}\delta_{1t}}{\sigma_X^2}[-\nu\beta_{0t}-\delta_{0t}]$, (ii) $\hat\kappa_{1t}= \alpha_{3t}\gamma_{1t}\beta_{1t}\Sigma-\alpha_{3t}^2\gamma_{1t}\Sigma+\frac{\nu\alpha_{3t}\gamma_{1t}+\gamma_{2t}\delta_{1t}}{\sigma_X^2}[-\nu\beta_{1t}-\delta_{1t}]$, (iii) $\hat\kappa_{2t}= -\alpha_{3t}\gamma_{1t}\alpha_{2t}\Sigma+\alpha_{3t}\gamma_{1t}\beta_{2t}\Sigma+\frac{\nu\alpha_{3t}\gamma_{1t}+\gamma_{2t}\delta_{1t}}{\sigma_X^2}[-\nu\beta_{2t}-\delta_{2t}]$, (iv) $\hat\kappa_{3t}=\left[\frac{\alpha_{3t}\gamma_{1t}}{\sigma_Y^2}-\frac{\nu\gamma_{2t}\delta_{1t}}{\sigma_X^2}\right]\beta_{3t}$ and (v) $\hat B_t= \frac{\nu\alpha_{3t}\gamma_{1t}+\gamma_{2t}\delta_{1t}}{\sigma_X^2}$.

Let $R(t,s)=\exp(\int_s^t \hat\kappa_{1u}du)$. Since $M_{0}=\mu$, we have  $M_t=R(t,0)\mu+\theta\int_0^t R(t,s)\hat\kappa_{3s}ds+\int_0^t R(t,s)[\hat\kappa_{0s}+\hat\kappa_{2s}L_s]ds+ \int_0^t R(t,s)\hat B_s dX_{s}$. Imposing equality with \eqref{eq:BeliefDecomposition} yields the equations $
\chi_{t}= \int_0^t R(t,s)\hat\kappa_{3s}ds$ and  $L_t = [R(t,0)\mu+\int_0^t R(t,s)[\hat\kappa_{0s}+\hat\kappa_{2s}L_s]ds+ \int_0^t R(t,s)\hat B_sdX_{s}]/[1-\chi_t]$. The validity of the construction boils down to finding a solution to the previously stated equation for $\chi$ that takes values in $[0,1)$. Indeed, when this is the case, it is easy to see that
\begin{align}\label{eq:L_differential}
 dL_t&=\{L_t[\hat{\kappa}_{1t}+\hat{\kappa}_{2t}+\hat{\kappa}_{3t}]dt+\hat{\kappa}_{0t}dt+\hat{B}_tdX_t\}/(1-\chi_t),
 \end{align}
from which it is easy to conclude that $L$ is a (linear) function of $X$ as conjectured.

We will find a solution to the $\chi$-equation that is $C^1$ with values in $[0,1)$. Differentiating $\chi_{t}= \int_0^t R(t,s)\hat\kappa_{3s}ds$ then yields an ODE for $\chi$ as below that is coupled with $\gamma_1$ and $\gamma_2$:
\begin{eqnarray}
\dot\gamma_{1t}&=&-\gamma_{1t}^2(\beta_{3t}+\beta_{1t}\chi_t)^2\Sigma\notag\\
\dot\gamma_{2t}&=&-2\gamma_{2t}\gamma_{1t}(\beta_{3t}+\beta_{1t}\chi_t)^2\Sigma+\gamma_{1t}^2(\beta_{3t}+\beta_{1t}\chi_t)^2\Sigma -\left(\nu\gamma_{1t}(\beta_{3t}+\beta_{1t}\chi_t)+\gamma_{2t}\delta_{1t}\right)^2/\sigma_X^2\notag\\
\dot\chi_t&=& \gamma_{1t}(\beta_{3t}+\beta_{1t}\chi_t)^2\Sigma(1-\chi_t)-(\nu[\beta_{3t}+\beta_{1t}\chi_t]+\delta_{1t}\chi_t)\left(\nu\gamma_{1t}(\beta_{3t}+\beta_{1t}\chi_t)+\gamma_{2t}\delta_{1t}\right)/\sigma_X^2.\notag
\end{eqnarray}
In the proof of Lemma \ref{lem:UniqueLearningSol} we establish that $\chi=\gamma_2/\gamma_1\in [0,1)$ taking the system above as a primitive (we do it for $\nu=0$, but it also holds otherwise). Setting $\nu=0$ and $\gamma_2=\chi\gamma_1$ in the third ODE, and writing $\gamma$ for $\gamma_1$, the first and third ODEs become \eqref{eq:gammadot}--\eqref{eq:chidot}. If now $\nu=0$, using (i)--(v) that define $(\vec{\hat{\kappa}},\hat B)$ yields that \eqref{eq:L_differential} becomes $dL_t=(\ell_{0t}+\ell_{1t}L_t)dt+B_t dX_t$ where 
\begin{eqnarray}\label{eq:l0_l1_B}
(l_{0t},l_{1t},B_t)=[\sigma_X^2(1-\chi_t)]^{-1} \times (-\gamma_t\chi_t\delta_{0t}\delta_{1t},-\gamma_t\chi_t\delta_{1t}(\delta_{1t}+\delta_{2t}),\gamma_t\chi_t\delta_{1t}). 
\end{eqnarray}
That $L_t$ coincides with $\mathbb{E}[\theta|\mathcal{F}_t^X]$ is proved in the Online Appendix. \hfill{$\square$}\\

\noindent\textbf{Proof of Lemma \ref{lem:UniqueLearningSol}.} Consider the system in $(\gamma_1,\gamma_2,\chi)$ from the proof of the previous lemma when $\nu=0$, and let $\delta_{1t}:=\hat u_{\hat a\theta}+\hat u_{a \hat a}\alpha_{3t}$.\footnote{All the results in this proof extend (i) to $\nu\in [0,1]$ and (ii) to $\delta_1$ a generic continuous function over $[0,T]$, the latter case arising when the myopic player becomes forward looking.} The local existence of a solution follows from Peano's Theorem. Suppose that the maximal interval of existence is $[0,\tilde T)$, with $\tilde T\leq T$. Since the system is locally Lipschitz continuous in $(\gamma_1,\gamma_2,\chi)$ uniformly in $t\in [0,T]$, its solution over $[0,\tilde T)$ is unique (Picard-Lindel\"{o}f). Applying the comparison theorem to the pairs $\{\gamma_1,0\}$ and $\{\gamma_1,\gamma^o\}$, we get $\gamma_{1t}\in (0,\gamma^o]$ over $[0,\tilde T)$. Hence, $\gamma_2/\gamma_1$ is well-defined, and since it solves the $\chi$-ODE, $\chi=\gamma_2/\gamma_1$ by uniqueness. Replacing $\gamma_2=\chi\gamma_1$ and $\nu=0$ in the $\chi-$ODE then yields \eqref{eq:chidot}. A second application of the comparison theorem to $\{\chi,0\}$ and $\{\chi,1\}$ then implies $\chi\in [0,1)$, and in turn $\gamma_2=\chi\gamma_1\in [0,\gamma^o)$, over $[0,\tilde T)$. Since the solution is bounded, if $\tilde T<T$, it can be extended to $\tilde T$ by the continuity of the RHS of the system; and then subsequently extended beyond $\tilde T$ by Peano's theorem, a contradiction. But if $\tilde T=T$, it can be extended to $T$---the first part of the lemma holds. If $\beta_{30}\neq 0$, then $\dot\gamma_{10}<0$ and $\dot\chi_0>0$, so by continuity of $\dot\gamma_1$ and $\dot\chi$, there exists $\epsilon>0$ such that $\gamma_{1t}<\gamma^o$ and $\chi_t>0$ for all $t\in (0,\epsilon)$, and by the comparison theorem, these strict inequalities hold up to time $T$.\hfill{$\square$}\\

\noindent\textbf{Proof of Lemma \ref{lem:M and L laws of motion}.} Inserting $\nu=0$ in \eqref{eq:mus_Mhat} defining $(\kappa_0,\kappa_1,\kappa_2,B^X_t)$ yields that \eqref{eq:M_filtering} becomes $dM_t=\frac{\gamma_t\alpha_{3t}}{\sigma_Y^2}(a_{t}-[\alpha_{0t}+\alpha_{2t}L_t+\alpha_{3t}M_t])dt+\frac{\chi_t\gamma_t\delta_{1t}}{\sigma_X}dZ_t$, 
where $dZ_t:= [dX_{t}-(\nu a_t+\delta_{0t}+\delta_{1t}M_t+\delta_{2t}L_t)dt]/\sigma_X$ a Brownian motion from the long-run player's standpoint. As for the law of motion of $L$, this one follows from \eqref{eq:dLonpath} using \eqref{eq:l0_l1_B} and that $dX_t=(\delta_{0t}+\delta_{2t}L_t+\delta_{1t}M_t)dt+\sigma_X dZ_t$ from the long-run player's perspective when $\nu=0$.

We conclude with three observations. First, from \eqref{eq:Mhat} and \eqref{eq:M_filtering}, $\hat M_t- M_t$ is independent of the strategy followed, and hence so is $Z_t$ due to $\sigma_XdZ_t = \delta_{1t}(\hat M_t- M_t)dt+\sigma_X dZ_t^X$ under the true data-generating process. This strategic independence enables us to fix an exogenous Brownian motion $Z$ and then solve the best-response problem with $Z$ in the laws of motion of $ M$ and $L$---i.e., the so-called \emph{separation principle} for control problems with unobserved states applies (see, for instance, \citealp{liptser1977statistics}, Chapter 16).

Second, it is clear from \eqref{eq:LRpprogram}, \eqref{eq:M_filtering}--\eqref{eq:gamma2ODE}, and the proof of Lemma \ref{lem:UniqueLearningSol} that no additional state variables are needed due to $\gamma_{2t}:=\mathbb{E}_t[(M_t-\hat M_t)^2]=\chi_t\gamma_t$ holding irrespective of the strategy chosen. Third, the set of admissible strategies for the best-response problem then consists of all square-integrable processes that are progressively measurable with respect to $(\theta, M,L)$. This set is clearly the appropriate set, and richer than that in Definition \ref{def:eqbm}.\hfill{$\square$}\\

\noindent\textbf{Proof of Lemma \ref{lem:GammaChiRelationship}.}
By Lemma \ref{lem:UniqueLearningSol}, the system of ODEs \eqref{eq:gammadot}-\eqref{eq:chidot} admits a unique solution. The proof then consists of showing that $(\gamma,\chi(\gamma))$ as in the lemma, with $c_2=\frac{\sigma_X^2}{2\hat u_{a\hat a}^2}[\sqrt{1/\sigma_Y^2+4(u_{a\hat a}/[\sigma_X\sigma_Y]^2}-1/\sigma_Y^2]\in (0,1)$, $c_1=\frac{\sigma_X^2}{2\hat u_{a\hat a}^2}[\sqrt{1/\sigma_Y^2+4(u_{a\hat a}/[\sigma_X\sigma_Y]^2}+1/\sigma_Y^2]>0$ and $d=[\sigma_Y \hat u_{a\hat a}]^2(c_1+c_2)/\sigma_X^2>0$, is a solution. This is done in the Online Appendix, where we also show how to construct the candidate $\chi$ and the above coefficients.\hfill{$\square$}

\vspace{-0.1in}
\subsection{Proof of Theorem \ref{thm:LeadingInteriorExistence}}

In light of the generality of Theorem \ref{thm:Leading_LME_BVP}, we only sketch the proof of Theorem \ref{thm:LeadingInteriorExistence} here (all the details are in the Online Appendix). Specifically, the proof can be divided into two steps: \emph{reduction} and  \emph{shooting}.\footnote{A final step of \emph{verification}---i.e., checking that the rest of the coefficients are well defined as the last step for finding a LME---is a special instance of ``Step 5" in the proof of Theorem \ref{thm:Leading_LME_BVP}, and is thus omitted.} In the \emph{reduction} step, we show that the problem of solving the BVP stated in Section \ref{subsec:existence_interior} can be translated to finding a value $\gamma^F$ such that the backward IVP consisting of ODEs for $(v_8\gamma, \beta_1,\beta_3,\gamma)$ with $\gamma_0=\gamma^F$ has a solution over $[0,T]$ that shoots $\gamma$ to $\gamma^o$ at $T$; this is done after recognizing that $v_6$ and $\beta_2$ can be written as functions of the other variables in closed form. We then tackle the \emph{shooting} step via a contradiction. Specifically, we consider the supremum over values $\tilde{\gamma}^F$ such that all the IVPs with $\gamma_0=\gamma^F\in [0,\tilde{\gamma}^F)$ admit a solution over $[0,T]$---towards a contradiction, assume that $\gamma_T<\gamma^o$ for all initial conditions in the maximal set induced (otherwise, our desired conclusion follows from the continuity of the solutions). Under this assumption, $\gamma_t\in [0,\gamma^o]$ at all times, and one can show that for some horizons as in the theorem, the solutions can be uniformly bounded over all initial values in the maximal set. But this in turn implies that a solution to our IVP defined over $[0,T]$ exists for initial values of $\gamma$ strictly above the supremum, a contradiction.\footnote{See \cite{bonatti2017dynamic} for an application of this method to a symmetric oligopoly model featuring dispersed fixed private information, imperfect public monitoring, and multiple long-run firms.}

\subsection{Proof of Theorem \ref{thm:Leading_LME_BVP}}
After a change of variables $\tilde \beta_{2}=\beta_2/(1-\chi), \tilde v_6= v_6\gamma/(1-\chi)^2, \tilde v_8=v_8\gamma/(1-\chi),
$ the BVP is
\begin{align}
\dot{\tilde{v}}_{6t}&=\gamma_t \left\{-\beta_{1t}^2+2 \beta_{1t} \tilde\beta_{2t}+\tilde\beta_{2t}^2+\tilde v_{6t}\left[\alpha_t^2/\sigma_Y^2+2 (\hat u_{ \hat a \theta}+\hat u_{\hat a a}\alpha_t)^2 \chi_t/\sigma_X^2\right]\right\}\label{eq:LeadingIntAsymChangev6}\\
\dot{ \tilde{v}}_{8t}&=\gamma_t \left\{(-2+4 \alpha_t) \beta_{1t}-2 \tilde\beta_{2t}+\tilde v_{8t}(\hat u_{\hat a \theta}+\hat u_{\hat a a}\alpha_t)^2 \chi_t/\sigma_X^2-4 \beta_{1t}^2 \chi_t\right\}\label{eq:LeadingIntAsymChangev8}\\
\begin{split}
\dot \beta_{1t}&=\gamma_t[4 \sigma_X^2 \sigma_Y^2 (1+\hat u_{\hat a \theta} \chi_t)]^{-1}\times \left\{2 \sigma_X^2 \alpha_t \left([\hat u_{\hat a \theta}+\hat u_{\hat a a}\alpha_t]^2-\alpha_t[\hat u_{\hat a \theta}+\hat u_{\hat a a}\alpha_t]\right)\right.\\
&\qquad\left.+4\sigma_X^2\alpha_t\beta_{1t}(\alpha_t-\beta_{1t})+\tilde v_{8t}\alpha_t\chi_t (\hat u_{\hat a \theta}+\hat u_{\hat a a}\alpha_t)^2 (\hat u_{\hat a \theta}-2 \beta_{1t})\right.\\
&\qquad\left.+4 \beta_{1t}\chi_t\left[\hat u_{\hat a \theta}^2 \sigma_Y^2+\hat u_{\hat a \theta} \alpha_t \left(\hat u_{\hat a \theta} \sigma_X^2+2 \hat u_{\hat a a}\sigma_Y^2-\sigma_X^2 \beta_{1t}\right)\right]\right.\\
&\qquad \left.+4\hat u_{\hat a a}\beta_{1t} \alpha_t^2\chi_t\left(2 \hat u_{\hat a \theta} \sigma_X^2+\hat u_{\hat a a}\sigma_Y^2+\sigma_X^2\alpha_t[\hat u_{\hat a a}-1]\right)\right.\\
&\qquad\left.-4 \sigma_Y^2 (\hat u_{\hat a \theta}+\hat u_{\hat a a}\alpha_t)^2 \tilde\beta_{2t}\chi_t+4 \sigma_Y^2 (\hat u_{\hat a \theta}+\hat u_{\hat a a}\alpha_t)^2 \beta_{1t} (\hat u_{\hat a \theta}-2 \tilde\beta_{2t}) \chi_t^2\right\}\label{eq:LeadingIntAsymBeta1}
\end{split}\\
\begin{split}
\dot{\tilde{\beta}}_{2t}&=\gamma_t[4 \sigma_X^2 \sigma_Y^2 (1+\hat u_{\hat a \theta} \chi_t)]^{-1} \times\left\{2 \sigma_X^2 \alpha_t \left[\hat u_{\hat a \theta}^2+2 \beta_{1t}^2+\alpha_t (\hat u_{\hat a \theta}[2\hat u_{\hat a a}-1]+2\tilde\beta_{2t})\right]\right.\\
&\qquad\left. +2\sigma_X^2\alpha_t^3\hat u_{\hat a a}(\hat u_{\hat a a}-1)+\alpha_t\chi_t(\hat u_{\hat a \theta}+\hat u_{\hat a a}\alpha_t)^2 \left[-4 \tilde v_{6t}+\tilde v_{8t} (\hat u_{\hat a \theta}-2 \tilde \beta_{2t})\right]\right.\\
&\qquad\left.+4\alpha_t\chi_t \hat  \sigma_X^2 \left[u_{\hat a \theta}\beta_{1t}^2+(\hat u_{\hat a \theta}^2+\hat u_{\hat a a} \alpha_t[2\hat u_{\hat a \theta}+(\hat u_{\hat a a}-1)\alpha_t]) \tilde\beta_{2t}\right]\right.\\
&\qquad \left.-4 (\hat u_{\hat a \theta}+\hat u_{\hat a a}\alpha_t)^2 \left[\hat u_{\hat a \theta} \tilde v_{6t}\alpha_t+\sigma_Y^2 \tilde\beta_{2t} (-\hat u_{\hat a \theta}+2 \tilde \beta_{2t})\right] \chi_t^2\right\}\label{eq:LeadingIntAsymChangeBeta2}
\end{split}\\
\begin{split}
\dot\beta_{3t}&=\gamma_t[4 \sigma_X^2 \sigma_Y^2 (1+\hat u_{\hat a \theta} \chi_t)]^{-1}\times \left\{-4 \sigma_X^2 \alpha_t^2 \beta_{1t}-\chi_t^2[\tilde v_{8t}\alpha_t (\hat u_{\hat a \theta}+\hat u_{\hat a a}\alpha_t)^2 (\hat u_{\hat a \theta}-2 \beta_{1t})]\right.\\
&\qquad\left.+2 \alpha_t\chi_t(\hat u_{\hat a \theta}+\hat u_{\hat a a}\alpha_t)  \left[-\hat u_{\hat a \theta} \sigma_X^2+\sigma_X^2 \alpha_t\left(2 \hat u_{\hat a \theta}+[\hat u_{\hat a a}-1](2\alpha_t-1)\right)\right]\right.\\
&\qquad \left.-2\alpha_t^2 \chi \tilde v_{8t} (\hat u_{\hat a \theta}+\hat u_{\hat a a}\alpha_t)^2-2\alpha_t\chi_t[2 \hat u_{\hat a \theta} \sigma_X^2 \alpha_t \beta_{1t}-2 \sigma_X^2 \beta_{1t}^2]\right.\\
&\qquad \left.-4\sigma_X^2\chi_t^2 \alpha_t \beta_{1t}[\hat u_{\hat a \theta} \alpha_t(2\hat u_{\hat a a}-1)+\hat u_{\hat a a} \alpha_t^2(\hat u_{\hat a a}-1)+\hat u_{\hat a \theta}(\hat u_{\hat a \theta}-\beta_{1t})]\right.\\
&\qquad\left.-4 \sigma_Y^2\chi_t^2 (\hat u_{\hat a \theta}+\hat u_{\hat a a}\alpha_t)^2 (-1+2 \alpha_t)\tilde \beta_{2t}+8 \sigma_Y^2 (\hat u_{\hat a \theta}+\hat u_{\hat a a}\alpha_t)^2 \beta_{1t} \tilde\beta_{2t} \chi_t^3\right\}\label{eq:LeadingIntAsymBeta3}
\end{split}\\
\dot{\gamma}_{t}&=-\alpha_t^2\gamma_t^2/\sigma_Y^2\label{eq:LeadingIntAsymGamma}\\
\dot{\chi}_{t}&=\gamma_t\left\{\alpha_t^2(1-\chi_t)/\sigma_Y^2-(\hat u_{\hat a \theta}+\hat u_{\hat a a}\alpha_t)^2\chi_t^2/\sigma_X^2\right\}\label{eq:LeadingIntAsymChi}.
\end{align}
with boundaries $(\gamma_0,\chi_0,\tilde{v}_{6T},\tilde{v}_{8T},\beta_{1T},\tilde{\beta}_{2T},\beta_{3T})=(\gamma^o,0,0,0,\frac{\hat u_{\hat a a}+2\hat u_{\hat a \theta}}{2(2-\hat u_{\hat a a}\chi_T)},\frac{\hat u_{\hat a a}+2\hat u_{\hat a \theta}}{2(2-\hat u_{\hat a a}\chi_T)},1/2)$ and where $\alpha:=\beta_3+\beta_1\chi$. (By Lemma \ref{lem:UniqueLearningSol} and Assumption \ref{assu:static} part (ii) (using $u_{a\theta}=u_{a\hat a}=1/2$)  the denominators are strictly positive, and hence well defined after all possible histories.)

The proof proceeds in five steps. The main task is to obtain a solution $(\tilde{v}_6,\tilde{v}_8,\beta_1,\tilde{\beta}_2,\beta_3,\gamma,\chi)$ to the boundary value problem for all $T<T(\gamma^o)$; from there, it is straightforward to verify that the remaining equilibrium coefficients are well defined, as we do at the end of the proof.

\noindent\textbf{Step 1}: \emph{Convert BVP to fixed point problem in terms of a parameterized IVP.} It is useful to introduce $z=(\tilde{v}_6,\tilde{v}_8,\beta_1,\tilde\beta_2,\beta_3,\gamma,\chi)^{\mathsf{T}}$ and write the system of ODEs \eqref{eq:LeadingIntAsymChangev6}-\eqref{eq:LeadingIntAsymChi} as $\dot{z}_t=F(z_t)$. We write $\tilde{z}=(z_1,z_2,\dots,z_5)^{\mathsf{T}}$ and $\tilde{F}(z)=(F_1(z),F_2(z),\dots,F_5(z))^{\mathsf{T}}$.

Define $\mathbf{B}:[0,1)\to \mathbb{R}^5$ by $\mathbf{B}(\chi)=\left(0,0,\frac{\hat u_{\hat a a}+2\hat u_{\hat a \theta}}{2(2-\hat u_{\hat a a}\chi)},\frac{(\hat u_{\hat a a}+2\hat u_{\hat a \theta})}{2(2-\hat u_{\hat a a}\chi)},1/2\right)^{\mathsf{T}}$, formed by writing the terminal value of $\tilde z$ as a function of $\chi$. Define $s_0\in \mathbb{R}^5$ by $s_0=\mathbf{B}(0)=(0,0,\frac{\hat u_{\hat a a}+2\hat u_{\hat a \theta}}{4},\frac{\hat u_{\hat a a}+2\hat u_{\hat a \theta}}{4},1/2)^{\mathsf{T}}$. For $x\in \mathbb{R}^n$, let $||x||_{\infty}$ denote the sup norm, $\sup_{1\leq i\leq n}|x_i|$. For any $\rho>0$, let $\mathcal{S}_\rho(s_0)$ denote the closed $\rho$-ball around $s_0$, $\mathcal{S}_\rho(s_0):=\{s\in \mathbb{R}^5|\, ||s-s_0||_\infty\leq \rho\}.$ 

For all $s\in \mathcal{S}_\rho(s_0)$, let IVP-$s$ denote the initial value problem defined by \eqref{eq:LeadingIntAsymChangev6}-\eqref{eq:LeadingIntAsymChi} and initial conditions $(\tilde{v}_{60},\tilde{v}_{80},\beta_{10},\tilde\beta_{20},\beta_{30},\gamma_0,\chi_0)=(s,\gamma^o,0)$. Whenever a solution to IVP-$s$ exists, it is unique as $F$ is of class $C^1$; denote it by $z(s)$, where $z(s)=(\tilde{z}(s)^{\mathsf{T}},\gamma(s),\chi(s))^{\mathsf{T}}=(\tilde{v}_{6}(s),\tilde{v}_{8}(s),\beta_{1}(s),\tilde\beta_{2}(s),\beta_3(s),\gamma(s),\chi(s))^{\mathsf{T}}$, where we suppress additional dependence on $(\gamma^o,0)$ which remain fixed. Note that such a solution solves the BVP if and only if 
\begin{align}
\tilde{z}_T(s)=\mathbf{B}(\chi_T(s)),\label{eq:IVPsolvesBVP}
\end{align}
as the initial values $\gamma_0(s)=\gamma^o$ and $\chi_0(s)=0$ are satisfied by construction. Note also that $\tilde{z}_T(s)=s+\int_0^T \tilde{F}(z_t(s)) \,dt$; hence \eqref{eq:IVPsolvesBVP} is satisfied if and only if $s$ is a fixed point of the function $g:\mathcal{S}_\rho(s_0)\to \mathbb{R}^5$ defined by $g(s):=\mathbf{B}(\chi_T(s))-\int_0^T \tilde{F}(z_t(s)) dt.$ Note, moreover, that for any solution, we have by Lemma \ref{lem:UniqueLearningSol} that $\chi_t\in [0,\bar{\chi})$, where we define $\bar{\chi}:=1$ for this proof.

\noindent\textbf{Step 2}: \emph{Obtain sufficient conditions for IVP-$s$ to have unique and uniformly bounded solutions for all $s\in \mathcal{S}_\rho(s_0)$, any $\rho>0$.} Specifically, for arbitrary $K>0$, we ensure that the solution $\tilde{z_t}(s)$ varies at most $K$ from its starting point $s$ for all $t\in [0,T]$, and thus by the triangle inequality, this solution varies most $\rho+K$ from $s_0$. These bounds will be used later.

\begin{lemma}\label{lem:LeadingIntSBC}
Fix $\gamma^o,\rho,K>0$. There exists a threshold $T^{SBC}(\gamma^o;\rho,K)>0$ such that if $T<T^{SBC}(\gamma^o;\rho,K)$, then for all $s\in \mathcal{S}_\rho(s_0)$ a unique solution to IVP-$s$ exists over $[0,T]$ with $\tilde z_t(s)\in \mathcal{S}_{\rho+K}(s_0)$, all $t\in [0,T]$. We call this property the System Bound Condition (SBC).
\end{lemma}

\begin{proof}
Fix any $s\in \mathcal{S}_\rho(s_0)$. Since $\tilde F$ is of class $C^1$, a local solution exists, and solutions are unique given existence. We now construct bounds on $\tilde F$ by writing $\tilde F(z(s))=\tilde F(z(s)-s_0+s_0)$ and using the conjectured bounds $||\tilde z(s)-s_0||_\infty<\rho+K$, $\gamma\in (0,\gamma^o]$, $\chi \in [0,\bar{\chi})$ for the solution, when it exists. Using these bounds on $\tilde F$, we identify $T^{SBC}(\gamma^o;\rho,K)$ such that for all $t<T^{SBC}(\gamma^o;\rho,K)$ the solution to IVP-$s$ (exists and) satisfies the conjectured bounds. 

Note that the desired component-wise inequalities $|z_{i t}(s)-s_{i0}|< \rho+K$, $i\in \{1,2,\dots,5\}$, imply the further bounds
$|\tilde{v}_{6t}|,|\tilde{v}_{8t}|<\bar v_6(\rho,K):=\bar v_8(\rho,K):=\rho+K$, $
|\beta_{1t}|,|\tilde{\beta}_{2t}|<\bar\beta_1(\rho,K):=\bar\beta_2(\rho,K):=\frac{|\hat u_{\hat a a}+2\hat u_{\hat a \theta}|}{4}+\rho+K$,
$|\beta_{3t}|<\bar\beta_3(\rho,K):=1/2+\rho+K$ and $|\alpha_t|<\bar{\alpha}(\rho,K):=\bar{\beta}_1(\rho,K)\bar{\chi}+\bar\beta_3(\rho,K)$.
A lower bound on $1+\hat u_{\hat a \theta}\chi_t$ in the denominators of \eqref{eq:LeadingIntAsymBeta1}-\eqref{eq:LeadingIntAsymBeta3} is $\underline u:=\min\{1,1+\hat u_{\hat a \theta}\}$; using that $u_{a\theta}=u_{a\hat a}=1/2$, Assumption \ref{assu:static} part (ii) implies $\underline u>0$.

By the triangle inequality, one can use these bounds construct functions $h_i(\gamma^o;\rho,K)$, $i\in \{1,2,\dots,5\}$, proportional to $\gamma^o$, which bound the magnitudes of the RHS in \eqref{eq:LeadingIntAsymChangev6}-\eqref{eq:LeadingIntAsymBeta3}. Now for arbitrary $(\rho,K)\in \mathbb{R}_{++}^2$, define $T^{SBC}(\gamma^o;\rho,K):= \min_{i\in \{1,2,\dots,5\}} \frac{K}{h_i(\gamma^o;\rho,K)}.$

We claim that, for any $t<T^{SBC}(\gamma^o;\rho,K)$, if a solution exists at time $t$, then $||\tilde z_{t}(s)-s||_\infty<K$, $\gamma_t\in (0,\gamma^o]$ and $\chi_t\in [0,\bar{\chi})$. To see this, suppose by way of contradiction that there is some $s\in \mathcal{S}_\rho$ and some $t<T^{SBC}(\gamma^o;\rho,K)$ at which a solution to IVP-$s$ exists but either $|z_{it}(s)-s_i|\geq K$ for some $i\in \{1,2,\dots,5\}$, $\gamma_t\notin (0,\gamma^o]$ or $\chi_t\notin [0,\bar{\chi})$; let $\tau$ be the infimum of such times. Now by Lemma \ref{lem:UniqueLearningSol}, it cannot be that $\gamma_t\notin (0,\gamma^o]$ or $\chi_t\notin [0,\bar{\chi})$ while $z_t(s)$ exists, so (by continuity of $z_t(s)$ in $t$) it must be that for some $i$, $|z_{i\tau}(s)-s_i|\geq K$, while $\gamma_t\in (0,\gamma^o]$ and $\chi_t\in [0,\bar{\chi})$ hold for $t\in [0,\tau]$. By construction of $h_i(\gamma^o;\rho,K)$, for all $t\in [0,\tau]$ we have $|F_i(z_t(s))|\leq h_i(\gamma^o;\rho,K)$ and thus
$|z_{i\tau}(s)-s_{i}|\leq \int_0^{\tau} |F_i(z_t(s))|dt\leq \tau \cdot h_i(\gamma^o;\rho,K)<T^{SBC}(\gamma^o;\rho,K) h_i(\gamma^o;\rho,K)\leq K$; this is a contradiction, so the claim holds. By the triangle inequality, $\tilde z_t(s)\in \mathcal{S}_{\rho+K}(s_0)$ if a solution exists at time $t<T^{SBC}(\gamma^o;\rho,K)$. Hence, if $T<T^{SBC}(\gamma^o;\rho,K)$, a (unique) solution to IVP-$s$ exists over $[0,T]$, since an explosion at any time would imply that the previous bound is violated at an earlier time.
\end{proof}

\noindent\textbf{Step 3:} \emph{Establish that $g$ is a well defined, continuous self-map on $\mathcal{S}_\rho$ when $T$ is below a threshold $T(\gamma^o;\rho,K)$.} The expression for the latter is shown in the proof Lemma \ref{lem:LeadingIntSMC} below.

\begin{lemma}\label{lem:LeadingIntSMC}
Fix $\gamma^o>0$, $\rho>0$ and $K>0$. There exists $T(\gamma^o;\rho,K)\leq T^{SBC}(\gamma^o;\rho,K)$ such that for all $T<T(\gamma^o;\rho,K)$, $g$ is a well defined, continuous self-map on $\mathcal{S}_\rho$.
\end{lemma}
\begin{proof}
First, the inequality $T(\gamma^o;\rho,K)\leq T^{SBC}(\gamma^o;\rho,K)$, which holds by construction (as carried out below), ensures that a unique solution to IVP-$s$ exists for all $s\in \mathcal{S}_\rho$, and hence $g$ is well defined on $\mathcal{S}_\rho$. Now $g(s)$ is equal to $\mathbf{B}(\chi_T(s))-[\tilde z_T(s)-s]$. The continuity of $g$ then follows from $\tilde z_T$ and $\chi_T$ being continuous in $s$ and $\mathbf{B}$ in $\chi$.

To complete the proof, we show that if $T<T(\gamma^o;\rho,K)$, $g$ satisfies the condition
$||g(s)-s_0||_{\infty}\leq \rho$ for all $s\in \mathcal{S}_\rho$,
which we refer to as the Self-Map Condition (SMC). 

Note that $g(s)-s_0=\Delta(s)-\int_0^T \tilde F(z_t(s))dt$, where $\Delta(s):=\mathbf{B}(\chi_T(s))-\mathbf{B}(0)$ takes the value 
$\frac{|\hat u_{\hat a a}+2\hat u_{\hat a \theta}|}{2}\left[(2-\hat u_{\hat a a}\chi_T(s))^{-1}-1/2\right]$ in its third and fourth components and zero in the others. The $h_i(\gamma^o;\rho,K)$ constructed in the proof of the previous lemma will provide us a bound for the components of $\int_0^T \tilde F(z_t(s))dt$, but we must also bound $\Delta_3(s)=\Delta_4(s)$.

Recalling that $\chi\in [0,1)$, the ODE for $\chi$ implies that $
\dot{\chi}_{t}\leq \gamma_t\left\{\alpha_t^2(1-\chi_t)/\sigma_Y^2\right\}\leq \gamma^o \bar\alpha^2/\sigma_Y^2$. We also have $\chi_T=\int_0^T \dot{\chi}_t dt\leq \left(\gamma^o \bar \alpha^2/\sigma_Y^2\right)T$. Next, observe that $2-\hat u_{\hat a a}\chi_T(s)\geq \phi:=\min\{2,2-\hat u_{\hat a a}\}>0$, where strict inequality follows from Assumption \ref{assu:static} part (iv) using $u_{a \hat a}=1/2$. Hence, $|\Delta_3(s)|= \frac{|\hat u_{\hat a a}+2\hat u_{\hat a \theta}|}{2}\left| \frac{\hat u_{\hat a a}\chi_T(s)}{2(2-\hat u_{\hat a a}\chi_T(s))}\right|\leq \frac{|\hat u_{\hat a a}+2\hat u_{\hat a \theta}|}{4}\frac{|\hat u_{\hat a a}|\left(\gamma^o \bar \alpha^2/\sigma_Y^2\right)T}{\phi}.$

Define $\bar{\Delta}_i:\mathbb{R}_{++}^3\to \mathbb{R}_+$ by $\Delta_i(\gamma^o;\rho,K)=\frac{|\hat u_{\hat a a}+2\hat u_{\hat a \theta}|}{4}\frac{|\hat u_{\hat a a}|\gamma^o \bar \alpha^2/\sigma_Y^2}{\phi}$ for $i\in \{3,4\}$ and $\bar{\Delta}_i(\gamma^o;\rho,K)=0$ for $i\in \{1,2,5\}$. Note that for all $i\in \{1,2,3,4,5\}$,  $\bar{\Delta}_i(\gamma^o;\rho,K)$ is proportional to $\gamma^o$, and by construction, $|\Delta_i(s)|\leq T \bar{\Delta}_i(\gamma^o;\rho,K)$. Finally, define
\begin{align}
T(\gamma^o;\rho,K)&:=\min \left\lbrace T^{SBC}(\gamma^o;\rho,K),\min_{i\in \{1,2,\dots,5\}} \frac{\rho}{\bar{\Delta}_i(\gamma^o;\rho,K)+h_i(\gamma^o;\rho,K)}\right \rbrace.
\end{align}

We now have $|g_i(s)-s_{i0}|=|\Delta_i(s)-\int_0^T F_i(z_t(s))dt|$ bounded above by $
 |\Delta_i(s)|+\int_0^T |F_i(z_t(s))|dt\leq T\bar{\Delta}_i(\gamma^o;\rho,K)+ T h_i(\gamma^o;\rho,K)< \rho$,
where we have used that $T<T(\gamma^o;\rho,K)\leq \frac{\rho}{\bar{\Delta}_i(\gamma^o;\rho,K)+h_i(\gamma^o;\rho,K)}$ by construction. Hence, $|g_i(s)-s_{i0}|\leq \rho$ for all $i\in \{1,2,\dots,5\}$, completing the proof.
\end{proof}

\noindent\textbf{Step 4}: \emph{Apply a fixed point theorem to $g$ to find $s$ such that the solution to IVP-$s$ solves the BVP.} By Lemma \ref{lem:LeadingIntSMC}, we can apply Brouwer's Theorem: there exists $s^*\in \mathcal S_\rho$ such that $s^*=g(s^*)$, and hence the solution to IVP-$s^*$ is a solution to the BVP. That $T(\gamma^o)\in O(1/\gamma^o)$ follows simply from the denominators of the underlying expressions begin proportional to $\gamma^o$. One can further optimize $T(\gamma^o;\rho,K)$ over $(\rho,K)\in \mathbb{R}_{++}^2$ to obtain $T(\gamma^o)$.

\noindent\textbf{Step 5}: \emph{Show that given a solution to the BVP as above, the remaining coefficients are well defined and thus a LME exists.} Since $\gamma>0$ and $\chi<1$ over $[0,T]$ in our solution, the tuple $(v_6,v_8,\beta_1,\beta_2,\beta_3,\gamma,\chi)$ (obtained by reversing the change of variables at the beginning of the proof) solves our original boundary value problem by construction.

To verify that the remaining coefficients are well defined, consider first the ODE for $\alpha$: $\dot\alpha_t=\frac{\alpha_t (\hat u_{\hat a \theta}+\hat u_{\hat a a}\alpha_t) \gamma_t \chi_t}{2 \sigma_X^2 \sigma_Y^2 (1+\hat u_{\hat a \theta} \chi_t)}\left\{(\hat u_{\hat a \theta}+\hat u_{\hat a a}\alpha_t)[2\sigma_X^2\alpha_t-\tilde{v}_{8t}\alpha_t -4 \sigma_Y^2 \tilde{\beta}_{2t} \chi_t]-2\sigma_X^2 \alpha_t^2\right\}$. By continuity of the solution to the BVP, the RHS of the equation above is locally Lipschitz continuous in $\alpha$, uniformly in $t$. Moreover, $\alpha_T=\beta_{1T}\chi_T+\beta_{3T}=\frac{1+\hat u_{\hat a \theta}\chi_T}{2-\hat u_{\hat a a}\chi_T}>0$ by Assumption \ref{assu:static} parts (ii) and (iv) given $u_{a\theta}=u_{a\hat a}=1/2$. By the comparison theorem, $\alpha_t>0$ for all $t\in [0,T]$.

Using the solution to the BVP and the inequalities above, we identify the remaining equilibrium coefficients. We have directly
$v_{2t}=\frac{2\sigma_Y^2\beta_{0t}}{\gamma_t \alpha_t}$, 
$v_{5t}=\frac{\sigma_Y^2[\beta_{1t}(2-\chi_t)-\beta_{3t}-\hat u_{\hat a \theta}]}{\gamma_t \alpha_t}$,
$v_{7t}=-\frac{2\sigma_Y^2(1-2\beta_{3t})}{\gamma_t \alpha_t}$ and 
$v_{9t}=\frac{2\sigma_Y^2\left[\beta_{2t}-\beta_{1t}(1-\chi_t)\right]}{\gamma_t \alpha_t}.$
The last three are well defined as $\alpha,\gamma>0$. In the remaining ODEs---included in the online appendix---$(\beta_0,v_1,v_3)$ is uncoupled from $(v_0,v_4)$. By inspection, the former has (unique) solution $(\beta_0,v_1,v_3)=(0,0,0)$. Hence $v_2=0$, and the solutions for $(v_0,v_4)$ can be obtained directly by integration, given their terminal values. \hfill{$\square$}

\vspace{-0.1in}
\subsection{Existence Proof Sketch for the General Model}\label{appendix_sketch}

In what follows, we refer the reader to the Mathematica file \path{spm.nb}  on our websites. There, we work under the normalization  $\partial^2 U/\partial a^2=\partial^2 \hat U/\partial \hat a^2=-1$, as scaling flow payoffs by a factor does not affect incentives; consequently, $U_{xy}=u_{xy}$, for $x,y\in \{a,\hat a, \theta\}$ in that file.

As in the paper, we first show that the task of finding LMEs can be reduced to a BVP in $(v_6,v_8,\vec{\beta},\gamma,\chi)$. In this BVP, Assumption \ref{assu:static} and the fact that $\chi\in [0,1)$ ensure that the terminal conditions ($(1-u_{a\hat{a}} \hat{u}_{a \hat{a}})$ and $(1-u_{a\hat{a}}\hat{u}_{a\hat{a}}\chi_T)$ in their denominators), as well as the ODEs (with  $(u_{a\theta}+u_{a\hat a}\hat u_{\hat{a}\theta}\chi_t)$  and $(1-u_{a\hat a}\hat u_{a\hat a})$ in their denominators) are all well defined. The change of variables $(\tilde v_6,\tilde v_8,\tilde \beta_2)=(\frac{v_6\gamma}{(1-\chi)^2}, \frac{v_8\gamma }{1-\chi},\frac{\beta_2}{1-\chi})$ takes us to a new well defined BVP consisting of $(\tilde{v}_6,\tilde{v}_8,\beta_1,\tilde\beta_2,\beta_3,\gamma,\chi)$, where $(1-\chi_t)$ is absent from their denominators.

The existence proof for the latter BVP now follows the same Steps 1--4 as in the proof of Theorem \ref{thm:Leading_LME_BVP}. Regarding Step 5: 1) $1-\chi>0$ and $\gamma>0$ allow us to recover $(v_6,v_8,\beta_2)$ from $(\tilde v_6,\tilde v_8,\tilde \beta_2)$; 2) $\alpha_T\neq 0$ and the comparison theorem applied to $\alpha$ and $0$ (in backward form) yield $\alpha\neq 0$, so $(v_2,v_5,v_7,v_9)$ are well defined; 3) $(\beta_0,v_3)$ form a linear system in themselves that does not contain $v_0,v_1$ or $v_4$, so its solution exists and is unique; 4) the ODEs for $v_0$, $v_1$ and $v_4$ are linear in themselves and uncoupled, so they have unique solutions.

\paragraph{Proofs for Section \ref{sec:applications}:} Refer to the online appendix.

\end{appendices}

\vspace{-0.4cm}

\bibliographystyle{ecta}

\bibliography{bib}

\end{document}